\documentclass[runningheads,a4paper]{llncs}

\usepackage{amsmath}
\usepackage{amssymb}
\usepackage{bm}
\usepackage{float}
\usepackage[linesnumbered,ruled]{algorithm2e}
\usepackage{subfigure}
\usepackage{multirow}
\usepackage{diagbox}
\usepackage{cite}
\usepackage{hyperref}
\usepackage{xcolor}
\usepackage{graphicx}
\usepackage{marvosym}
\def\wt{\hbox{\rm{wt}}}

\newtheorem{Model}{Model}

\begin{document}
\title{Correlation Cube Attack Revisited\index{Wang,  Jianhua}\index{Qin, Lu}\index{Wu, Baofeng}
}
\subtitle{Improved Cube Search and Superpoly Recovery Techniques}
%
%

\author{Jianhua Wang\inst{1}\textsuperscript{(\Letter)}\orcidID{0009-0003-8895-676X} \and
Lu Qin\inst{2,3}\textsuperscript{(\Letter)}\orcidID{0009-0009-0806-4838} \and
Baofeng Wu\inst{4,5}\textsuperscript{(\Letter)}\orcidID{0000-0002-6567-9216}}
\authorrunning{Wang, J., et al.}
%
\institute{Key Laboratory of Mathematics Mechanization,  Academy of Mathematics and Systems Science, Chinese Academy of Sciences, Beijing, China\\
\email{wangjianhua@amss.ac.cn} \and
China UnionPay Co., Ltd., Shanghai, China\\
\email{qinlu@unionpay.com}\\
\and
School of electronic information and electrical engineering, Shanghai Jiao Tong University, Shanghai, China \and
 Institute of Information Engineering, Chinese Academy of Sciences, Beijing, China\\
\email{wubaofeng@iie.ac.cn} \and 
School of Cyber Security, University of Chinese Academy of Sciences, Beijing, China
} 
\maketitle              
%
%



\begin{abstract}


In this paper, we improve the cube attack by exploiting low-degree factors of the superpoly w.r.t. certain \textit{"special" } index set of cube (\textit{ISoC}). This can be viewed as a special case of the correlation cube attack proposed at Eurocrypt 2018, but under our framework  more beneficial equations on the key variables can be obtained in the  key-recovery phase. To mount our attack, one has two challenging  problems: (1) effectively recover algebraic normal form of the superpoly and extract out its low-degree factors; and (2) efficiently search a large quantity  of good \textit{ISoC}s. We bring in new techniques to solve both of them.

First, we propose the \textit{variable substitution technique} for middle rounds of a cipher, in which  polynomials on the key variables in the algebraic expressions of internal states are substituted by new variables. This will improve computational complexity of the superpoly recovery and  promise more compact superpolys that can be easily decomposed with respect to the new variables. Second, we propose the \textit{vector numeric mapping technique}, which seeks out a tradeoff between efficiency of the numeric mapping technique (Crypto 2019) and accuracy of the monomial prediction technique (Asiacrypt 2020) in degree evaluation of  superpolys.
Combining with this technique,  a fast pruning method is given and modeled by MILP to filter good \textit{ISoC}s of which the algebraic degree satisfies some fixed threshold. Thanks to  automated MILP solvers, it becomes practical to comprehensively search for good cubes across the entire search space.

To illustrate the power of our techniques, we apply all of them to Trivium stream cipher. As a result, we have recovered the superpolys of three cubes given by Kesarwani et al.  in 2020, only to find they do not have \texttt{zero-sum} property up to 842 rounds as claimed in their paper. To our  knowledge, the previous best practical key recovery attack was on 820-round Trivium with complexity $2^{53.17}$. We put forward 820-, 825- and 830-round practical key-recovery attacks, in which there are $\mathbf{2^{80}\times 87.8\%}$, $\mathbf{2^{80}\times 83\%}$ and $\mathbf{2^{80}\times 65.7\%}$ keys that could be practically recovered, respectively, if we consider $\mathbf{2^{60}}$ as the upper bound for practical computational complexity. 
 Besides, even for computers with computational power not exceeding $\mathbf{2^{52}}$ (resp. $\mathbf{2^{55}}$), we can still recover $\mathbf{58\%}$ (resp. $\mathbf{46.6\%}$) of the keys in the key space for 820 rounds (resp. 830 rounds). 
 Our attacks have led 10 rounds more than the previous best practical attack.  

\keywords{Correlation cube attack \and Variable substitution \and Vector numeric mapping  \and MILP \and Trivium.}

\end{abstract}
  \section{Introduction}
  Cube attack was introduced by Dinur and Shamir \cite{EC:DinSha09} at {Eurocrypt} 2009, which is a chosen plaintext key-recovery attack. In performing such an attack, one would like to express the outputs of a cryptosystem as Boolean functions on the inputs, namely, key bits and plaintext bits (say, IV bits for stream ciphers). By examining the integral properties of the outputs over some cubes, i.e., some indices of plaintext variables, one can obtain  equations for the so-called superpolys over certain key bits of the cipher.  
After the introduction of cube attack, several variants of it were proposed, including cube testers \cite{FSE:ADMS09}, dynamic cube attack \cite{FSE:DinSha11b}, conditional cube attack \cite{EC:HWXWZ17}, division-property-based cube attack \cite{C:TIHM17} and correlation cube attack \cite{EC:LYWL18}. Among these, correlation cube attack was proposed at {Eurocrypt} 2018 by Liu et al.~\cite{EC:LYWL18}. It exploits correlations between  the superpoly $f_I$ of a cube and the so-called \textit{basis} $Q_I$, which is a set of low-degree Boolean functions over key bits such that $f_I$ can be expanded over them in terms of $f_I=\bigoplus_{h\in Q_I}h\cdot q_h$. Then the adversary could utilize the obtained equations regarding $h$ to extract information about the encryption key.
  
  Superpoly recovery has always been the most important step in a cube attack. At the beginning, one can only guess the superpolys by performing experiments, such as linearity tests  \cite{EC:DinSha09} and degree tests  \cite{FSE:FouVan13}. It only became possible to recover the exact expressions of superpolys for some cubes when the division property was introduced to cube attacks.
  
  Division property  was introduced by Todo \cite{EC:Todo15} in 2015, which turned out to be  a generalization of the integral property. The main idea is, according to the parity of $\boldsymbol{x}^{\boldsymbol{u}}$ for all $\boldsymbol{x}$ in a multiset $\mathbb{X}$ is even or unknown, one can divide the set of $\boldsymbol{u}$'s into two parts.
  By applying the division property, Todo \cite{EC:Todo15} improved the integral distinguishers for some specific cryptographic primitives, such as \textsc{Keccak}-$f$~\cite{EC:GJMG}, Serpent~\cite{FSE:ERL} and the Simon family~\cite{ACM:RDJSBL}. Then, the bit-based division property  was proposed in 2016 \cite{FSE:TodMor16}, which aimed at  cryptographic primitives  only performing  bit operations. It was also generalized to the three subsets setting to describe the parity of $\boldsymbol{x}^{\boldsymbol{u}}$ for all $\boldsymbol{x}$ in $\mathbb{X}$ as not only even or unknown but also odd. Since it is more refined than the conventional division property, integral cryptanalysis against the Simon family of block ciphers was further improved. Afterwards, Xiang et al. \cite{AC:XZBL16} firstly transformed the propagation of bit-based division property into a mixed integer linear programming (MILP) model, and since then, one could search integral distinguishers by using off-the-shelf MILP solvers.
  
  At Crypto 2017, cube attack based on the division property  was proposed by Todo et al. \cite{C:TIHM17}. One can evaluate values of the key bits that are not involved in the superpoly of a cube by using the division property. If we already know the superpoly is independent of most key bits, then we can recover the superpoly by trying out all possible combinations of other key variables which may be involved. At {Crypto} 2018, Wang et al.  \cite{C:WHTLIM18} improved the division-property-based cube attack in both complexity and accuracy. They reduced the complexity of recovering a superpoly by evaluating the upper bound of its degree. In addition, they improved the preciseness of the MILP model by using the ``flag" technique so that one could obtain a non-zero superpoly. However, with these techniques, it remains  impossible to recover superpolys with high degrees  or  superpolys for large-size cubes, as the time complexity grows exponentially in both cases.     
  
  Wang et al. \cite{AC:WHGZS19} transformed the problem of superpoly recovery into evaluating the trails of division property with three subsets, and one could recover superpolys practically thanks to a breadth-first search algorithm and the pruning technique. As a result, they successfully recovered the superpolys of large-size \textit{ISoC}s for 839- and 840-round Trivium practically, but only gave a theoretical attack against 841-round Trivium. In \cite{EC:HLMTW20}, Hao et al. pointed out that the pruning technique is not always so efficient. Therefore, instead of a breadth-first search algorithm, they simply utilized an MILP model for three-subset division property without unknown subset. As a result, they successfully recovered the superpoly of 840-, 841- and 842-round Trivium  with the aid of an off-the-shelf MILP solver. At {Asiacrypt} 2020, Hu et al. \cite{AC:HSWW20} introduced the monomial prediction technique to describe the division property and provided deeper insights to understand it. They also established the equivalence between three-subset division property without unknown subsets and  monomial predictions, showing both of them were perfectly accurate. However, the complexity of both techniques are very dependent on the efficiency of the MILP solvers. Once the number of division trails is very large, it is hard to recover superpolys by these two techniques, since the MILP solver may not find all solutions in  an acceptable time. Afterwards, Hu et al.~\cite{AC:HSTWW21} proposed an improved framework called nested monomial prediction to recover massive superpolys. Recently, based on this technique, He et al.~\cite{AC:HHPW22} proposed a new framework which contains two main steps: one is to obtain the so-called valuable terms which contributes to the superpoly in the middle rounds, and the other is to compute the coefficients of these valuable terms. To recover the valuable terms, non-zero bit-based division property (NBDP) and core monomial prediction (CMP) were introduced, which promoted great improvement to the computational complexity of superpoly recovery.
  
  In addition to superpoly recovery, degree evaluation of cryptosystems is also an important issue in  cube attacks, since the algebraic degree is usually used to judge whether the superpoly is zero and to search for good \textit{ISoC}s. In \cite{C:Liu17}, Liu introduced the numeric mapping technique and proposed an algorithm for degree evaluation of nonlinear feedback shift register (NFSR) based cryptosystems, which could give  upper bounds of the degree. This method has low complexity but the estimation is less accurate generally. For example, it performs badly for Trivium-like ciphers when there exist adjacent indices in an \textit{ISoC}. On the other hand, Hu et al.'s 
  monomial prediction technique \cite{AC:HSWW20} can promise accurate degree evaluation,  but the time consumption  is too considerable which limits its application in large-scale search.  An algorithm seeking a trade-off between accuracy and efficiency in degree evaluation has been missing in the literature.

The Trivium cipher \cite{ISC:DeCanniere06}, a notable member of the eSTREAM portfolio, has consistently been a primary target for cube attacks.   Notably, the advances in cube attacks in recent years were significantly propelled by analysis of this cipher \cite{AC:HSTWW21,DBLP:conf/cisc/CheT22,AC:HHPW22}. When it comes to theoretical attacks on 840 rounds of Trivium and beyond, the key challenge is to identify balanced superpolys. These superpolys often encompass millions to billions of terms, generally involving the majority of key bits. Due to the infeasibility of solving these high-degree equations, researchers have resorted to exhaustively enumerating most potential keys. This process simplifies the equations but often only results in the recovery of a handful of key bits. When it comes to practical attacks, we can look at the attacks mentioned in \cite{DBLP:conf/cisc/CheT22}. Here, a thorough search for \textit{ISoC}s with simpler superpolys, such as linear or quadratic polynomials, is necessary. However, as the number of rounds increases, smaller \textit{ISoC}s  increasingly produce complex superpolys, making higher-round attacks infeasible. These complexities in superpolys obstruct effective key recovery attacks, leading us to the question that how can we gain more key information from the equation system to enhance the attack. In this work, we propose methods to address this challenge.

\textbf{Our contributions.} 
  To handle complex superpolys, leveraging the correlation between superpolys and low-degree Boolean functions is a promising approach for key recovery. 
 In this paper, we revisit the correlation cube attack and propose an improvement by utilizing a significant number of so-called ``special" \textit{ISoC}s whose superpolys have low-degree Boolean factors,  improving both the quantity and quality of equations obtained in the online phase. However, this approach introduces two challenges: superpoly recovery and the search for good \textit{ISoC}s.

For superpoly recovery, we  propose a novel and effective variable substitution technique. By introducing new variables to replace complex expressions of key bits and eliminating trails in intermediate states, we achieve a more compact representation of the  superpoly on these new variables, making it easier to factorize. This technique also improves the computational complexity of superpoly recovery, enabling us to effectively identify special \textit{ISoC}s.

To search good \textit{ISoC}s, a common method is to filter \textit{ISoC}s based on a comparison between the estimated algebraic degree and a fixed threshold. We introduce the concept of \emph{vector degree} for a Boolean function, which contains more information than the conventional algebraic degree. We further employ a new technique called ``vector numeric mapping" to depict the propagation of vector degrees in compositions of Boolean functions. As a result, we can iteratively estimate an upper bound for the vector degree of the entire composite function. Our vector numeric mapping technique outperforms Liu's numeric mapping in accuracy.

Furthermore, by studying properties of the vector numeric mapping, we introduce a pruning technique to quickly filter out good \textit{ISoC}s  whose superpolys have estimated degrees  satisfying  a threshold. We also construct an MILP model to describe this process, promissing an efficient automated selection of good \textit{ISoC}s.

 Our techniques are applied to the Trivium stream cipher. Initially, we apply our algorithms to three \textit{ISoC}s proposed in \cite{DBLP:journals/dcc/KesarwaniRSM20}, which were claimed to have \texttt{zero-sum} distinguishers up to $842$ rounds. 
 However, it is verified that these three \textit{ISoC}s do not possess \texttt{zero-sum} properties for certain numbers of rounds. Nevertheless, two of them still exhibit the 841-round \texttt{zero-sum} property, which is the maximum number of rounds discovered so far for Trivium. Leveraging our good \textit{ISoC} search technique and superpoly recovery with variable substitution technique, we mount correlation cube attacks against Trivium with \textbf{820}, \textbf{825} and \textbf{830} rounds, respectively.  
 As a result, there are $\mathbf{2^{80}\times 87.8\%}$, $\mathbf{2^{80}\times 83\%}$ and $\mathbf{2^{80}\times 65.7\%}$ keys that can be practically recovered, respectively, if we consider $\mathbf{2^{60}}$ as the upper bound for practical computational complexity. 
 Besides, even for computers with computational power not exceeding $\mathbf{2^{52}}$, we can still recover $\mathbf{58\%}$ of the keys in the key space for 820 rounds. 
 For computers with computational power not exceeding $\mathbf{2^{55}}$, we can recover $\mathbf{46.6\%}$ of the keys in the key space for 830 rounds.
 Our attacks have achieved a significant improvement compared to  the previous best practical attack \cite{DBLP:conf/cisc/CheT22}, with up to 10 additional rounds recovered. Furthermore, for the first time, the complexity for recovering 830 rounds is less than $2^{75}$, even surpassing the threshold of $2^{60}$.
 Previous results on key recovery attacks against Trivium and our results are compared in Table \ref{Summary}.

  \begin{table}
      \caption{A summary of key-recovery attacks against Trivium}
      \label{Summary}
      \centering  
      \tabcolsep=.15cm
      \begin{tabular}{cccccccc}       
          \hline
           Attack &\# of & \multicolumn{2}{c}{Off-line phase}&On-line &Total &\# of & \multirow{2}{*}{Ref.}\\
           type & Round& size of \textit{ISoC} & \# of \textit{ISoC}s & phase&time & keys &  \\\hline
           \multirow{10}{*}{Practical} & 672 & 12 & 63 & $2^{17}$ & $2^{18.56}$ & $2^{80}$ & \cite{EC:DinSha09}\\
           & 767 & 28-31 & 35 & $2^{45}$ & $2^{45.00}$ & $2^{80}$ & \cite{EC:DinSha09}\\
           & 784 & 30-33 & 42 & $2^{38}$ & $2^{39.00}$& $2^{80}$ & \cite{FSE:FouVan13}\\
           & 805 & 32-38 & 42 & $2^{38}$ & $2^{41.40}$ & $2^{80}$& \cite{AC:YeTia21}\\ 
           & 806 & 34-37 & 29 & $2^{35}$ & $2^{39.88}$& $2^{80}$ & \cite{DBLP:journals/tosc/000421}\\
           & 808 & 39-41 & 37 & $2^{43}$ & $2^{44.58}$& $2^{80}$ & \cite{DBLP:journals/tosc/000421}\\ 
           & 815 & 44-46 & 35 & $2^{45}$ & $2^{47.32}$ & $2^{80}$& \cite{DBLP:conf/cisc/CheT22}\\
           & 820 & 48-51 & 30 & $2^{50}$ & $2^{53.17}$& $2^{80}$ & \cite{DBLP:conf/cisc/CheT22}\\
           & 820 & 38 & $2^{13}$ & $2^{51}$ & $2^{52}$ & $2^{79.2} $& Sect.\ref{subsec.keyrecovery}\\ 
           & 820 & 38 & $2^{13}$ & $2^{51}$ & $2^{60}$ & $2^{79.8} $& Sect.\ref{subsec.keyrecovery}\\
           & 825 & 41 & $2^{12}$ & $2^{53}$ & $2^{54}$& $2^{79.3}$ & Sect.\ref{subsec.keyrecovery}\\
           & 825 & 41 & $2^{12}$ & $2^{53}$ & $2^{60}$& $2^{79.7}$ & Sect.\ref{subsec.keyrecovery}\\
           & 830 & 41 & $2^{13}$ & $2^{54}$ & $2^{55}$& $2^{78.9}$ & Sect.\ref{subsec.keyrecovery}\\
           & 830 & 41 & $2^{13}$ & $2^{54}$ & $2^{60}$& $2^{79.4}$ & Sect.\ref{subsec.keyrecovery}\\\hline
           \multirow{16}{*}{Theoretical} & 799 & 32-37 & 18 & $2^{62}$ & $2^{62.00}$& $2^{80}$ & \cite{FSE:FouVan13} \\
           & 802 & 34-37 & 8 & $2^{72}$ & $2^{72.00}$& $2^{80}$ & \cite{ACISP:YeTia18}\\
           & 805 & 28 & 28 & $2^{73}$ & $2^{73.00}$ & $2^{80}$& \cite{EC:LYWL18}\\
           & 832 & 72 & 1 & $2^{79}$ & $2^{79.01}$& $2^{80}$ & \cite{C:TIHM17}\\
           & 832 & 72 & 1 & $2^{79}$ & $2^{79.01}$& $2^{80}$ & \cite{DBLP:journals/tc/TodoIHM18}\\
           & 832 & 72 & 1 & $2^{79}$ & $2^{79.01}$& $2^{80}$ & \cite{AC:WHGZS19}\\
           & 835 & 72 & 4 & $2^{79}$ & $<2^{79.01}$& $2^{80}$ & \cite{DBLP:journals/iet-ifs/YeT20}\\
           & 835 & 35 & 41 & $2^{75}$ & $2^{75.00}$& $2^{80}$ & \cite{EC:LYWL18}\\
           & 840 & 75 & 3 & $2^{77}$ & $2^{77.32}$ & $2^{80}$& \cite{AC:HSWW20}\\
           & 840 & 78 & 2 & $2^{79}$ & $2^{79.58}$ & $2^{80}$& \cite{EC:HLMTW20}\\
           & 841 & 78 & 2 & $2^{79}$ & $2^{79.58}$ & $2^{80}$& \cite{EC:HLMTW20}\\
           & 841 & 76 & 2 & $2^{78}$ & $2^{78.58}$& $2^{80}$ & \cite{AC:HSWW20}\\
           & 842 & 76 & 2 & $2^{79}$ & $2^{78.58}$& $2^{80}$ & \cite{AC:HSWW20}\\
           & 842 & 78 & 2 & $2^{79}$ & $2^{79.58}$& $2^{80}$ & \cite{JC:HLMTW21}\\
           & 843 & 54-57,76 & 5 & $2^{75}$ & $2^{76.58}$& $2^{80}$ & \cite{AC:HSTWW21}\\
           & 843 & 78 & 2 & $2^{79}$ & $2^{79.58}$ & $2^{80}$& \cite{DBLP:journals/tosc/000421}\\
           & 844 & 54-55 & 2 & $2^{78}$ & $2^{78.00}$& $2^{80}$ & \cite{AC:HSTWW21}\\
           & 845 & 54-55 & 2 & $2^{78}$ & $2^{78.00}$& $2^{80}$ & \cite{AC:HSTWW21}\\
           & 846 & 51-54 & 6 & $2^{51}$ & $2^{79.00}$& $2^{80}$ & \cite{AC:HHPW22}\\
           & 847 & 52-53 & 2 & $2^{52}$ & $2^{79.00}$& $2^{80}$ & \cite{AC:HHPW22}\\
           & 848 & 52 & 1 & $2^{52}$ & $2^{79.00}$& $2^{80}$ & \cite{AC:HHPW22}\\\hline
           
  %
           
      \end{tabular}
  \end{table}
  
 \textbf{Organization.}
  The rest of this paper is organized as follows. In Section \ref{sec.Preliminaries}, we give some preliminaries including some notations and concepts. In Section \ref{sec.Improvements of Correlation Cube Attacks}, we review correlation cube attack and propose  strategies to improve it. In Section \ref{sec.recover superpoly for novel perspective}, we propose the variable substitution technique to improve the superpoly recovery. In Section \ref{sec:Find_cube}, we introduce the definition of vector degree for any Boolean function and present an improved technique for degree evaluation. Then we introduce an \textit{ISoC} search method. In Section \ref{apptrivium}, we apply our techniques to Trivium. Conclusions are given in Section \ref{sec.conclusion}.

\section{Preliminaries}
\label{sec.Preliminaries}

\subsection{Notations}
Let $\boldsymbol{v} = (v_0, \cdots, v_{n-1})$ be an $n$-dimensional vector. For any $\boldsymbol{v},\boldsymbol{u}\in \mathbb{F}_2^n$, denote $\prod_{i=0}^{n-1} v_i^{u_i}$ by $\boldsymbol{v}^{\boldsymbol{u}}$ or $\pi_{\boldsymbol{u}}(\boldsymbol{v})$, and define an order $\boldsymbol{v} \preccurlyeq \boldsymbol{u}$  ($\boldsymbol{v} \succcurlyeq \boldsymbol{u}$, resp.), which means $v_i \le u_i$ ($v_i \ge u_i$, resp.) for all $0\leq i\leq n-1$. For any $\boldsymbol{u}_0, \cdots, \boldsymbol{u}_{m-1} \in \mathbb{F}_2^n$, we use $\bm{u}=\bigvee_{i=0}^{m-1} \boldsymbol{u}_i \in \mathbb{F}_2^n$ to represent the bitwise logical OR operation, that is, for $0\leq j\leq n-1$, $u_j = 1$ if and only if there exists an $\boldsymbol{u}_i$ whose $j$-th bit equal to 1.  Use $\boldsymbol{1}$ and $\boldsymbol{0}$ to represent the all-one and all-zero vector, respectively.

For a set $I$, denote its cardinality by $|I|$.
For $I \subset [n]= \{0, 1, \cdots, n - 1\}$, let $I^c$ be its complement.
For an $n$-dimensional vector $\boldsymbol{x}$, let $\boldsymbol{x}_I$ represent the $|I|$-dimensional  vector $(x_{i_0}, \cdots, x_{i_{|I|-1}})$ for $I = \{i_0,\cdots, i_{|I|-1}\}$. Note that we always list the elements of $I$ in an increasing order to eliminate ambiguity. 

In this paper, we always distinguish $j \in \mathbb{Z}_{2^d}$ with a $d$-bit vector $\boldsymbol{u}$ in the sense that $\sum_{k=0}^{d-1}u_k 2^{k} = j$.


\subsection{Algebraic Normal Form and Algebraic Degree of Boolean Functions}
An $n$-variable Boolean function $f$ can be uniquely written in the form
$f(\boldsymbol{x})=\bigoplus _{\boldsymbol{u} \in \mathbb{F}_2^{n}} a_{\boldsymbol{u}}\boldsymbol{x}^{\boldsymbol{u}},$ 
which is called the algebraic normal form (ANF) of $f$. 
If the term $\boldsymbol{x}^{\boldsymbol{u}}$ appears in $f$, i.e., $a_{\boldsymbol{u}} = 1$, we denote  $\boldsymbol{x}^{\boldsymbol{u}} \rightarrow f$. Otherwise, denote  $\boldsymbol{x}^{\boldsymbol{u}} \nrightarrow f$.\par

For an index set $I\subset [n]$ with size  $d$, if $\boldsymbol{x}_I$ are considered as variables and $\boldsymbol{x}_{I^c}$ are considered as parameters in $f$, we can write the ANF of $f$ w.r.t. $\boldsymbol{x}_I$ as 
\begin{equation*}
f(\boldsymbol{x})= \bigoplus _{\boldsymbol{v}\in \mathbb{F}_2^{d}}g_{\boldsymbol{v}}(\boldsymbol{x}_{I^c})\boldsymbol{x}_I^{\boldsymbol{v}},
\end{equation*}
where $g_{\boldsymbol{v}}(\boldsymbol{x}_{I^c})=\bigoplus _{\{ \boldsymbol{u}\in\mathbb{F}_2^{n}\mid  \boldsymbol{u}_I = \boldsymbol{v}\}} a_{\boldsymbol{u}}\boldsymbol{x}_{I^c}^{\boldsymbol{u}_{I^c}}$.

The algebraic degree of $f$ w.r.t. $\boldsymbol{x}_I$ is defined as 
$$\deg(f)_{\boldsymbol{x}_I} = \max_{\boldsymbol{v} \in \mathbb{F}_2^d} \{\wt(\boldsymbol{v}) \mid g_{\boldsymbol{v}}(\boldsymbol{x}_{I^c}) \neq 0\},$$
where $\wt(\boldsymbol{v})$ is the Hamming weight of $\boldsymbol{v}$.

\subsection{Cube Attack}
The cube attack was proposed by Dinur and Shamir in \cite{EC:DinSha09}, which is essentially an extension of the higher-order differential attack. Given a Boolean function $f$ whose inputs are $\boldsymbol{x}\in \mathbb{F}_2^n$ and $\boldsymbol{k}\in \mathbb{F}_2^m$, and given a subset $I=\{i_0,\cdots,i_{d-1}\}\subset [n]$, we can write $f$ as
 $$f(\boldsymbol{x}, \boldsymbol{k})=f_I(\boldsymbol{x}_{I^c}, \boldsymbol{k})\cdot \boldsymbol{x}_I^{\boldsymbol{1}}+q_I(\boldsymbol{x}_{I^c}, \boldsymbol{k}),$$
where each term in $q_I$ is not divisible by $\boldsymbol{x}_I^{\boldsymbol{1}}$. Let $C_I$, called a cube (defined by $I$), be the set of vectors $\boldsymbol{x}$ whose components w.r.t. the index set $I$ take all possible $2^d$ values and other components are undetermined. $I$ is called the index set of the cube (\textit{ISoC}). For each $\boldsymbol{y} \in C_I$, there will be a Boolean function with $n-d$ variables derived from $f$. Summing all these $2^d$ derived functions, we have $$\bigoplus _{C_{I}}f(\boldsymbol{x}, \boldsymbol{k})=f_I(\boldsymbol{x}_{I^c}, \boldsymbol{k}).$$
The polynomial $f_I$ is called the superpoly of the cube $C_{I}$ or of the \textit{ISoC} $I$. Actually, $f_I$ is the coefficient of $\boldsymbol{x}_I^{\boldsymbol{1}}$ in the ANF of $f$ w.r.t. $\boldsymbol{x}_I$. If we assign all the values of $\boldsymbol{x}_{I^c}$ to 0, $f_I$ becomes the coefficient of $\boldsymbol{x}^{\boldsymbol{u}}$ in $f$, which is a Boolean function  in $\boldsymbol{k}$, where $u_i=1$ if and only if $i \in I$. We denote it by $\mathtt{Coe}(f, \boldsymbol{x}^{\boldsymbol{u}})$.

\subsection{Correlation Cube Attack}
\label{sec.correlation_cube_attack}
The correlation cube attack was proposed at {Eurocrypt} 2018 by Liu et al.~\cite{EC:LYWL18}. The objective and high-level idea of this attack is to obtain key information by exploiting the correlations between superpolys and their low-degree basis, thereby deriving equations for the basis rather than the superpolys. 

In mathematical terms, for an \textit{ISoC} $I$, denote the basis of a superpoly $f_I$ as $Q_I=\{h_1, \cdots, h_r\}$, such that $h_i$ has low degree  w.r.t. $\boldsymbol{k}$ and  $$f_I(\boldsymbol{x}_{J},\boldsymbol{k})=\bigoplus_{i=1}^rh_iq_i,$$ where $J\subset I^c$.
This attack primarily works in two phases:  
\begin{enumerate}
    \item \textbf{Preprocessing phase} (see Algorithm \ref{alg.preprocessing phase of correlation cube attacks}): In this stage, the adversary tries to obtain a basis  $Q_I$ of the superpoly $f_I$  and add the tuples $(I,h_i,b)$ leading to  $\Pr(h_i = b \mid f_I)$ greater than a threshold  $p$ into $\Omega$, where $\Pr(h_i = b \mid  f_I)$ is the probability of $h_i = 0$ (or $h_i = 1$) given that $f_I$ is zero constant (or not) on $\boldsymbol{x}_J$ for a random fixed key, respectively. \par
    \item \textbf{Online phase} (see Algorithm \ref{alg.online phase of correlation cube attacks}): The adversary randomly chooses $\alpha$ values for  non-cube public bits, and computes corresponding values of the superpoly $f_I$ to check whether it is zero constant or not. If all the values of $f_I$ are zero, for each $(I,h_i,0)$ in  $\Omega$ the equation $h_i=0$ holds with probability greater than $p$. Otherwise, for each $(I,h_i,1)$ in $\Omega$ the equation $h_i=1$ holds with probability greater than $p$. If all the $h_i$'s are balanced and independent with each other, the adversary would recover $r$-bit key information with a probability greater than $p^r$ by solving these $r$ equations. 
\end{enumerate}

This method, though intricate, provides a solution for dealing with high-degree superpolys, and has demonstrated effectiveness   in extending  theoretical attacks  on Trivium to more rounds.

\subsection{Superpoly Recovery With Monomial Prediction/Three-subset Division Property Without Unknown Subset}\label{subsec:monomialpred}
In~\cite{AC:HSWW20}, Hu et al. established the equivalence between monomial prediction and three-subset division property without unknown subset~\cite{EC:HLMTW20}, showing  both  techniques could give  accurate criterion on the existence of a monomial in $f$. Here we take the monomial prediction technique as an example to explain  how to recover a superpoly.

For a vector Boolean function $\boldsymbol{f} = \boldsymbol{f}_{r-1}\circ \cdots \circ \boldsymbol{f}_0$, denote the input and output of $\boldsymbol{f}_i$ by $\boldsymbol{x}_i$ and $\boldsymbol{x}_{i+1}$ respectively. If any $\pi_{\boldsymbol{u}_i}(\boldsymbol{x}_i) \rightarrow \pi_{\boldsymbol{u}_{i+1}}(\boldsymbol{x}_{i+1})$, i.e., the coefficient of $\pi_{\boldsymbol{u}_i}(\boldsymbol{x}_i)$ in $\pi_{\boldsymbol{u}_{i+1}}(\boldsymbol{x}_{i+1})$ is nonzero, then we call
$$\pi_{\boldsymbol{u}_0}(\boldsymbol{x}_0) \rightarrow \pi_{\boldsymbol{u}_1}(\boldsymbol{x}_1) \rightarrow \cdots \rightarrow \pi_{\boldsymbol{u}_{r-1}}(\boldsymbol{x}_{r-1})$$  a monomial trail from $\pi_{\boldsymbol{u}_0}(\boldsymbol{x}_0)$ to $\pi_{\boldsymbol{u}_{r-1}}(\boldsymbol{x}_{r-1})$, denoted by $\pi_{\boldsymbol{u}_0}(\boldsymbol{x}_0)\rightsquigarrow \pi_{\boldsymbol{u}_{r-1}}(\boldsymbol{x}_{r-1})$. If there is no trail from $\pi_{\boldsymbol{u}_0}(\boldsymbol{x}_0)$ to $\pi_{\boldsymbol{u}_{r-1}}(\boldsymbol{x}_{r-1})$, we denote $\pi_{\boldsymbol{u}_0}(\boldsymbol{x}_0)$ $\not\rightsquigarrow \pi_{\boldsymbol{u}_{r-1}}(\boldsymbol{x}_{r-1})$. The set of all  trails from $\pi_{\boldsymbol{u}_0}(\boldsymbol{x}_0)$ to $\pi_{\boldsymbol{u}_{r-1}}(\boldsymbol{x}_{r-1})$ are denoted by $\pi_{\boldsymbol{u}_0}(\boldsymbol{x}_0)\Join\pi_{\boldsymbol{u}_{r-1}}(\boldsymbol{x}_{r-1})$. Obviously, for any $0< i < r-1$, it holds that 
 $$|\pi_{\boldsymbol{u}_0}(\boldsymbol{x}_0)\Join\pi_{\boldsymbol{u}_{r-1}}(\boldsymbol{x}_{r-1})|=  \sum_{\boldsymbol{u}_i}|\pi_{\boldsymbol{u}_0}(\boldsymbol{x}_0)\Join\pi_{\boldsymbol{u}_{i}}(\boldsymbol{x}_{i})|\cdot |\pi_{\boldsymbol{u}_i}(\boldsymbol{x}_i)\Join\pi_{\boldsymbol{u}_{r-1}}(\boldsymbol{x}_{r-1})|.$$

\begin{theorem}[Monomial prediction~\cite{AC:HSWW20,EC:HLMTW20}]
	\label{thm.monomial_prediction} 
  We have $\pi_{\boldsymbol{u}_0}(\boldsymbol{x}_0)\rightarrow\pi_{\boldsymbol{u}_{r-1}}(\boldsymbol{x}_{r-1})$ if and only if $$|\pi_{\boldsymbol{u}_0}(\boldsymbol{x}_0)\Join\pi_{\boldsymbol{u}_{r-1}}(\boldsymbol{x}_{r-1})|\equiv 1 \pmod 2.$$
  That is, $\pi_{\boldsymbol{u}_0}(\boldsymbol{x}_0)\rightarrow\pi_{\boldsymbol{u}_{r-1}}(\boldsymbol{x}_{r-1})$ if and only if, for any $0<i<r-1$,
  $$|\pi_{\boldsymbol{u}_0}(\boldsymbol{x}_0)\Join\pi_{\boldsymbol{u}_{r-1}}(\boldsymbol{x}_{r-1})|\equiv \sum_{\pi_{\boldsymbol{u}_{i}}(\boldsymbol{x}_{i})\rightarrow \pi_{\boldsymbol{u}_{r-1}}(\boldsymbol{x}_{r-1})}|\pi_{\boldsymbol{u}_0}(\boldsymbol{x}_0)\Join\pi_{\boldsymbol{u}_{i}}(\boldsymbol{x}_{i})| \pmod 2.$$
\end{theorem}

\begin{theorem}[{Superpoly recovery}~\cite{AC:HSWW20,EC:HLMTW20}]
	\label{thm.Superpoly_Recovery} 
  Let $f$ be a Boolean function with input $\boldsymbol{x}$ and $\boldsymbol{k}$, and 
   $f= f_{r-1}\circ  \boldsymbol{f}_{r-2}\circ\cdots \circ \boldsymbol{f}_0(\boldsymbol{x}, \boldsymbol{k})$. When setting $\boldsymbol{x}_{I^c}=\boldsymbol{0}$, the superpoly of an \textit{ISoC} $I$ is 
   $$\mathtt{Coe}(f, \boldsymbol{x}^{\boldsymbol{u}}) = \bigoplus_{|\boldsymbol{k}^{\boldsymbol{w}}\boldsymbol{x}^{\boldsymbol{u}}\Join f|\equiv 1 \pmod 2}\boldsymbol{k}^{\boldsymbol{w}},$$
   where $\boldsymbol{u}_I=\boldsymbol{1}$ and $\boldsymbol{u}_{I^c}=\boldsymbol{0}$. 

\end{theorem}

\subsubsection{MILP model for monomial trails.} It is a difficult task to search all the monomial trails manually. Since Xiang et al. \cite{AC:XZBL16} first transformed the propagation of bit-based division property into an MILP model, it only becomes possible to solve such searching problems by using off-the-shelf MILP solvers. To construct an MILP model for the monomial trail of a Boolean function, one needs only to model three basic operations, i.e.,  \texttt{COPY}, \texttt{AND} and \texttt{XOR}. Please refer to Appendix \ref{sec.model} for details.

\subsection{Nested Monomial Prediction With NBDP and CMP Techniques} 
\label{subsec.nestmonomialprediction}
At {Asiacrypt} 2021, Hu et al.\cite{AC:HSTWW21} proposed a framework, called nested monomial prediction, to exactly recover  superpolys. For a Boolean function $f(\boldsymbol{x}, \boldsymbol{k}) = f_{r-1}\circ \boldsymbol{f}_{r-2}\circ\cdots \circ \boldsymbol{f}_0(\boldsymbol{x},\boldsymbol{k})$, denote the input and output of $\boldsymbol{f}_i$ by $\boldsymbol{y}_i$ and $\boldsymbol{y}_{i+1}$ respectively. To compute $\mathtt{Coe}(f, \boldsymbol{x}^{\boldsymbol{u}})$, the process is as follows:
\begin{enumerate}
  \item Set $n = r - 1$, $Y_{n}=\{f\}$ and set a polynomial $p=0$. 
  \item \label{step.iterative} Choose $l$ such that $0 < l < n$ with certain criterion, and set $Y_l=\emptyset$ and $T_l=\emptyset$. 
  \item \label{step.expand}Express each term in $Y_n$ with $\boldsymbol{y}_{l}$ by constructing and solving MILP model of monomial prediction and save the terms $\pi_{\boldsymbol{u}_{l}}(\boldsymbol{y}_{l})$ 
 satisfying  that  the size of $\{\pi_{\boldsymbol{u}_{n}}(\boldsymbol{y}_{n})\in Y_n\mid \pi_{\boldsymbol{u}_{l}}(\boldsymbol{y}_{l})\rightarrow \pi_{\boldsymbol{u}_{n}}(\boldsymbol{y}_{n})\}$ is odd into $T_{l}$. 
  \item \label{step.solve}For each $\pi_{\boldsymbol{u}_{l}}(\boldsymbol{y}_{l})\in T_{l}$, compute $\mathtt{Coe}(\pi_{\boldsymbol{u}_{l}}(\boldsymbol{y}_{l}), \boldsymbol{x}^{\boldsymbol{u}})$ by constructing and solving MILP model of monomial prediction. If the model about $\pi_{\boldsymbol{u}_{l}}(\boldsymbol{y}_{l})$ is successfully solved with acceptable time, update $p$ by $p\oplus\mathtt{Coe}(\pi_{\boldsymbol{u}_{l}}(\boldsymbol{y}_{l}), \boldsymbol{x}^{\boldsymbol{u}})$ and save the unsolved $\pi_{\boldsymbol{u}_{l}}(\boldsymbol{y}_{l})$ into $Y_{l}$.  
  \item If $Y_{l}\neq\emptyset$, set $n=l$ and go to Step 2. Otherwise, return the polynomial $p$.
\end{enumerate}

The idea of Step \ref{step.expand} and Step \ref{step.solve} comes  from Theorem \ref{thm.monomial_prediction} and Theorem \ref{thm.Superpoly_Recovery}, i.e., 
\begin{align}
  \mathtt{Coe}(f, \boldsymbol{x}^{\boldsymbol{u}})&=\bigoplus_{\pi_{\boldsymbol{u}_n}(\boldsymbol{y}_n)\rightarrow f}\mathtt{Coe}(\pi_{\boldsymbol{u}_n}(\boldsymbol{y}_n), \boldsymbol{x}^{\boldsymbol{u}})\label{eq.nested_monomial_prediction}\\
  &=\bigoplus_{\pi_{\boldsymbol{u}_n}(\boldsymbol{y}_n)\rightarrow f}\bigoplus_{\pi_{\boldsymbol{u}_l}(\boldsymbol{y}_l)\rightarrow \boldsymbol{y}_n}\mathtt{Coe}(\pi_{\boldsymbol{u}_l}(\boldsymbol{y}_l), \boldsymbol{x}^{\boldsymbol{u}})\\
  &=\bigoplus_{\pi_{\boldsymbol{u}_l}(\boldsymbol{y}_l)\in T_l}\mathtt{Coe}(\pi_{\boldsymbol{u}_l}(\boldsymbol{y}_l), \boldsymbol{x}^{\boldsymbol{u}})\\
  &= p \oplus \left(\bigoplus_{\pi_{\boldsymbol{u}_l}(\boldsymbol{y}_l)\in Y_l}\mathtt{Coe}(\pi_{\boldsymbol{u}_l}(\boldsymbol{y}_l), \boldsymbol{x}^{\boldsymbol{u}})\right).
\end{align}

Since the number of monomial trails grows sharply as the number of rounds of a cipher increases, it becomes infeasible to compute a superpoly for a high number of rounds with nested monomial prediction.
At {Asiacrypt} 2022, He et al. \cite{AC:HHPW22} proposed new techniques to improve the nested monomial prediction. They no longer took  the way of trying to solve out the coefficient of $\boldsymbol{x}^{\boldsymbol{u}}$ in $\pi_{\boldsymbol{u}_{l}}(\boldsymbol{y}_{l})$ at multiple numbers of middle rounds. Instead, for a fixed number of middle round $r_m$, they focused  on recovering  a set of \textit{valuable terms} (see Definition  \ref{def.valuable_terms}), denoted by $\mathtt{VT}_{r_m}$,  and then computing   coefficient of $\boldsymbol{x}^{\boldsymbol{u}}$ in every \textit{valuable term}. They discard the terms $\pi_{\boldsymbol{u}_{r_m}}(\boldsymbol{y}_{r_m})$ satisfying  there exists no $\boldsymbol{k}^{\boldsymbol{w}}$ such that $\boldsymbol{k}^{\boldsymbol{w}}\boldsymbol{x}^{\boldsymbol{u}}\rightsquigarrow \pi_{\boldsymbol{u}_{r_m}}(\boldsymbol{y}_{r_m})$, i.e., $\mathtt{Coe}(\pi_{\boldsymbol{u}_{r_m}}(\boldsymbol{y}_{r_m}), \boldsymbol{x}^{\boldsymbol{u}})=0$ in Eq.  \eqref{eq.nested_monomial_prediction} for $n=r_m$. The framework of this technique is as follows:
\begin{enumerate}
  \item Try to recover $\mathtt{VT}_{r_m}$. If the model is solved within an acceptable time, go to Step 2.
  \item \label{step.solve2}For each term $\pi_{\boldsymbol{u}_{r_m}}(\boldsymbol{y}_{r_m})$ in $\mathtt{VT}_{r_m}$, compute $\mathtt{Coe}(\pi_{\boldsymbol{u}_{r_m}}(\boldsymbol{y}_{r_m}), \boldsymbol{x}^{\boldsymbol{u}})$ and then sum all of them.    
\end{enumerate}
To recover $\mathtt{VT}_{r_m}$, He et al. proposed two techniques: non-zero bit-based division property (NBDP) and core monomial prediction (CMP), which led to great improvement of the complexity of recovering the \textit{valuable terms} compared to nested monomial prediction. For  details, please refer ~\cite{AC:HHPW22}.  

\begin{definition}[{Valuable terms}~\cite{AC:HHPW22}] \label{def.valuable_terms}
  For a Boolean function $f(\boldsymbol{x}, \boldsymbol{k}) = f_{r-1}\circ \boldsymbol{f}_{r-2}\circ\cdots \circ \boldsymbol{f}_0(\boldsymbol{x},\boldsymbol{k})$, denote the input and output of $\boldsymbol{f}_i$ by $\boldsymbol{y}_i$ and $\boldsymbol{y}_{i+1}$, respectively. Given $0 \le r_m < r$, if a term $\pi_{\boldsymbol{u}_{r_m}}(\boldsymbol{y}_{r_m})$ satisfies (1) $\pi_{\boldsymbol{u}_{r_m}}(\boldsymbol{y}_{r_m}) \rightarrow f$ and (2) $\exists\boldsymbol{k}^{\boldsymbol{w}}$ such that $\boldsymbol{k}^{\boldsymbol{w}}\boldsymbol{x}^{\boldsymbol{u}}\rightsquigarrow \pi_{\boldsymbol{u}_{r_m}}(\boldsymbol{y}_{r_m})$, then it is called a \textit{valuable term} of $\mathtt{Coe}(f, \boldsymbol{x}^{\boldsymbol{u}})$ at round $r_m$.
  
\end{definition}

\section{Improvements to Correlation Cube Attack}
\label{sec.Improvements of Correlation Cube Attacks}
 As the number of rounds of a cipher increases, it becomes infeasible to search small-size \textit{ISoC}s with low-degree superpolys. Correlation cube attack~\cite{EC:LYWL18} provides a viable solution to recover keys by using the correlation property between keys and superpolys, allowing for the use of high-degree superpolys. However, the correlation cube attack has not shown significant improvements or practical applications since its introduction. We revisit this attack first and then propose strategies to improve it.
 
 For convenience, we will continue to use the notations from Section \ref{sec.correlation_cube_attack}, where $$f_I(\boldsymbol{x}_{J},\boldsymbol{k})=\bigoplus_{i=1}^rh_iq_i.$$ 
In the online phase of a correlation cube attack, the adversary  computes the values of $f_I$ for all possible values of $\boldsymbol{x}_{J}$. Using these values, the adversary can make guesses about the value of $h_i$ in $Q_I$. The guessing strategy is as follows: for the tuple $(I,h_i,1)$ satisfying $\Pr(h_i = 1 \mid f_I) > p$, if there exists a value of $f_I$ is 1,  guess $h_i = 1$; for the tuple $(I,h_i,0)$ satisfying $\Pr(h_i = 0 \mid f_I) > p$, if $f_I \equiv 0$,  guess $h_i = 0$. Therefore, the adversary can  obtain some low-degree equations over $\boldsymbol{k}$. 

Now we examine  the probability  of one  such equation being correct.  For certain $i$, in the first case, the success  probability is $\Pr(h_i=1\mid f_I\not\equiv 0)$. If   $r>1$, and $f_I=1$, $q_i = 1$ and $\bigoplus_{j\neq i}h_j q_j = 1$ for some value of $\boldsymbol{x}_{I^c}$, then we have $h_{i}=0$. That is, the guess about $h_i$ is incorrect.  In the second case, the success probability is $\Pr(h_i=0\mid f_I\equiv 0)$. If $r>1$ and $f_I\equiv 0$, there still exists the possibility that  $h_{i}=1$ and $q_i \equiv \bigoplus_{j\neq i}h_j q_j$, leading to incorrect guess of $h_i$.

Therefore, since in the case  $r>1$ only probabilistic equations can be obtained,
 we first improve the strategy by constraining $r=1$.  That is, we consider the case 
 $$f_I = hq,$$
 and call the \textit{ISoC} $I$ satisfying this condition a ``special" \textit{ISoC}.  Note that now the success probability becomes 1 for the first case, and  the fail probability for the second case is actually equal to $\Pr(h = 1 , q \equiv 0)$. Considering there are a set of special \textit{ISoC}s $\{I_1, \cdots, I_m\}$ such that $f_{I_i} = hq_i$, we can modify 
the strategy as follows: if $\exists i$ such that $f_{I_i}\not\equiv 0$, guess $h=1$; otherwise, guess $h = 0$. The success probability is still 1 for the first case. The fail probability for the second case is now reduced to $\Pr(h = 1 , q_1\equiv 0,\ldots, q_m\equiv 0)$. 
In summary, we can improve the success probability of the guessing by searching for a large number of special \textit{ISoC}s. 

Based on the above observations, we propose the improved correlation cube attack in Algorithm \ref{alg.preprocessing phase of improved correlation cube attacks} and Algorithm \ref{alg.online phase of improved correlation cube attacks}.
This attack is executed in two phases: 
\begin{enumerate}
    \item \textbf{Preprocessing phase}: 
    \begin{itemize}
        \item[a.] Identify special \textit{ISoC}s.
        \item[b.] For each $h$, gather all the special \textit{ISoC} $I$ for which $h$ is a factor of $f_I$ into a set $T_h$.
        \item[c.] To reduce the number of equations derived from wrong  guesses of $h$, for those $h$ whose success probability in the second case is at or below a threshold $p$, they will be exclusively guessed in the first case. Their associated $T_h$ are then added to a set $\mathcal{T}_1$.
        \item[d.] The remaining $h$ will be guessed in both cases with their associated $T_h$ forming  a set $\mathcal{T}$.
    \end{itemize}
    \item \textbf{Online phase}: 
    \begin{itemize}
        \item[a.] Computes the value of $f_I$  for each \textit{ISoC} $I$.
        \item[b.] For every $T_h$ in $\mathcal{T}$, make a guess on the value of $h$ based on $f_I$'s value for all $I$ in $T_h$.
        \item[c.] If for any $T_h$ in $\mathcal{T}_1$, the values of $f_I$  for all $I$ in $T_h$ satisfy the condition in the first case, then $h = 1$. Otherwise, no guess  is formulated concerning $h$.
        \item[d.] Store the equations  $h = 1$ in to a set $G_1$, while store the other equations into  a set $G_0$. Note that only equations in $G_0$ may be incorrect.
        \item[e.] Using these derived equations along with partial key guesses, we can try to obtain a candidate of the key. If verifications for all partial key guesses do not yield a valid key, it indicates that there exist incorrect equations. In this case, modify some equations from $G_0$   and solve again until a valid key is obtained. Repeat this iteration  until the correct key is ascertained.
    \end{itemize}
\end{enumerate}

A crucial factor for the success of this attack is to acquire a significant number of special \textit{ISoC}s. To achieve this goal, the first step is to search for a large number of good \textit{ISoC}s and recover their corresponding superpolys. Then, low-degree factors of these superpolys need to be computed.

Using degree estimation techniques is one of the common methods for searching cubes. In Section \ref{sec:Find_cube}, we will first introduce a vector numeric mapping technique to improve the accuracy of degree estimation. By combining this attack, we will propose an algorithm for fast search of lots of good \textit{ISoC}s on a large scale. 

To our knowledge, it is difficult to decompose a complicated Boolean polynomial. To solve this problem, we propose a novel and effective technique to recover superpolys in Section \ref{sec.recover superpoly for novel perspective}. Using this technique, not only  the computational complexity for recovering superpolys can be  reduced, making it feasible to recover a large number of superpolys, but also it  allows for obtaining compact superpolys that are easy to decompose.

\begin{algorithm}
  \label{alg.preprocessing phase of improved correlation cube attacks}
\DontPrintSemicolon
  \SetAlgoLined
  Generate a set $\mathcal{I}$ of \textit{ISoC}'s;\\
  $\mathcal{T} = \emptyset$, and $\mathcal{T}_1 =\emptyset$;\\
  \For {each \textit{ISoC} $I$ in $\mathcal{I}$}{
    Recover the superpoly $f_I$;\\ 
    \For {each low-degree factor $h$ of $f_I$} {
      If $T_h\in \mathcal{T}$, set $T_h = T_h \cup \{I\}$; Otherwise, insert $T_h=\{I\}$ into $\mathcal{T}$;  
    }
  }
  \For {$T_h$ in $\mathcal{T}$} {
  Estimate the conditional probability $\Pr(h = 0 \mid f_I = 0 \text{ for } \forall I \in T_h)$; If its value is $ <= p$, insert $T_h$ into $\mathcal{T}_1$ and remove $T_h$ from $\mathcal{T}$. 
  }
  \KwRet{$\mathcal{T}$ and $\mathcal{T}_1$.}
  \caption{Preprocessing Phase of Improved Correlation Cube Attacks}
\end{algorithm}

\begin{algorithm}
  \label{alg.online phase of improved correlation cube attacks}
\DontPrintSemicolon
\SetKwData{Req}{\textbf{Require:}}
  \SetAlgoLined
  \Req $\mathcal{T}$ and $\mathcal{T}_1$;\\
   $\mathcal{I} = \bigcup_{T_h\in \mathcal{T}\cup\mathcal{T}_1}T_h$\\
  $G_0 = \emptyset$ and $G_1 = \emptyset$;\\
  \For {each \textit{ISoC} $I$ in $\mathcal{I}$}{
    Compute the sum of the output function $f$ over all values in the cube $C_I$, i.e., the value of the superpoly $f_I$;\\
  } 
  \For {$T_h$ in $\mathcal{T}$} {
    \eIf {for any $I\in T_h$ the value of $f_I$ is equal to 0 } {
      Set $G_0 = G_0 \cup \{h = 0\}$;}{Set $G_1 = G_1 \cup \{h = 1\}$;}
  }
  \For {$T_h$ in $\mathcal{T}_1$} {
    \If {there exists $I\in T_h$ s.t. the value of $f_I$ is equal to 1 } {Set $G_1 = G_1 \cup \{h = 1\}$;}
  }
  Set $e = 0$; \\
  \label{{step.adjust_g0}}\For {all possible choices of $e$ equations from $G_0$}
  {Reset $h$ = 1 for these $e$ equations, and remain others in $G_0$;\\
  Solve these $|G_0| + |G_1|$ equations and check whether the solutions are correct;}
  If none of the solutions is correct, set $e = e + 1$ and go to  Step 20. 
  \caption{Online Phase of Improved Correlation Cube Attacks}
\end{algorithm}



\section{Recover Superpolys From A Novel Perspective}
\label{sec.recover superpoly for novel perspective}

\subsection{Motivation}
As discussed in Section \ref{sec.Improvements of Correlation Cube Attacks}, we need lots of {special} \textit{ISoC}s to improve the correlation cube attack. On the one hand, it is still difficult to compute the factor of a complicated polynomial effectively with current techniques to our best knowledge. On the other hand, the efficiency of recovering superpolys needs to be improved in order to recover a large number of superpolys within an accetable  time. Therefore, we propose new techniques to address the aforementioned issues. Let $f(\boldsymbol{x}, \boldsymbol{k}) = f_{r-1}\circ \boldsymbol{f}_{r-2}\circ\cdots \circ \boldsymbol{f}_0(\boldsymbol{x},\boldsymbol{k})$ and denote the input and output of $\boldsymbol{f}_i$ by $\boldsymbol{y}_i$ and $\boldsymbol{y}_{i+1}$, respectively. Here we adopt the notations used in the monomial prediction technique (see Section \ref{subsec:monomialpred}).
Since \begin{align*}
  \mathtt{Coe}(f, \boldsymbol{x}^{\boldsymbol{u}})&= \bigoplus_{\pi_{\boldsymbol{u}_{r_m}}(\boldsymbol{y}_{r_m})}\mathtt{Coe}(f, \pi_{\boldsymbol{u}_{r_m}}(\boldsymbol{y}_{r_m}))\mathtt{Coe}(\pi_{\boldsymbol{u}_{r_m}}(\boldsymbol{y}_{r_m}), \boldsymbol{x}^{\boldsymbol{u}})\\
  &=\bigoplus_{\pi_{\boldsymbol{u}_{r_m}}(\boldsymbol{y}_{r_m})\rightarrow f}\mathtt{Coe}(\pi_{\boldsymbol{u}_{r_m}}(\boldsymbol{y}_{r_m}), \boldsymbol{x}^{\boldsymbol{u}})\\
  &=\bigoplus_{\pi_{\boldsymbol{u}_{r_m}}(\boldsymbol{y}_{r_m})\rightarrow f \text{ and }\exists \boldsymbol{w}\  \text{s.t.} \ \boldsymbol{k}^{\boldsymbol{w}} \boldsymbol{x}^{\boldsymbol{u}}\rightsquigarrow \pi_{\boldsymbol{u}_{r_m}}(\boldsymbol{y}_{r_m}) }\mathtt{Coe}(\pi_{\boldsymbol{u}_{r_m}}(\boldsymbol{y}_{r_m}), \boldsymbol{x}^{\boldsymbol{u}}).
\end{align*} 
By  Definition \ref{def.valuable_terms}, the superpoly is equal to $$\mathtt{Coe}(f, \boldsymbol{x}^{\boldsymbol{u}})=\bigoplus_{\pi_{\boldsymbol{u}_{r_m}}(\boldsymbol{y}_{r_m})\in \mathtt{VT}_{r_m}}\mathtt{Coe}(\pi_{\boldsymbol{u}_{r_m}}(\boldsymbol{y}_{r_m}), \boldsymbol{x}^{\boldsymbol{u}}).$$ Therefore, recovering a superpoly requires two steps: obtaining the \textit{valuable terms} 
 $\mathtt{VT}_{r_m}$ and recovering the coefficients $\mathtt{Coe}(\pi_{\boldsymbol{u}_{r_m}}(\boldsymbol{y}_{r_m}), \boldsymbol{x}^{\boldsymbol{u}})$. 
The specific steps are as follows:
\begin{enumerate}
  \item Try to obtain $\mathtt{VT}_{r_m}$. If the model is solved within an acceptable time, go to Step 2.
  \item \label{step.solve3}For each term $\pi_{\boldsymbol{u}_{r_m}}(\boldsymbol{y}_{r_m})$ in $\mathtt{VT}_{r_m}$, compute $\mathtt{Coe}(\pi_{\boldsymbol{u}_{r_m}}(\boldsymbol{y}_{r_m}), \boldsymbol{x}^{\boldsymbol{u}})$ with our new techniques and sum them.    
\end{enumerate}
We will provide a detailed explanation of the procedures for each step.

\subsection{Obtain Valuable Terms}
One important item to note about the widly used MILP solver, the Gurobi optimizer, is that model modifications are done in a {lazy} fashion, meaning that effects of modifications of a model are not seen immediately. 
We can set up an MILP model with callback function indicating whether the optimizer finds a new solution.
Algorithm \ref{alg.obtain_vt} shows the process of how to obtain the $r_m$-round \textit{Valuable Terms}. The main steps are:
\begin{enumerate}
    \item Establish a model $\mathcal{M}$ to search for all trails $\boldsymbol{k}^{\boldsymbol{w}}\boldsymbol{x}^{\boldsymbol{u}} \rightsquigarrow \pi_{\boldsymbol{u}_{r_1}}(\boldsymbol{y}_{r_1})\rightsquigarrow \cdots \rightsquigarrow f$.
    \item Solve the model $\mathcal{M}$. Once a trail is found, go to Step 3. If there is no solution, go to Step 4.
    \item (\texttt{VTCallbackFun}) Determine whether $\pi_{\boldsymbol{u}_{r_m}}(\boldsymbol{y}_{r_m})\rightarrow f$ by the parity of the number of trails  $\pi_{\boldsymbol{u}_{r_m}}(\boldsymbol{y}_{r_m})\rightsquigarrow f$. If $\pi_{\boldsymbol{u}_{r_m}}(\boldsymbol{y}_{r_m})\rightarrow f$, add $\pi_{\boldsymbol{u}_{r_m}}(\boldsymbol{y}_{r_m})$ to the set $\mathtt{VT}_{r_m}$. Remove all trails from $\mathcal{M}$ that satisfy $\boldsymbol{k}^{\boldsymbol{w}}\boldsymbol{x}^{\boldsymbol{u}}\rightsquigarrow \pi_{\boldsymbol{u}_{r_m}}(\boldsymbol{y}_{r_m}) \rightsquigarrow f$. Go to the Step 2.
    \item Return the \textit{Valuable Terms} $\mathtt{VT}_{r_m}$.
\end{enumerate}
Note that for each $\pi_{\boldsymbol{u}_{r_m}}(\boldsymbol{y}_{r_m})$ satisfying $\pi_{\boldsymbol{u}_{r_m}}(\boldsymbol{y}_{r_m})\rightsquigarrow f$, the parity of the number of trails is calculated only once due to the removal of all trails satisfying $\boldsymbol{k}^{\boldsymbol{w}}\boldsymbol{x}^{\boldsymbol{u}}\rightsquigarrow \pi_{\boldsymbol{u}_{r_m}}(\boldsymbol{y}_{r_m}) \rightsquigarrow f$.

He et al.\cite{AC:HHPW22} also applied the same framework, but they used different techniques. By combining their NBDP and DBP techniques, we can further improve the efficiency  of recovering $\mathtt{VT}_{r_m}$. We will show the results of experiments in Section \ref{apptrivium}.

\subsection{Variable Substitution Technique for Coefficient Recovery}
For a Boolean function $f(\boldsymbol{x}, \boldsymbol{k}) = f_{r-1}\circ \boldsymbol{f}_{r-2}\circ\cdots \circ \boldsymbol{f}_0(\boldsymbol{x},\boldsymbol{k})$ whose inputs are $\boldsymbol{x}\in\mathbb{F}_2^n$ and $\boldsymbol{k}\in\mathbb{F}_2^m$,  denote the input and output of $\boldsymbol{f}_i$ by $\boldsymbol{y}_i$ and $\boldsymbol{y}_{i+1}$, respectively. We study about the problem of recovering $\mathtt{Coe}(\pi_{\boldsymbol{u}_{r_m}}(\boldsymbol{y}_{r_m}), \boldsymbol{x}^{\boldsymbol{u}})$ at middle rounds from an algebraic perspective. Let $\overleftarrow{\boldsymbol{f}_{r_m}}$ denote $\boldsymbol{f}_{r_m-1}\circ \cdots \circ \boldsymbol{f}_0$, i.e., $\boldsymbol{y}_{r_m} = \overleftarrow{\boldsymbol{f}_{r_m}}(\boldsymbol{x}, \boldsymbol{k})$. Assume the algebraic normal form of  
$\overleftarrow{\boldsymbol{f}_{r_m}}$ in $\boldsymbol{x}$ is $$\overleftarrow{\boldsymbol{f}_{r_m}}=\bigoplus_{\boldsymbol{v}\in\mathbb{F}_2^n}\boldsymbol{h}_{\boldsymbol{v}}(\boldsymbol{k})\boldsymbol{x}^{\boldsymbol{v}}.$$ 
Then one could get that
$\mathtt{Coe}(\pi_{\boldsymbol{u}_{r_m}}(\boldsymbol{y}_{r_m}), \boldsymbol{x}^{\boldsymbol{u}})$ is an XOR of some products over $\boldsymbol{h}_{\boldsymbol{v}}(\boldsymbol{k})$. Assume that the number of different non-constant ${\boldsymbol{h}}_{\boldsymbol{v}}[j]$'s is $t$ for all $\boldsymbol{v}$ and $j$, where ${\boldsymbol{h}}_{\boldsymbol{v}}[j]$ represents the $j$-th component of ${\boldsymbol{h}}_{\boldsymbol{v}}$. Now we introduce new intermediates denoted by $\boldsymbol{z}$ to substitute these $t$ $\boldsymbol{h}_{\boldsymbol{v}}[j]$'s. Without loss of generality, assume $\boldsymbol{z} = \boldsymbol{d}(\boldsymbol{k})$, where $\boldsymbol{d}[i]$ is equal to a certain non-constant $\boldsymbol{h}_{\boldsymbol{v}}[j]$. From the ANF of $\overleftarrow{\boldsymbol{f}_{r_m}}$, it is natural to derive the vectorial Boolean function $\boldsymbol{g}_{r_m}$ such that $\boldsymbol{y}_{r_m} = \boldsymbol{g}_{r_m}(\boldsymbol{x},\boldsymbol{z})$, whose ANF in $\boldsymbol{x}$ and $\boldsymbol{z}$ can be written as 
$$\boldsymbol{g}_{r_m}[j] = \bigoplus_{\boldsymbol{v}}a_{\boldsymbol{v},j}\boldsymbol{z}^{\boldsymbol{c}_{\boldsymbol{v},j}}\boldsymbol{x}^{\boldsymbol{v}},$$ 
where $\boldsymbol{g}_{r_m}[j]$ represents $j$-th component of $\boldsymbol{g}_{r_m}$, and $a_{\boldsymbol{v},j}\in \mathbb{F}_2$ and $\boldsymbol{c}_{{\boldsymbol{v}},j}\in\mathbb{F}_2^t$ are both determined by $\boldsymbol{v}$ and $j$.\par 

Example \ref{ex.subtech} serves as an illustration of the process of variable substitution. The transition from round 0 to round $r_m$ with $(k_0k_1\oplus k_2k_5\oplus k_9+k_{10})(k_2k_7\oplus k_8)x_0x_2x_3$ will have at least 4 * 2 = 8 monomial trails. But after variable substitution, there remains  only one trail $z_0z_2x_0x_2x_3$, which means we have consolidated 8 monomial trails into a single one.  As the coefficients become more intricate and the number of terms in the product increases, the magnitude of this reduction becomes more pronounced. Additionally, it is evident that this also makes the superpoly more concise. In general, the more compact the superpoly is, the easier it is to factorize.

\begin{example}
  \label{ex.subtech}
  Assume $\boldsymbol{y}_{r_m} = \boldsymbol{g}_{r_m}(\boldsymbol{x},\boldsymbol{k}) = [(k_0k_1\oplus k_2k_5\oplus k_9+k_{10})x_0x_2\oplus (k_3\oplus k_6)x_5, (k_2k_7\oplus k_8)x_3\oplus x_6x_7]$. Through variable substitution, all coefficients within $\boldsymbol{y}_{r_m}$, including $k_0k_1\oplus k_2k_5\oplus k_9+k_{10}$, $k_3\oplus k_6$, and $k_2k_7\oplus k_8$, will be replaced with new variables $z_0$, $z_1$, and $z_2$, respectively. Then $\boldsymbol{y}_{r_m}$ could be rewritten as $\boldsymbol{y}_{r_m} = \boldsymbol{g}_{r_m}(\boldsymbol{x},\boldsymbol{z}) = [z_0x_0x_2\oplus z_1x_5, z_2x_3\oplus x_6x_7]$.
\end{example}

Therefore, we take such a way of substituting variables at the middle round $r_m$ to recover $\mathtt{Coe}(\pi_{\boldsymbol{u}_{r_m}}(\boldsymbol{y}_{r_m}), \boldsymbol{x}^{\boldsymbol{u}})$, and the process is as follows:
\begin{enumerate}
  \item Compute the ANF of $\boldsymbol{y_{r_m}}$ in $\boldsymbol{x}$.
  \item Substitute all different non-constant $\boldsymbol{h}_{\boldsymbol{v}}[j]$ for all $\boldsymbol{v}$ and $j$ by new variables $\boldsymbol{z}$. 
  \item Recover $\mathtt{Coe}(\pi_{\boldsymbol{u}_{r_m}}(\boldsymbol{y}_{r_m}), \boldsymbol{x}^{\boldsymbol{u}})$ in $\boldsymbol{z}$ by monomial prediction.
\end{enumerate}  \par
In fact, to solve $\mathtt{Coe}(\pi_{\boldsymbol{u}_{r_m}}(\boldsymbol{y}_{r_m}), \boldsymbol{x}^{\boldsymbol{u}})$ in $\boldsymbol{z}$ by monomial prediction is equivalent to find all possible monomial trails $\boldsymbol{z}^{\boldsymbol{c}}\boldsymbol{x}^{\boldsymbol{u}} \rightsquigarrow  \pi_{\boldsymbol{u}_{r_m}}(\boldsymbol{y}_{r_m})$ about $\boldsymbol{c}$. We can construct an MILP model to describe all feasible trails.

\subsubsection{Model for recovering $\mathtt{Coe}(\pi_{\boldsymbol{u}_{r_m}}(\boldsymbol{y}_{r_m}), \boldsymbol{x}^{\boldsymbol{u}})$ in $\boldsymbol{z}$. } To describe monomial prediction into an MILP model, we actually  need only to construct an MILP model to describe all the trails for $\boldsymbol{g}_{r_m}$. Since the ANF of $\boldsymbol{g}_{r_m}$ is known,   three consecutive operations $\mathtt{Copy}\rightarrow\mathtt{And}\rightarrow\mathtt{XOR}$ are sufficient to describe $\boldsymbol{g}_{r_m}$. The process is as follows:
\begin{itemize}
  \item [\textendash][\texttt{Copy}] For each $x_i$ (resp. $z_i$), the number of copies is equal to the number of monomials divisible by $x_i$ (resp. $z_i$) contained in $\boldsymbol{g}_{r_m}[j]$ for all $j$. 
  \item [\textendash][\texttt{And}] \ Generate all monomials contained in $\boldsymbol{g}_{r_m}[j]$ for all $j$. 
  \item [\textendash][\texttt{XOR}] \ According to the ANF of each $\boldsymbol{g}_{r_m}[j]$, collect monomials using \texttt{XOR} to form $\boldsymbol{g}_{r_m}[j]$.
\end{itemize}
 We give an example  to show how to describe  $\boldsymbol{g}_{r_m}$ by $\mathtt{Copy}\rightarrow\mathtt{And}\rightarrow\mathtt{XOR}$. The algorithm for recovering  $\mathtt{Coe}(\pi_{\boldsymbol{u}_{r_m}}(\boldsymbol{y}_{r_m}), \boldsymbol{x}^{\boldsymbol{u}})$ can be found  in Algorithm \ref{alg.coefficient_recovery}.
 
\begin{example}
  \label{ex.copyandxor}
  If $\boldsymbol{y}_{r_m} = \boldsymbol{g}_{r_m}(\boldsymbol{x},\boldsymbol{z}) = (x_0x_1x_2\oplus x_0z_0\oplus z_1, x_2\oplus z_0z_1\oplus z_0)$, we can describe $\boldsymbol{g}_{r_m}$ by the following three steps.
  \begin{equation*}
      \begin{aligned}
  &(x_0, x_1, x_2, z_0, z_1) \overset{\textup{\texttt{Copy}}}{\longrightarrow} (x_0, x_0, x_1, x_2, x_2, z_0, z_0, z_0, z_1, z_1) \overset{\textup{\texttt{And}}}{\longrightarrow}\\
   &(x_0x_1x_2, x_0z_0, z_1, x_2, z_0z_1, z_0) \overset{\textup{\texttt{XOR}}}{\longrightarrow} (x_0x_1x_2\oplus x_0z_0\oplus z_1, x_2\oplus z_0z_1\oplus z_0)
      \end{aligned}
  \end{equation*}
\end{example}

\paragraph{Discussion.}\label{para.discussion} We have given a method of describing $\boldsymbol{g}_{r_m}$ into an MILP model, which is easy to understand and implement. In general, there may be other ways to construct the MILP model for a concrete $\boldsymbol{g}_{r_m}$. Of course, different  ways do not affect the correctness of the coefficients recovered. It is difficult to find theoretical methods to illustrate what kind of way of modeling $\boldsymbol{g}_{r_m}$ is easier to solve. In order to verify the improvement of our variable substitution technique over previous methods, we will compare the performance by some experiments. 

\begin{algorithm}
  \label{alg.coefficient_recovery}
  \SetKwData{Req}{\textbf{Require:}}
  \DontPrintSemicolon
  \SetAlgoLined
  \KwIn{$\boldsymbol{u}$, $\boldsymbol{u}_{r_m}$ and the ANF of $\boldsymbol{g}_{r_m}$}
  \KwOut{$q = \mathtt{Coe}(\pi_{\boldsymbol{u}_{r_m}}(\boldsymbol{y}_{r_m}), \boldsymbol{x}^{\boldsymbol{u}})$}
	 Declare an empty MILP model $\mathcal{M}$.
   Let $\mathrm{\mathbf{a}}$ be $n+t$ MILP variables of $\mathcal{M}$ corresponding to the $n+t$ components of $\boldsymbol{x}||\boldsymbol{z}$.\\
	 $\mathcal{M}.con \leftarrow \mathrm{a}_i = u_i$ for all $i\in[n]$.\\
	 Update  $\mathcal{M}$ according to the function $\boldsymbol{g}_{r_m}$ and denote $\mathrm{\mathbf{b}}$ as the output state of $\boldsymbol{g}_{r_m}$.\\
	 $\mathcal{M}.con \leftarrow \mathrm{b}_i = \boldsymbol{u}_{r_m}[i]$ for all $i$.\\
   $\mathcal{M}.optimize()$.\\

	 Prepare a hash table $H$ whose key is $t$-bit string and value is counter.\\
   \For{each feasible solution of $\mathcal{M}$}{
    Let $\boldsymbol{c}$ denote the solution  $(\mathrm{a}_{n}, \cdots, \mathrm{a}_{n+t-1})$.\\
		 $H[\boldsymbol{c}] \leftarrow H[\boldsymbol{c}] + 1$.}
		 Prepare a polynomial $q \leftarrow 0$.\\
		\For{each $\boldsymbol{c}$ satisfying $H[\boldsymbol{c}]$ is odd}
		{$q \leftarrow q \oplus \boldsymbol{z}^{\boldsymbol{c}}$.}
		\caption{Coefficient Recovery with Variable Substitution}
\end{algorithm}

\section{ Improved Method for Searching A Large Scale of Cubes}
\label{sec:Find_cube}
The search of \textit{ISoC}s in cube attacks often involves degree evaluations of cryptosystems. While the numeric mapping technique \cite{C:Liu17} offers lower complexity, it performs not well for Trivium-like ciphers when dealing with sets of adjacent indices. This limitation arises from the repeated accumulation of estimated degrees due to the multiplications of adjacent indices during updates. 
Although the monomial prediction technique \cite{AC:HSWW20} provides exact results, it is time-intensive. Thus, efficiently obtaining the exact degree of a cryptosystem remains a challenge. To efficiently search for promising cubes with adjacent indices on a large scale, we propose a compromise approach for degree evaluation called the ``vector numeric mapping" technique. This technique yields a tighter upper bound than the numeric mapping technique while maintaining lower time complexity than monomial prediction. Additionally, we have developed an efficient algorithm based on an MILP model for large-scale  search of \textit{ISoC}s.

\subsection{The Numeric Mapping}
Let  $\mathbb{B}_n$ be the set consisting of all $n$-variable Boolean functions. The numeric mapping \cite{C:Liu17}, denoted by $\texttt{DEG}$, is defined as 
\begin{align*}
	\texttt{DEG}:\quad\mathbb{B}_n\times\mathbb{Z}^n&\longrightarrow\mathbb{Z}\\
	(f,\boldsymbol{d})&\longmapsto \max_{a_{\boldsymbol{u}}\ne 0}\left\{\sum_{i=0}^{n-1}\boldsymbol{u}[i]\boldsymbol{d}[i]\right\},
\end{align*}
where $a_{\boldsymbol{u}}$ is the  coefficient of the term $x^{\boldsymbol{u}}$ in the ANF of $f$. 

Let $\boldsymbol{g}=(g_1,\ldots,g_n)$ be an $(m,n)$-vectorial Boolean function, i.e. $g_i\in\mathbb{B}_m$, $1\leq i\leq n$.
Then for $f\in\mathbb{B}_n$, the numeric degree of the composite function $h=f \circ \boldsymbol{g}=f(g_1,\ldots,g_n)$, denoted by $\texttt{DEG}(h)$, is defined as $\texttt{DEG}(f, \boldsymbol{d}_{\boldsymbol{g}})$, 
where $\boldsymbol{d}_{\boldsymbol{g}}[i] \ge \deg (g[i])$ for all $0 \le i \le n - 1$.  
The algebraic degree of $h$ is always no greater than $\texttt{DEG}(h)$, therefore, the algebraic degrees of  internal states of an NFSR-based cryptosystem can be estimated iteratively by using the numeric mapping.

\subsection{The Vector Numeric Mapping}


Firstly, we introduce the definition of  vector degree of a Boolean function, from which we will easily understand the motivation of the vector numeric mapping.
For the sake of simplicity, let $\deg(g_1,\ldots, g_n)$ represent $\left(\deg(g_1), \ldots, \deg(g_n) \right)$.

\begin{definition}[Vector Degree]\label{vardegree}
	Let $f$ be an $n$-variable Boolean function represented w.r.t. $\boldsymbol{x}_I$ as
	$$f(\boldsymbol{x}) = \bigoplus_{\boldsymbol{u}\in \mathbb{F}_2^d} g_{{\boldsymbol{u}}}(\boldsymbol{x}_{I^c}) \boldsymbol{x}_I^{\boldsymbol{u}},$$ where  $I   \subset [n]$,  $|I|=d$.
	The vector degree of $f$ w.r.t. $\boldsymbol{x}$ and the index set $I$, denoted by $\mathbf{vdeg}_{[I,\boldsymbol{x}]}$, is defined as 
	\begin{equation*}
		\mathbf{vdeg}_{[I,\boldsymbol{x}]}(f) = \deg(g_{\boldsymbol{u}_0}, g_{\boldsymbol{u}_1}, \ldots, g_{\boldsymbol{u}_{2^d - 1}})_{\boldsymbol{x}_{I^c}} =\left(\deg(g_{\boldsymbol{u}_0})_{\boldsymbol{x}_{I^c}},\ldots, \deg(g_{\boldsymbol{u}_{2^d - 1}})_{\boldsymbol{x}_{I^c}}\right), 
	\end{equation*}
	where $\boldsymbol{u}_j$ satisfies $\sum_{k=0}^{d - 1} {\boldsymbol{u}_j[k]}2^k = j$, $0\leq j\leq 2^d-1$.
	
\end{definition}

When we do not emphasize $I$ and $\boldsymbol{x}$, we abbreviate $\mathbf{vdeg}_{[I,\boldsymbol{x}]}$ as $\mathbf{vdeg}_{I}$ or $\mathbf{vdeg}$. Similarly, for a vectorial Boolean function $\boldsymbol{g}=(g_1,\ldots,g_n)$, we denote the vector degree of $\boldsymbol{g}$ by $\textbf{vdeg}(\boldsymbol{g})=\left(\textbf{vdeg}(g_1), \ldots, \textbf{vdeg}(g_n)\right)$. 

According to Definition \ref{vardegree}, it is straightforward to get an upper bound of the vector degree of $f$, which is shown in Proposition \ref{vecdegupbound}.
\begin{proposition}\label{vecdegupbound} 
	For any $0 \le j < 2^{\left\lvert I \right\rvert}$, $\mathbf{vdeg}_{[I,\boldsymbol{x}]}(f)[j] \le n - \left\lvert I \right\rvert $.
\end{proposition}

Moreover, it is obvious that the vector degree of $f$ contains more information about $f$ than the algebraic degree. We can also derive the algebraic degree of $f$ from its vector degree, that is,
$$\deg(f) = \max_{0 \le j < 2^{\left\lvert I \right\rvert}}\{\mathbf{vdeg}_{I}(f)[j] + \wt(j)\}.$$
Therefore, the upper bound of the algebraic degree can be estimated by the upper bound of the vector degree. 

\begin{corollary}\label{degcor}
	Let $\boldsymbol{v}$ be an upper bound of the vector degree of $f$, i.e., $\mathbf{vdeg}_{[I,\boldsymbol{x}]}(f)\preccurlyeq \boldsymbol{v}$. Then we have
	$$\deg(f) \le \max_{0 \le j < 2^{\left\lvert I \right\rvert}}\left\{ \min\left\{\boldsymbol{v}[j], n - \left\lvert I \right\rvert \right\} + \wt(j)\right\}.$$ 
\end{corollary}

In fact, the algebraic degree of $f$ is the degenerate form of the vector degree of $f$ w.r.t. $I = \emptyset$. Moreover, if $I_1 \subset I_2$, the vector degree of $f$ w.r.t. $I_1$ can be deduced from the vector degree of $f$ w.r.t. $I_2$, that is,
\begin{equation}\label{vdegI1I2}
	\mathbf{vdeg}_{I_1}(f)[j] = \max_{0 \le j'<2^{\lvert I_2 \rvert - \lvert I_1 \rvert}}\left\{\mathbf{vdeg}_{I_2}(f)[j'\cdot 2^{\lvert I_1 \rvert} + j] + \wt(j')\right\}
\end{equation}
for any $0 \le j < 2^{\lvert I_1 \rvert}$.
 
In order to estimate the vector degree of composite functions, we propose the concept of vector numeric mapping. 

\begin{definition}[Vector Numeric Mapping] \label{vectormd}
	Let $d \ge 0$. The vector numeric mapping, denoted by $\mathtt{VDEG}_d$, is defined as
	\begin{align*}
		\mathtt{VDEG}_d:\quad\mathbb{B}_n\times\mathbb{Z}^{n \times 2^d}&\longrightarrow\mathbb{Z}^{2^d}\\
		(f,V)&\longmapsto \boldsymbol{w},
	\end{align*}
where $f=\bigoplus_{\boldsymbol{u}\in \mathbb{F}_2^n }a_{\boldsymbol{u}}\boldsymbol{x}^{\boldsymbol{u}}$ and for any $0\leq j<2^d$,
	\begin{equation*}
		\boldsymbol{w}[j] := \max_{a_{\boldsymbol{u}} \neq 0}\max_{\substack{j_0,\cdots,j_{n - 1} \\ 0 \le j_i \le \boldsymbol{u}[i](2^d-1) \\ j = \bigvee_{i=0}^{n - 1} \boldsymbol{u}[i]j_i}}\left\{\sum_{i=0}^{n - 1} \boldsymbol{u}[i]V[i][j_i]\right\}.
	\end{equation*}

For an $(m,n)$-vectorial Boolean function $\boldsymbol{g}=(g_0,\ldots, g_{n-1})$, we define its vector numeric mapping as $\texttt{VDEG}(\boldsymbol{g},V)=(\texttt{VDEG}(g_0,V),\ldots,\texttt{VDEG}(g_{n-1}, V))$.

\end{definition}

\begin{theorem}
	\label{vdegthm}
	Let $f$ be an $n$-variable Boolean function and $\boldsymbol{g}$ be an $(m,n)$-vectorial Boolean function. Assume $\mathbf{vdeg}_I(g_i) \preccurlyeq \boldsymbol{v}_i$ for all $0 \le i \le n - 1$ w.r.t. an index set $I$. Then each component of the vector degree of $f \circ\boldsymbol{g}$ is less than or equal to the corresponding component of $\mathtt{VDEG}_{I}(f, V)$, where $V=(\boldsymbol{v}_0,\cdots,\boldsymbol{v}_{n - 1})$.
\end{theorem}
The proof of Theorem  \ref{vdegthm} is given in Appendix \ref{proof1}. By Theorem  \ref{vdegthm}, we know that the vector numeric mapping $\mathtt{VDEG}(f,V)$ gives an upper bound of the vector degree of the composite function $f\circ \boldsymbol{g}$ when $V$ is the upper bound of the vector degree of the vectorial Boolean function $\boldsymbol{g}$. 

For a Boolean function $f(\boldsymbol{x}) = f_{r-1}\circ \boldsymbol{f}_{r-2}\circ\cdots \circ \boldsymbol{f}_0(\boldsymbol{x})$, let $I$ be the index set. We denoted the upper bound of the vector degree of $f$ w.r.t. $\boldsymbol{x}$ and $I$ by $$\widehat{\mathbf{vdeg}}_{[I,\boldsymbol{x}]}(f)=\mathtt{VDEG}(f_{r-1},V_{r-2}),$$ where $V_i=\mathtt{VDEG}(\boldsymbol{f}_i,V_{i-1})$, $0<i\leq r-2$, and $V_0=\mathbf{vdeg}_{[I,\boldsymbol{x}]}(\boldsymbol{f}_0)$. 

According to Proposition \ref{vecdegupbound} and Corollary \ref{degcor}, the estimation  of  algebraic degree of $f$ w.r.t. $\boldsymbol{x}$ and $I$, denoted by $\widehat{\mathbf{deg}}_{[I,\boldsymbol{x}]}(f)$, can be derived from $\widehat{\mathbf{vdeg}}_{[I,\boldsymbol{x}]}(f)$. To meet different goals in various scenes, we give the following three modes to get $\widehat{\mathbf{deg}}_{[I,\boldsymbol{x}]}(f)$:\\
\textbf{Mode 1.} $\widehat{\mathbf{deg}}_{[I,\boldsymbol{x}]}(f)=\max_{0\leq j<2^{|I|}}\{\min\{\widehat{\mathbf{vdeg}}_{[I,\boldsymbol{x}]}(f)[j],n-|I|\}+\wt(j)\}.$\\
\textbf{Mode 2.} $\widehat{\mathbf{deg}}_{[I,\boldsymbol{x}]}(f)=\widehat{\mathbf{vdeg}}_{[I,\boldsymbol{x}]}(f)[2^{|I|}-1]+|I|.$\\
\textbf{Mode 3.} $\widehat{\mathbf{deg}}_{[I,\boldsymbol{x}]}(f)=\max_{0\leq j<2^{|I|}}\{\widehat{\mathbf{vdeg}}_{[I,\boldsymbol{x}]}(f)[j]+\wt(j)\}.$

Mode 1 gives the estimated degree that can be totally derived from  previous discussions, which is  most precise. Mode 2 focuses on the value of the last coordinate of $\widehat{\mathbf{vdeg}}_{[I,\boldsymbol{x}]}(f)$, which may tell us whether the algebraic degree can reach the maximum value. Mode 3 gives the estimated degree without revision, which will be used when choosing the index set of the vector degree.

Since the index set $I$ is an important parameter when estimating the vector degree of $f$, we learn about how different choices of the index set influence the estimation of the vector degree. Then, we give the relationship between numeric mapping and vector numeric mapping.

\begin{theorem}
	\label{vdegsubset}
	Let $f\in\mathbb{B}_n$ and $I_1$ and $I_2$ be two index sets with
	$|I_1|=k$, $|I_2|=d$ and 
	 $I_1\subset I_2$. If $V_1 \in \mathbb{Z}^{n\times 2^{k}}$ and $V_2 \in \mathbb{Z}^{n \times 2^{d}}$ satisfy
	\begin{equation} \label{vdegI1I2con1}
		V_1[i][j] \ge \max_{0 \le j'<2^{d-k}}\left\{V_2[i][j'\cdot 2^{k} + j] + wt(j')\right\}
	\end{equation}
	for any $0 \le i \le n - 1$ and $0 \le j<2^{k}$,  then we have
\begin{equation}\label{vdegI1I2conclusion1}
		\mathtt{VDEG}_k(f, V_1)[j] \ge \max_{0 \le j'<2^{d-k}}\left\{\mathtt{VDEG}_{d}(f,V_2)[j'\cdot 2^{k} + j] + wt(j')\right\}
	\end{equation}
	for any $0 \le j < 2^{k}.$
\end{theorem}

The proof of Theorem \ref{vdegsubset} is given in Appendix \ref{proof2}. Let $V_i \succcurlyeq \mathbf{vdeg}_{I_i}(\boldsymbol{g})$ for $i = 1,2$ in Theorem \ref{vdegsubset}, and assume that they satisfy the inequality (\ref{vdegI1I2con1}). Since $\mathtt{VDEG}_d(f,V_2) \succcurlyeq \mathbf{vdeg}_{I_2}(f\circ \boldsymbol{g})$ by Theorem \ref{vdegthm}, we can see that the RHS of (\ref{vdegI1I2conclusion1}) is larger than or equal to $\mathbf{vdeg}_{I_1}(f\circ \boldsymbol{g})[j]$ from (\ref{vdegI1I2}). It implies that the RHS of (\ref{vdegI1I2conclusion1}) gives a tighter upper bound of $\mathbf{vdeg}_{I_1}(f\circ \boldsymbol{g})[j]$ than the LHS of (\ref{vdegI1I2conclusion1}). Moreover, the relation in (\ref{vdegI1I2con1}) would be maintained after iterations of the vector numeric mapping by Theorem \ref{vdegsubset}.

In fact, the numeric mapping is  the degenerate form of the vector numeric mapping in the sense of $d=0$. Therefore, we can assert that  $\deg(\boldsymbol{g}_r \cdots \boldsymbol{g}_1)$ derived from the iterations of the vector numeric mapping $\mathtt{VDEG}(\boldsymbol{g_i},V_i)$ leads to a tighter upper bound than the iterations of the numeric mapping $\mathtt{DEG}(\boldsymbol{g}_i,\boldsymbol{d}_i)$. We gave an example in Appendix \ref{sec.example}.

How to choose a suitable index set of the vector degree? One can consider the index set $I=[m]$, where $m$ is the size of the input of the function $\boldsymbol{g}$. Of course, it is the best set by Theorem \ref{vdegsubset} if we only consider the accuracy of the estimated degree. However, the space and time complexity of the vector numeric mapping is exponential w.r.t. such a set.
Therefore, we should choose the index set of the vector degree carefully. We will put forward some heuristic ideas for the Trivium cipher in Section \ref{apptrivium}. 

\subsection{Algorithm for Searching Good \text{ISoC}s}\label{sec.search_good_ISoCs}
As mentioned in Section \ref{sec.Improvements of Correlation Cube Attacks}, finding a large scale of {special} \textit{ISoC}s is quite important in improving correlation cube attacks. 
Indeed, we observe that if the estimated algebraic degree of $f$ over an \textit{ISoC} exceeds the size of it, the higher the estimated algebraic degree is, the more complex the corresponding superpoly tends to. Therefore, when searching \textit{ISoC}s of a fixed size, imposing the constraint that the estimated algebraic degree of $f$ is below a  threshold may significantly increase the likelihood of obtaining a relatively simple superpoly. Then, we heuristically convert our goal of finding large scale of {special} \textit{ISoC}s to finding large scale of good \textit{ISoC}s whose corresponding estimated algebraic degrees of $f$ are lower than a  threshold $d$.



In the following, we propose an efficient algorithm for searching large scale of such good \textit{ISoC}s. 



\begin{theorem}\label{vecdegunion}
	Let $f(\boldsymbol{x},\boldsymbol{k})=f_{r-1}\circ \boldsymbol{f}_{r-2}\circ \cdots \circ \boldsymbol{f}_0(\boldsymbol{x},\boldsymbol{k})$ be a Boolean function, where $\boldsymbol{x}\in\mathbb{F}_2^n$ represents the initial vector and $\boldsymbol{k}\in\mathbb{F}_2^m$ represents the key. Let $J\subset[n]$ be an index set for vector degree and $I$ and $K$ be two \textit{ISoC}s satisfying $J\subset K \subset I$. Then we have 	$$\widehat{\mathbf{vdeg}}_{[J,\boldsymbol{x}_K]}(f|_{\boldsymbol{x}_{K^c}=0})\preccurlyeq \widehat{\mathbf{vdeg}}_{[J,\boldsymbol{x}_I]}(f|_{\boldsymbol{x}_{I^c}=0}).$$
\end{theorem}

\begin{proof}
	Let $U_0=\mathbf{vdeg}_{[J,\boldsymbol{x}_K]}(\boldsymbol{f}_0|_{\boldsymbol{x}_{K^c}=0})$, $V_0=\mathbf{vdeg}_{[J,\boldsymbol{x}_I]}(\boldsymbol{f}_0|_{\boldsymbol{x}_{I^c}=0})$, and $U_t=\mathtt{VDEG}(\boldsymbol{f}_t,U_{t-1})$, $V_t=\mathtt{VDEG}(\boldsymbol{f}_t,V_{t-1})$ for $1\leq t \leq r-2$. Then $\widehat{\mathbf{vdeg}}_{[J,\boldsymbol{x}_K]}(f|_{\boldsymbol{x}_{K^c}=0})$ $=\mathtt{VDEG}(f,U_{r-2})$, $\widehat{\mathbf{vdeg}}_{[J,\boldsymbol{x}_I]}(f|_{\boldsymbol{x}_{I^c}=0})=\mathtt{VDEG}(f,V_{r-2}).$
	
	It is obvious that the set of monomials in $\boldsymbol{f}_0|_{\boldsymbol{x}_{I^c}=0}$ is a superset of the set of monomials in $\boldsymbol{f}_0|_{\boldsymbol{x}_{K^c}=0}$ since $I^c\subset K^c$. Thus, we can get $U_0\preccurlyeq V_0$ from Definition  \ref{vardegree}. According to Definition  \ref{vectormd}, we can iteratively get  $U_i\preccurlyeq V_i$ for all $1\leq i\leq r-2$, which leads to $\widehat{\mathbf{vdeg}}_{[J,\boldsymbol{x}_K]}(f|_{\boldsymbol{x}_{K^c}=0})\preccurlyeq \widehat{\mathbf{vdeg}}_{[J,\boldsymbol{x}_I]}(f|_{\boldsymbol{x}_{I^c}=0})$.
\end{proof}

\begin{corollary}\label{delall}
	Let $f(\boldsymbol{x},\boldsymbol{k})=f_{r-1}\circ \boldsymbol{f}_{r-2}\circ \cdots \circ \boldsymbol{f}_0(\boldsymbol{x},\boldsymbol{k})$ be a Boolean function. Let $J$ be an  index set of vector degree,  $d>|J|$ be a threshold of algebraic degree,  and  $K$ be an \textit{ISoC} satisfying $J\subset K$. If $\widehat{\mathbf{deg}}_{[J,\boldsymbol{x}_K]}(f|_{\boldsymbol{x}_{K^c}=0})\geq d$, then $\widehat{\mathbf{deg}}_{[J,\boldsymbol{x}_I]}(f|_{\boldsymbol{x}_{I^c}=0})\geq d$ for all \textit{ISoC}s $I$ satisfying $K\subset I$. 
\end{corollary}
Corollary \ref{delall} can be derived from Theorem \ref{vecdegunion} directly. Theorem \ref{vecdegunion} shows a relationship between the estimated vector degrees of $f$ w.r.t. a fixed index set $J$ for two \textit{ISoC}s containing $J$. According to Corollary \ref{delall}, we can delete all the sets $I$ containing an \textit{ISoC} $K$ from the searching space of \textit{ISoC}s if $K$ satisfies $\widehat{\mathbf{deg}}_{[J,\boldsymbol{x}_K]}(f|_{\boldsymbol{x}_{K^c}=0})\geq d$. 
Therefore, in order to delete more ``bad" \textit{ISoC}s from the searching space, we can try to find such an  \textit{ISoC} $K$ as small as possible. 

For a given \textit{ISoC} $I$ satisfying $\widehat{\mathbf{deg}}_{[J,\boldsymbol{x}_{I}]}(f|_{\boldsymbol{x}_{I^c}=0})\geq d$, we can iteratively choose a series of \textit{ISoC}s $I\supsetneqq I_1 \supsetneqq  \cdots \supsetneqq  I_q\supset J$ such that $\widehat{\mathbf{deg}}_{[J,\boldsymbol{x}_{I_i}]}(f|_{\boldsymbol{x}_{I_i^c}=0})\geq d$ for all $1\leq i \leq q$ and $\widehat{\mathbf{deg}}_{[J,\boldsymbol{x}_{I'}]}(f|_{\boldsymbol{x}_{I'^c}=0})<d$ for any $I'\subsetneqq I_q$. Note that this process can terminate with a smallest \textit{ISoC}  $I_q$ from $I$ since $\widehat{\mathbf{deg}}_{[J,\boldsymbol{x}_{J}]}(f|_{\boldsymbol{x}_{J^c}=0})\le|J|<d$.


Next, we give a new algorithm according to previous discussions for searching a large scale of good \textit{ISoC}s.





\subsubsection{Process of searching good \textit{ISoC}s. }  Let $J$ be a given index set,  $\Omega$ be the set of all subsets of $[n]$ containing  $J$ and with size $k$,  $d$ be a threshold of degree, and $a$ be the number of repeating times. The main steps are: 
\begin{enumerate}
    \item Prepare an empty set $\mathcal{I}$.
    \item Select an  element $I$ from $\Omega$ as an \textit{ISoC}.
    \item Estimate the algebraic degree of $f$ w.r.t. the variable $\boldsymbol{x}_{I}$ and the index set $J$, denoted by $d_I$. If $d_I<d$, then add $I$ to $\mathcal{I}$ and go to Step 5; otherwise, set $count = 0$ and go to Step 4.
    \item Set $count = count+1$. Let $I'=I$, randomly remove an element $i\in I'\setminus J$ from $I'$ and let $x_i=0$. Then, estimate the algebraic degree of $f$ w.r.t. the variable $\boldsymbol{x}_{I'}$. If the degree is less than $d$ and $count<a$, continue to execute Step 4; if the degree is less than $d$ and $count\geq a$, go to step 5; if the degree is greater than or equal to $d$, let $I=I'$ and go to Step 3.
    \item Remove all the sets containing $I$ from $\Omega$. If $\Omega \neq \emptyset$, go to Step 2; otherwise, output $\mathcal{I}$.
\end{enumerate}






The output $\mathcal{I}$ is the set of all  good \textit{ISoC}s we want. In the algorithm, Step 4 shows the process of finding a ``bad" \textit{ISoC} as small as possible. Since the index $i$ we choose to remove from $I'$ is random every time, we use a counter to record the number of repeating times and set the number $a$ as an upper bound of it to ensure that the algorithm can continue to run. 

To implement the  algorithm efficiently, we establish an MILP model and use the automated searching tool Gurobi to solve the model, and then we can get a large scale of good \textit{ISoC}s that are needed.

\subsubsection{MILP Model for searching good \textit{ISoC}s.} In order to evaluate the elements of $\Omega$ more clearly, we use linear inequalities over integers to describe $\Omega$. We use a binary variables $b_i $ to express whether to choose $v_i$ as a cube variable, namely, $b_i=1$ iff  $v_i$ is chosen as a cube variable, $0\leq i \leq n-1$. Then the sub-models are established as follows:
\setcounter{Model}{0}
\begin{Model}\label{NumofCube}
	To describe that the size of each element of $\Omega$ is equal to $k$, we use
	$$\sum_{i=0}^{n-1}b_i=k.$$
\end{Model}

\begin{Model}\label{IncludeJ}
	To describe that each element of $\Omega$ includes the set $J$, we use 
	$$b_j=1~\text{for}~\forall j\in J.$$
\end{Model}

\begin{Model}\label{DelI}
	To describe removing all the sets that contain $I$ from $\Omega$, we use 
	$$\sum_{i\in I}b_i<|I|.$$
\end{Model}

Since some \textit{ISoC}s are deleted in Step 5 during the searching process, we need to adjust the MILP model continuously. Thus we can  use the \textit{Callback} function of Gurobi to implement this process. In fact, using \textit{Callback} function to adjust the model will not repeat the test for excluded nodes that do not meet the conditions, and will continue to search for nodes that have not been traversed, so the whole process of adjusting the model will not cause the repetition of the solving process, and will not result in a waste of time.

According to the above descriptions and the MILP model we have already established, we give an algorithm for searching good \textit{ISoC}s. The algorithm includes two parts which are called the main procedure and the callback function, and the complete algorithm is given in Appendix \ref{AlgSearchGoodCube}.

\section{Application To Trivium}
\label{apptrivium}
In this section, we apply all of our techniques to Trivium, including degree estimation, superpoly recovery and improved correlation cube attack. 
We set $r_m = 200$ in the experiment of recovering superpolys below, and expression of the states after $200$-round initialization of Trivium has been computed and rewritten in  new variables as described in Section \ref{sec.recover superpoly for novel perspective}, where the ANF of new variables in the key $\boldsymbol{k}$ is also determined. For details, please visit the git repository \href{https://github.com/faniw3i2nmsro3nfa94n/Results}{https://github.com/faniw3i2nmsro3nfa94n/Results}.
All experiments are completed on a personal computer due to the promotion of the algorithms.

\subsection{Description of Trivium Stream Cipher}
Trivium \cite{ISC:DeCanniere06} consists of three nonlinear feedback shift registers whose size is 93, 84, 111, denoted by $r_0, r_1, r_2$, respectively. 
Their internal states, denoted by $\boldsymbol{s}$ with a size of  $288$, are initialized by loading  80-bit key $k_i$ into $s_i$ and 80-bit IV $x_i$ into $s_{i+93}$, $0\leq i \leq 79$, and other bits are set to 0 except for the last three bits of the third register. 
During the initialization stage, the algorithm would not output any keystream bit until the internal states are updated for 1152 rounds. 
The linear components of the three update functions are denoted by $\ell_1, \ell_2$ and $\ell_3$, respectively, and the update process can be described as
\begin{equation}\label{eq.simupdate}
\begin{aligned}
	s_{n_i}&=s_{n_i-1} \cdot s_{n_i - 2} \oplus \ell_i(\boldsymbol{s})~ \text{ for }~ i = 1,~ 2,~ 3,\\
 \boldsymbol{s} & \leftarrow \left( s_{287}, s_{0}, s_{1}, \cdots, s_{286} \right),
 \end{aligned}
\end{equation}
where $n_1, n_2, n_3$ are equal to $ 92, 176, 287$, respectively. 
Denote $z$ to be the output bit of Trivium. Then the output function is $z=s_{65}\oplus s_{92}\oplus  s_{161}\oplus s_{176}\oplus s_{242} \oplus s_{287}$.

\subsection{Practical Verification for Known Cube Distinguishers}
\label{sec.PracticalVerfiction}
In~\cite{DBLP:journals/dcc/KesarwaniRSM20}, Kesarwani et al. found three \textit{ISoC}s having \texttt{Zero-Sum} properties till 842 initialization rounds of Trivium by cube tester experiments.
The \textit{ISoC}s are listed in Appendix \ref{sec.ISoc}, namely $I_1, I_2,I_3$.
We apply the superpoly recovery algorithm proposed in Section \ref{sec.recover superpoly for novel perspective} to these \textit{ISoC}s. It turns out that the declared  \texttt{Zero-Sum} properties of these \textit{ISoC}s is incorrect, which is due to the randomness of experiments on a small portion of the keys. The correct results are listed in Table \ref{zero-sumtest}, where ``Y" represents the corresponding \textit{ISoC} has \texttt{Zero-Sum} property, while ``N" represents the opposite.
For more details about the superpolys of these \textit{ISoC}s, please refer to our git repository. 
We also give some values of the key for which the value of non-zero superpolys is equal to 1, listed in Appendix \ref{sec.foundsecretkeys}.
\begin{table}[H]
	\caption{Verification of \texttt{Zero-Sum} properties in \cite{DBLP:journals/dcc/KesarwaniRSM20}}
	\label{zero-sumtest}
 \tabcolsep=.2cm
	\centering  
	\begin{tabular}{ccccccc} 
				\hline
		Rounds    & $\leq$ 835   & 836           & 837 -- 839 & 840    & 841          & 842   \\\hline
		$I_1$     & Y          & N             & N          & N   & Y            & N\\ 
		$I_2$     & Y          & N             & N          & N   & N            & N\\ 
		$I_3$     & Y          & Y             & N          & Y   & Y            & N\\ \hline
		
	\end{tabular}
\end{table}


\subsubsection{Comparison of computational complexity for superpoly recovery.} 
For comparison, we recover superpoly of  the \textit{ISoC} $I_2$ for 838 rounds by nested monomial prediction, nested monomial prediction with NBDP and CMP techniques, and nested monomial prediction with our variable substitution technique, respectively, where the number of middle rounds is set to $r_m = 200$ for the last two techniques. As a result, it takes more than one day for superpoly recovery by nested monomial prediction, about 13 minutes by NBDP and CMP techniques, and 15 minutes by our method. It implies that variable substitution technique plays a role as important  as the NBDP and CMP techniques in improving the complexity of superpoly recovery. Further, by combining our methods with NBDP and CMP techniques to obtain \textit{valuable terms}, it takes about 2 minutes to recover this superpoly. Thus, it is the best choice to combine our variable substitution technique with NBDP and CMP in  superploy recovery.      

\subsection{Estimation of Vector Degree of Trivium}
Recall the algorithm proposed by Liu in \cite{C:Liu17} for estimating the degree of Trivium-like ciphers. We replace the numeric mapping with the vector numeric mapping. 
The reason is that vector numeric mapping  can perform well for the \textit{ISoC}s containing adjacent indices but numeric mapping cannot.

The algorithm for estimation of the vector degree of Trivium is detailed in Algorithm \ref{alg.EstimateDegree} and Algorithm \ref{alg.DegreeMul} in Appendix \ref{AlgEstVecDegTri}.
The main idea is the same as Algorithm 2 in \cite{C:Liu17}, but the numeric mapping is replaced.
For the sake of simplicity, we denote   $\mathtt{VDEG}(\prod_{i=1}^k x[i],(\boldsymbol{v}_1,\cdots, \boldsymbol{v}_k))$ as $\mathtt{VDEGM}(\boldsymbol{v}_1, \cdots, \boldsymbol{v}_k)$  in the algorithms.

\subsubsection{Heuristics method for choosing indices of vector degree.}
As we discussed earlier, the size of the index set of vector degree should not be too large, and we usually set the size  less than 13.
How to choose the indices to obtain a good degree evaluation? 
We  give the following two heuristic strategies.
\begin{enumerate}
	\item Check whether there are adjacent elements in the \textit{ISoC} $I$. 
	If yes, add all the adjacent elements into the index set $J$. 
	When the size of the set $J$ exceeds a preset threshold, randomly remove elements from $J$ until its size  is equal to the threshold.
	Otherwise, set $I = I \setminus J$ and execute Strategy 2.
	\item Run Algorithm \ref{alg.EstimateDegree} 
 with the input $(\boldsymbol{s}^0,I_i,\emptyset,R,3)$ for all $i\in I$, where $I_i=\{i\}$. Remove the index  with the largest degree evaluation of the $R$-round output bit from $I$ every time, and add it to $J$ until the size of $J$ is equal to the preset threshold. 
	If there exist multiple choices that have equal degree evaluation,  randomly pick one of them.
\end{enumerate}

After applying the above two strategies, we will get an index set of vector degree.
Since there are two adjacent states multiplied in the trivium update function, the variables with adjacent indices may be multiplied many times.
So in Strategy 1, we choose  adjacent indices in $I$ and add them to the index set of vector degree. In Strategy 2, we compute the degree evaluation of the $R$-round output bit by setting the degree of $x_j$ to be zero for all $j \in I$ except $i$.
Although the exact degree of the output bit is less than or equal to 1, the evaluation is usually much larger than 1. This is because the variable $x_i$ is multiplied by itself many times and the estimated degree is added repeatedly.
So we choose these variables, whose estimated degrees are too large, as the index of vector degree. Once we fix a threshold of the size of the index set of vector degree,  we can obtain the index set by these two strategies.

\subsubsection{Degree of Trivium on all IV bits.}
We have estimated the upper bound of the degree of the output bit on all IV bits for $R$-round Trivium by  Algorithm \ref{alg.EstimateDegree} 
with $mode = 1$.
Every time we set the threshold to be $8$ to obtain the index set of vector degree and run the procedure of degree estimation with the index set. 
We repeat 200 times and choose the minimum value as the upper bound of the output bit's degree.
The results compared with the numeric mapping technique are illustrated in Figure
 \ref{figuredegree}. In our experiments, the upper bound of the output bit's degree reaches the maximum degree 80 till 805 rounds using vector numeric mapping, while till 794 rounds using numeric mapping.
 Besides, the exact degree \cite{ToSC:CXZZ21} exhibits the behavior of a decrease when the number of rounds increases at certain points. The vector numeric mapping can also capture this phenomenon, whereas numerical mapping cannot. This is because the vector numeric mapping can eliminate the repeated degree estimation of variables whose indices are in the index set of vector degree.

\begin{figure}
	\centering
	\includegraphics[width=0.6\textwidth]{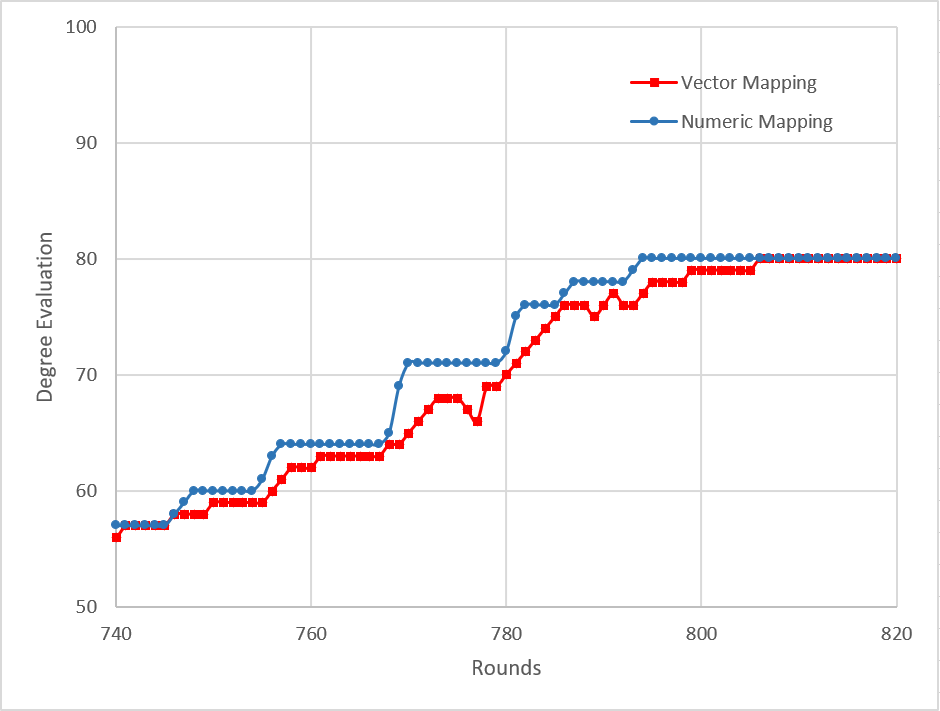}
 \caption{Degree evaluations by vector numeric mapping  and numeric mapping}
	\label{figuredegree}
\end{figure}

\subsubsection{Degree of Trivium on partial IV bits.}
In fact, the degree evaluation algorithm will perform better when there are a few adjacent indices in the \textit{ISoC}. We generate the \textit{ISoC} in the following way. Firstly, randomly generate a set $I_0 \subset [n]$ with size 36 which does not contain  adjacent indices. Next, find a set $I_0\subset I$ with size $36+l$ such that there are exactly $l$ pairs adjacent indices in $I$. Then, one can estimate the degree for the \textit{ISoC} $I$ by numeric mapping technique and vector numeric mapping technique, where the size of the index set of vector degree is set to 8, and calculate the difference of a maximum number of zero-sum rounds between these two techniques. For each $l$, we repeat 200 times and record the average of the differences;  see Table \ref{tab:differenceRounds} for details.

\begin{table}
    \caption{Average improved number of rounds by vector numeric mapping relative to numeric mapping technique}
    \label{tab:differenceRounds}
    \centering
    \footnotesize
    \setlength{\tabcolsep}{7.5pt}
    \renewcommand{\arraystretch}{1.1}
    \begin{tabular}{|c|c|c|c|c|c|c|c|c|c|} 
		\hline
		$l$ &0	&1	&2	&3	&4	&5	&6	&7	&8 \\\hline
        Number &6.8	&27.8	&41.0	&44.7	&45.4	&39.4	&34.6	&31.5	&29.7 \\\hline
    \end{tabular}
\end{table}

It is obvious that when the \textit{ISoC} contains adjacent indices, the vector numeric mapping technique can improve more than 27 rounds compared with the numerical mapping technique on average, even to 45 rounds. 
When there are no adjacent index or few adjacent indices, the difference between degree evaluations by numerical mapping technique and vector mapping technique is small. It implies the reason for the success of degree evaluation for  cubes with no adjacent index by numeric mapping in \cite{C:Liu17}. As $l$ increases, the improved number of rounds first increases and then slowly decreases. This is because the index set of vector degree cannot contain all adjacent indices when $l$ is large. But the vector numeric mapping technique compared with the numeric mapping technique can still improve by about 30 rounds.

\subsubsection{Complexity and precision comparison of degree evaluation.}
In theory, the complexity of degree evaluation using vector numeric mapping technique is no more than $2^{|J|}$ times that of degree evaluation using numeric mapping technique, where $J$ is the index set of vector degree. As evidenced by the experiments conducted above, we have observed that our degree estimation is notably more accurate when the \textit{ISoC} involves only a small adjacent subset. Moreover, since complexity is exponentially related to the size of the  index set of vector degree, we typically limit its size to not exceed 10.

The runtime of our algorithm for 788-round Trivium with various  sizes is detailed in Table \ref{tab:CompexityComparisonofDegEva}. In comparison to degree estimation based on the division property \cite{ToSC:CXZZ21}, the difference between the two methods is not substantial when the \textit{ISoC} consists of only a few adjacent indices. Furthermore, our algorithm significantly outpaces that method, as they require nearly 20 minutes to return degree evaluations for 788 rounds of Trivium.

\subsection{The complexity of fast cube search }
To validate the effectiveness of our pruning technique, we conducted a comparative experiment. As a comparison, we replicated a partial experiment by Liu \cite{C:Liu17}, which involved searching for 837-round distinguishers using cubes of size 37 with non-adjacent indices. As a result, our search algorithm made a total of 9296 calls to the degree estimation algorithm to complete the search of entire space, while exhaustive search required over 38320568 calls to the degree estimation algorithm. This clearly demonstrates the effectiveness of our pruning technique.

\subsection{Practical Key Recovery Attacks }
\label{subsec.keyrecovery}
Benefiting from the new framework of superpoly recovery and the \textit{ISoC} search technique, we could obtain a large scale of {special} \textit{ISoC}s within an acceptable time so that we can mount practical correlation cube attacks against Trivium  with large number of rounds. For correlation cube attacks, we choose the threshold of the conditional probability as $p = 0.77$. We will not elaborate further on these parameters.

\subsubsection{Practical key recovery attacks against 820-round Trivium.}
\paragraph{Parameter settings.}  Set $\Omega$ to be the total space of the \textit{ISoC} with size $k = 38$. Set the index set $J=\{0, 1, 2, i, i+1\}$, the threshold of degree $d$  to be 41 in the \textit{ISoC} search algorithm in Section \ref{sec.search_good_ISoCs}, where $i$ ranges from 3 to 26. We call the search algorithms in parallel for different $i$. 

\paragraph{Attacks.}
We have finally obtained 27428 {special} \textit{ISoC}s with size $38$, whose concrete information can be found in our git repository, including the \textit{ISoC}s, superpolys, factors and balancedness of 
 superpolys, where the balancedness of each superpoly is estimated by randomly testing 10000 keys. Besides, these \textit{ISoC}s are sorted by balancedness of superpolys in descending order. Finally, we choose the first $2^{13}$ \textit{ISoC}s to mount key recovery attacks.
 
For the first $2^{13}$ \textit{ISoC}s, we call
Algorithm \ref{alg.preprocessing phase of improved correlation cube attacks} 
to generate the sets $\mathcal{T}$ and $\mathcal{T}_1$ whose elements are pairs composed of the factor of superpoly and the corresponding special \textit{ISoC}, and  sizes  are 30 and 31, respectively. The results are listed in Appendix \ref{sec.tablefor820Roud}, where the probabilities are estimated by randomly testing 10000 keys. The details about the  \textit{ISoC} corresponding to each factor $h$ are listed in our git repository. 

In the online phase, after computing all the values of the superpolys, one obtain the set of equations $G_0$ and $G_1$. To make full use of the equations, one should recover keys as follows: 
\begin{enumerate}
    \item \label{step.guess} For all $54\le i \le 79$, guess the value of  $k_{i}$ if the equation for $k_{i}$ is not in $G_0 \cup G_1$. 
    \item \label{step.recover} For $i$ from $53$ to $0$, if the equation for $k_{i}+k_{i+25}k_{i+26}+k_{i+27}$ or $k_{i}+k_{i+25}k_{i+26}$ is in $G_0 \cup G_1$, recover the value of $k_i$. Otherwise, guess the value of  $k_{i}$.
    \item Go through over all possible values of $k_i$ guessed in Step 1 and Step 2, and repeat Step 1 until the solution is correct.
    \item If none of the solutions is correct, adjust the equations in $G_0$ according to Step 20 in Algorithm \ref{alg.online phase of improved correlation cube attacks} and go to Step 1.
\end{enumerate}
 Note that the complexity of recovering the value of $k_i$ for $i < 53$ is $\mathcal{O}(1)$, since the values of $k_{i+25}$, $k_{i+26}$ and $k_{i+27}$  are known before. In our experiments, the factors are all chosen in the form $k_{i}+k_{i+25}k_{i+26}+k_{i+27}$ for $0\le i\le 52$ or $k_{53}+k_{78}k_{79}$ or $k_i$ for $54 \le i \le 65$. Thus the
number of key bits obtained by the equations is always equal to the number of equations.

Now we talk about computing the complexity of our improved correlation cube attack. 
Since the set $\mathcal{I}$ of \textit{ISoC}s is fixed, for each fixed key $\boldsymbol{k}$, the corresponding values of the superpolys of all \textit{ISoC}s are determined. 
Therefore, we can calculate the time complexity of recovering this $\boldsymbol{k}$ using the following method. The complexity for computing the values of superpolys remains the same, which is $\mathcal{O}(2^{13}\cdot 2^{38})$. 
For brute force key recovery, the complexity can be determined by combining the values of the superpolys with the guessing strategy, allowing us to obtain the number of equations in $G_0$ and $G_1$, say, $a_{\boldsymbol{k}}$ and $b_{\boldsymbol{k}}$, respectively,  as well as the numbers of incorrect equations in $G_0$, denoted by $e_{\boldsymbol{k}}$. 
It then enables us to determine the complexity of the preprocessing phase to be $2^{80-a_{\boldsymbol{k}}-b_{\boldsymbol{k}}}\cdot\left(\sum_{i=0}^{e_{\boldsymbol{k}}} \binom{a_{\boldsymbol{k}}}{i}\right)$. 
Thus, the complexity for recovering $\boldsymbol{k}$ is $$\mathcal{C}_{\boldsymbol{k}} = \mathcal{O}(2^{13}\cdot 2^{38}) + \mathcal{O}\left(2^{80-a_{\boldsymbol{k}}-b_{\boldsymbol{k}}}\cdot\left(\sum_{i=0}^{e_{\boldsymbol{k}}} \binom{a_{\boldsymbol{k}}}{i}\right)\right).$$
We estimated the proportion of keys with a complexity not larger than $\mathcal{C}$ by randomly selecting 10,000 keys, namely, ${|\{\boldsymbol{k}: \mathcal{C}_{\boldsymbol{k}} \le \mathcal{C}\}|}/{10000}$, and the result is listed in Table \ref{tab:res820}. 
Due to the extensive key space, we have performed a hypothesis testing in  Appendix \ref{sec.hypothesis} to assess whether these proportions can accurately approximate the true proportions. In conclusion, our findings indicate a very strong correlation between them. From Table \ref{tab:res820}, it can be seen that 87.8\% of the keys can be practically recovered by the attack. In particular, 58.0\% of keys can be recovered with a complexity of only $\mathcal{O}(2^{52})$.

\begin{table}
    \caption{The proportion of keys with attack complexities not exceeding  $\mathcal{C}$ for 820 rounds}
    \label{tab:res820}
    \centering
    \footnotesize
    \setlength{\tabcolsep}{6.8pt}
    \renewcommand{\arraystretch}{1.1}
    \begin{tabular}{|c|c|c|c|c|c|} 
		\hline
		$\mathcal{C}$ & $2^{52}$ & $2^{54}$	&$2^{56}$ &$2^{58}$	& $2^{60}$\\\hline
         proportion & 58.0\%	& 69.2\%	&77.0\%	&82.8\%	&87.8\%\\\hline
    \end{tabular}
\end{table}



\subsubsection{Practical key recovery attacks against 825-round Trivium.}
\paragraph{Parameter settings.} Set $\Omega$  to be the total space of \textit{ISoC} with size $41$. Set the index set $J=\{0, 1, \cdots, 10\}\setminus\{j_0, j_1, j_2\}$,  the threshold of degree $d$ to be 44 in the \textit{ISoC} search algorithm in Section \ref{sec.search_good_ISoCs}, where $j_0 > 2$, $j_1 > j_0+1$ and $ j_1 + 1 < j_2 < 11$. We call the search algorithms in parallel for different tuples $(j_0, j_1, j_2)$. 
\paragraph{Attacks.}
We finally obtained 12354 {special} \textit{ISoC}s with size $41$, and we provide their concrete information in our git repository. Besides, these \textit{ISoC}s are sorted by balancedness of superpolys in descending order, where the balancedness is estimated by randomly testing 10000 keys. We choose the first $2^{12}$ \textit{ISoC}s to mount key recovery attacks. 

For the first $2^{12}$ \textit{ISoC}s, we call Algorithm \ref{alg.preprocessing phase of improved correlation cube attacks} 
to generate the sets $\mathcal{T}$ and $\mathcal{T}_1$ whose elements are  pairs composed of the factor of superpoly  and the corresponding special \textit{ISoC}, and the sizes  are 31 and 30, respectively. The results are listed in Appendix \ref{sec.tablefor825Roud}, where the probabilities are estimated by randomly testing 10000 keys. The details about the  \textit{ISoC} corresponding to each factor $h$ are listed in our git repository. 

We estimate the proportion of keys with a complexity not larger than $\mathcal{C}$ by randomly selecting 10,000 keys, and the result is listed in Table \ref{tab:res825}. From Table \ref{tab:res825}, it can be seen that 83\% of the keys can be practically recovered by the attack. In particular, 60.9\% of keys can be recovered with a complexity of only $\mathcal{O}(2^{54})$.

\begin{table}
    \caption{The proportion of keys with attack complexities not exceeding $\mathcal{C}$ for 825 rounds}
    \label{tab:res825}
    \centering
    \footnotesize
    \setlength{\tabcolsep}{6.8pt}
    \renewcommand{\arraystretch}{1.1}
    \begin{tabular}{|c|c|c|c|c|} 
		\hline
		$\mathcal{C}$ &  $2^{54}$	&$2^{56}$ &$2^{58}$	& $2^{60}$\\\hline
         proportion &  60.9\%	&70.7\%	&77.7\%	&83.0\%\\\hline
    \end{tabular}
\end{table}

\subsubsection{Practical key recovery attacks against 830-round Trivium.}
\paragraph{Parameter settings.} The parameter settings are the same as that of 825 rounds, except the threshold of degree $d$ is set to 45 here. We also call the search algorithms in parallel for different tuples $(j_0, j_1, j_2)$. 

\paragraph{Attacks.}
We finally obtained 11099 {special} \textit{ISoC}s with size $41$, whose concrete information can be found in our git repository. Besides these \textit{ISoC}s are sorted by balancedness of superpolys in descending order, where the balancedness is estimated by randomly testing 10000 keys. We choose the first $2^{13}$ \textit{ISoC}s to mount key recovery attacks. 

For the first $2^{13}$ \textit{ISoC}s, we call Algorithm \ref{alg.preprocessing phase of improved correlation cube attacks} 
to generate the sets $\mathcal{T}$ and $\mathcal{T}_1$, with sizes  25 and 41, respectively. The results are listed in Appendix \ref{sec.tablefor830Roud}, where the probabilities are estimated by randomly testing 10000 keys. The details about the  \textit{ISoC} corresponding to each factor $h$ are listed in our git repository. \par

We also estimate the proportion of keys with a complexity not larger than $\mathcal{C}$ by randomly selecting 10000 keys, and the result is listed in Table \ref{tab:res830}. From Table \ref{tab:res830}, it can be seen that 65.7\% of the keys can be practically recovered by the attack. In particular, 46.6\% of keys can be recovered with a complexity of only $\mathcal{O}(2^{55})$.

\begin{table}
    \caption{The proportion of keys with attack complexities not exceeding  $\mathcal{C}$ for 830 rounds}
    \label{tab:res830}
    \centering
    \footnotesize
    \setlength{\tabcolsep}{6.8pt}
    \renewcommand{\arraystretch}{1.1}
    \begin{tabular}{|c|c|c|c|c|c|c|} 
		\hline
		$\mathcal{C}$ &  $2^{55}$	&$2^{56}$ &$2^{57}$ &$2^{58}$ &$2^{59}$	& $2^{60}$\\\hline
         proportion &  46.6\% & 50.6\% &54.2\%	&58.0\%	&61.9\%	&65.7\%\\\hline
    \end{tabular}
\end{table}

Due to limited  computational resources, we were unable to conduct practical validations of key recovery attacks. Instead, we randomly selected some generated superpolys and verified the model's accuracy through cross-validation, utilizing publicly accessible code for superpoly recovery. Furthermore, we have performed practical validations for the non-zero-sum case presented in Table \ref{tab.foundkeys} to corroborate the accuracy of our model. In addition, as mentioned in \cite{DBLP:conf/cisc/CheT22}, attempting to recover keys would take approximately two weeks on a PC equipped with two RTX3090 GPUs when the complexity reaches $\mathcal{O}(2^{53})$. Therefore, for servers with multiple GPUs and nodes, it is feasible to recover a 830-round key within a practical time.

\subsubsection{Discussion about the parameter selections.} Parameter selection is a nuanced process. 
The number of middle rounds $r_m$ is determined by the complexity of computing the expression of $\boldsymbol{g}_{r_m}$. Once $r_m$ exceeds 200, the expression for $\boldsymbol{g}_{r_m}$ becomes intricate and challenging to compute, and overly complex expressions also hinder efficient computation of  $\mathtt{Coe}(\pi_{\boldsymbol{u}_{r_m}}(\boldsymbol{y}_{r_m}), \boldsymbol{x}^{\boldsymbol{u}})$ for MILP solvers. 
For \textit{ISoC}s, we chose their size not exceeding 45 to maintain manageable complexity. We focused on smaller adjacent indices as bases when searching for \textit{ISoC}s. The decision is based on the observation that smaller indices become involved later in the  update process of Trivium. Consequently, this usually results in comparatively simpler superpolys. We directly selected these preset index sets as index of set for vector degree. When determining the threshold for searching good \textit{ISoC}s, we noticed that a higher threshold tended  to result in more complex superpolys. Thus, we typically set the threshold slightly above the size of the \textit{ISoC}s.
In the improved correlation cube attacks, the probability threshold significantly affects the complexity. Too low  threshold will increase the number of incorrect guessed bits $e_{\boldsymbol{k}}$, raising the complexity. Conversely, an excessively high threshold reduces the number of equations in $G_0$, i.e., $a_{\boldsymbol{k}}$, prolonging the brute-force search. One can modify the $p$-value to obtain a relative high success probability.

\subsubsection{Comparison with other attacks.}
From the perspective of key recovery, our correlation cube attack  differs from attacks  in \cite{AC:HSTWW21,DBLP:conf/cisc/CheT22,AC:HHPW22} in how we leverage key information from the superpolys. We obtain equations from the 
superpolys' factors through their correlations with superpolys, whereas \cite{AC:HSTWW21,DBLP:conf/cisc/CheT22,AC:HHPW22} directly utilize the equations of the superpolys. This allows us to extract key information even from high-degree complex superpolys. 
We also expect that this approach will be effective for theoretical attacks and find applications in improving theoretical attacks to more extended rounds.

\section{Conclusions}
\label{sec.conclusion}
In this paper, we propose a variable substitution technique for cube attacks, which makes great improvement to the computational complexity of superpoly recovery and can provide  more concrete superpolys in new variables. To search good cubes, we give a generalized definition of  degree of Boolean function and give out a degree evaluation method with the vector numeric mapping technique. Moreover, we introduce a pruning technique to fast filter the \textit{ISoC}s and describe it into an MILP model to search automatically. It turn out that, these techniques perform well in cube attacks. We also propose  practical verifications for some former work by other authors and perform practical key recovery attacks on $820$-, $825$- and $830$-round Trivium cipher, promoting up to  10 more rounds than previous best practical attacks as we know. In the future study, we will apply our techniques to more ciphers to show their power.

\bibliographystyle{splncs04}
\bibliography{abbrev3,crypto,biblio}

\begin{thebibliography}{10}
\providecommand{\url}[1]{\texttt{#1}}
\providecommand{\urlprefix}{URL }
\providecommand{\doi}[1]{https://doi.org/#1}

\bibitem{FSE:ADMS09}
Aumasson, J.P., Dinur, I., Meier, W., Shamir, A.: Cube testers and key recovery
  attacks on reduced-round {MD6} and {Trivium}. In: Dunkelman, O. (ed.)
  FSE~2009. {LNCS}, vol.~5665, pp. 1--22. Springer, Heidelberg (Feb 2009).
  \doi{10.1007/978-3-642-03317-9\_1}

\bibitem{ACM:RDJSBL}
Beaulieu, R., Shors, D., Smith, J., Treatman-Clark, S., Weeks, B., Wingers, L.:
  The simon and speck lightweight block ciphers. In: Proceedings of the 52nd
  Annual Design Automation Conference. DAC '15, Association for Computing
  Machinery, New York, NY, USA (2015). \doi{10.1145/2744769.2747946}

\bibitem{EC:GJMG}
Bertoni, G., Daemen, J., Peeters, M., Van~Assche, G.: Keccak. In: Johansson,
  T., Nguyen, P.Q. (eds.) Advances in Cryptology -- EUROCRYPT 2013. pp.
  313--314. Springer Berlin Heidelberg, Berlin, Heidelberg (2013).
  \doi{10.1007/978-3-642-38348-9\_19}

\bibitem{FSE:ERL}
Biham, E., Anderson, R., Knudsen, L.: Serpent: A new block cipher proposal. In:
  Vaudenay, S. (ed.) Fast Software Encryption. pp. 222--238. Springer Berlin
  Heidelberg, Berlin, Heidelberg (1998)

\bibitem{DBLP:conf/cisc/CheT22}
Che, C., Tian, T.: An experimentally verified attack on 820-round trivium. In:
  Deng, Y., Yung, M. (eds.) Information Security and Cryptology - 18th
  International Conference, Inscrypt 2022, Beijing, China, December 11-13,
  2022, Revised Selected Papers. Lecture Notes in Computer Science, vol. 13837,
  pp. 357--369. Springer (2022). \doi{10.1007/978-3-031-26553-2\_19}

\bibitem{ToSC:CXZZ21}
Chen, S., Xiang, Z., Zeng, X., Zhang, S.: On the relationships between
  different methods for degree evaluation. {IACR} Trans. Symm. Cryptol.
  \textbf{2021}(1),  411--442 (2021). \doi{10.46586/tosc.v2021.i1.411-442}

\bibitem{ISC:DeCanniere06}
{De Canni{\`e}re}, C.: {Trivium}: A stream cipher construction inspired by
  block cipher design principles. In: Katsikas, S.K., Lopez, J., Backes, M.,
  Gritzalis, S., Preneel, B. (eds.) ISC~2006. {LNCS}, vol.~4176, pp. 171--186.
  Springer, Heidelberg (Aug~/~Sep 2006)

\bibitem{EC:DinSha09}
Dinur, I., Shamir, A.: Cube attacks on tweakable black box polynomials. In:
  Joux, A. (ed.) EUROCRYPT~2009. {LNCS}, vol.~5479, pp. 278--299. Springer,
  Heidelberg (Apr 2009). \doi{10.1007/978-3-642-01001-9\_16}

\bibitem{FSE:DinSha11b}
Dinur, I., Shamir, A.: Breaking {Grain}-128 with dynamic cube attacks. In:
  Joux, A. (ed.) FSE~2011. {LNCS}, vol.~6733, pp. 167--187. Springer,
  Heidelberg (Feb 2011). \doi{10.1007/978-3-642-21702-9\_10}

\bibitem{FSE:FouVan13}
Fouque, P.A., Vannet, T.: Improving key recovery to 784 and 799 rounds of
  {Trivium} using optimized cube attacks. In: Moriai, S. (ed.) FSE~2013.
  {LNCS}, vol.~8424, pp. 502--517. Springer, Heidelberg (Mar 2014).
  \doi{10.1007/978-3-662-43933-3\_26}

\bibitem{EC:HLMTW20}
Hao, Y., Leander, G., Meier, W., Todo, Y., Wang, Q.: Modeling for three-subset
  division property without unknown subset - improved cube attacks against
  {Trivium} and {Grain}-{128AEAD}. In: Canteaut, A., Ishai, Y. (eds.)
  EUROCRYPT~2020, Part~I. {LNCS}, vol. 12105, pp. 466--495. Springer,
  Heidelberg (May 2020). \doi{10.1007/978-3-030-45721-1\_17}

\bibitem{JC:HLMTW21}
Hao, Y., Leander, G., Meier, W., Todo, Y., Wang, Q.: Modeling for three-subset
  division property without unknown subset. Journal of Cryptology
  \textbf{34}(3), ~22 (Jul 2021). \doi{10.1007/s00145-021-09383-2}

\bibitem{AC:HHPW22}
He, J., Hu, K., Preneel, B., Wang, M.: Stretching cube attacks: Improved
  methods to recover massive superpolies. In: Agrawal, S., Lin, D. (eds.)
  Advances in Cryptology -- ASIACRYPT 2022. pp. 537--566. Springer Nature
  Switzerland, Cham (2022). \doi{10.1007/978-3-031-22972-5\_19}

\bibitem{AC:HSTWW21}
Hu, K., Sun, S., Todo, Y., Wang, M., Wang, Q.: Massive superpoly recovery with
  nested monomial predictions. In: Tibouchi, M., Wang, H. (eds.)
  ASIACRYPT~2021, Part~I. {LNCS}, vol. 13090, pp. 392--421. Springer,
  Heidelberg (Dec 2021). \doi{10.1007/978-3-030-92062-3\_14}

\bibitem{AC:HSWW20}
Hu, K., Sun, S., Wang, M., Wang, Q.: An algebraic formulation of the division
  property: Revisiting degree evaluations, cube attacks, and key-independent
  sums. In: Moriai, S., Wang, H. (eds.) ASIACRYPT~2020, Part~I. {LNCS}, vol.
  12491, pp. 446--476. Springer, Heidelberg (Dec 2020).
  \doi{10.1007/978-3-030-64837-4\_15}

\bibitem{EC:HWXWZ17}
Huang, S., Wang, X., Xu, G., Wang, M., Zhao, J.: Conditional cube attack on
  reduced-round {Keccak} sponge function. In: Coron, J.S., Nielsen, J.B. (eds.)
  EUROCRYPT~2017, Part~II. {LNCS}, vol. 10211, pp. 259--288. Springer,
  Heidelberg (Apr~/~May 2017). \doi{10.1007/978-3-319-56614-6\_9}

\bibitem{DBLP:journals/dcc/KesarwaniRSM20}
Kesarwani, A., Roy, D., Sarkar, S., Meier, W.: New cube distinguishers on
  nfsr-based stream ciphers. Des. Codes Cryptogr.  \textbf{88}(1),  173--199
  (2020). \doi{10.1007/s10623-019-00674-1}

\bibitem{C:Liu17}
Liu, M.: Degree evaluation of {NFSR}-based cryptosystems. In: Katz, J.,
  Shacham, H. (eds.) CRYPTO~2017, Part~III. {LNCS}, vol. 10403, pp. 227--249.
  Springer, Heidelberg (Aug 2017). \doi{10.1007/978-3-319-63697-9\_8}

\bibitem{EC:LYWL18}
Liu, M., Yang, J., Wang, W., Lin, D.: Correlation cube attacks: From weak-key
  distinguisher to key recovery. In: Nielsen, J.B., Rijmen, V. (eds.)
  EUROCRYPT~2018, Part~II. {LNCS}, vol. 10821, pp. 715--744. Springer,
  Heidelberg (Apr~/~May 2018). \doi{10.1007/978-3-319-78375-8\_23}

\bibitem{DBLP:journals/tosc/000421}
Sun, Y.: Automatic search of cubes for attacking stream ciphers. {IACR} Trans.
  Symmetric Cryptol.  \textbf{2021}(4),  100--123 (2021).
  \doi{10.46586/tosc.v2021.i4.100-123}

\bibitem{EC:Todo15}
Todo, Y.: Structural evaluation by generalized integral property. In: Oswald,
  E., Fischlin, M. (eds.) EUROCRYPT~2015, Part~I. {LNCS}, vol.~9056, pp.
  287--314. Springer, Heidelberg (Apr 2015).
  \doi{10.1007/978-3-662-46800-5\_12}

\bibitem{C:TIHM17}
Todo, Y., Isobe, T., Hao, Y., Meier, W.: Cube attacks on non-blackbox
  polynomials based on division property. In: Katz, J., Shacham, H. (eds.)
  CRYPTO~2017, Part~III. {LNCS}, vol. 10403, pp. 250--279. Springer, Heidelberg
  (Aug 2017). \doi{10.1007/978-3-319-63697-9\_9}

\bibitem{DBLP:journals/tc/TodoIHM18}
Todo, Y., Isobe, T., Hao, Y., Meier, W.: Cube attacks on non-blackbox
  polynomials based on division property. {IEEE} Trans. Computers
  \textbf{67}(12),  1720--1736 (2018). \doi{10.1109/TC.2018.2835480}

\bibitem{FSE:TodMor16}
Todo, Y., Morii, M.: Bit-based division property and application to simon
  family. In: Peyrin, T. (ed.) FSE~2016. {LNCS}, vol.~9783, pp. 357--377.
  Springer, Heidelberg (Mar 2016). \doi{10.1007/978-3-662-52993-5\_18}

\bibitem{C:WHTLIM18}
Wang, Q., Hao, Y., Todo, Y., Li, C., Isobe, T., Meier, W.: Improved division
  property based cube attacks exploiting algebraic properties of superpoly. In:
  Shacham, H., Boldyreva, A. (eds.) CRYPTO~2018, Part~I. {LNCS}, vol. 10991,
  pp. 275--305. Springer, Heidelberg (Aug 2018).
  \doi{10.1007/978-3-319-96884-1\_10}

\bibitem{AC:WHGZS19}
Wang, S., Hu, B., Guan, J., Zhang, K., Shi, T.: {MILP}-aided method of
  searching division property using three subsets and applications. In:
  Galbraith, S.D., Moriai, S. (eds.) ASIACRYPT~2019, Part~III. {LNCS}, vol.
  11923, pp. 398--427. Springer, Heidelberg (Dec 2019).
  \doi{10.1007/978-3-030-34618-8\_14}

\bibitem{AC:XZBL16}
Xiang, Z., Zhang, W., Bao, Z., Lin, D.: Applying {MILP} method to searching
  integral distinguishers based on division property for 6 lightweight block
  ciphers. In: Cheon, J.H., Takagi, T. (eds.) ASIACRYPT~2016, Part~I. {LNCS},
  vol. 10031, pp. 648--678. Springer, Heidelberg (Dec 2016).
  \doi{10.1007/978-3-662-53887-6\_24}

\bibitem{ACISP:YeTia18}
Ye, C.D., Tian, T.: A new framework for finding nonlinear superpolies in cube
  attacks against {Trivium}-like ciphers. In: Susilo, W., Yang, G. (eds.) ACISP
  18. {LNCS}, vol. 10946, pp. 172--187. Springer, Heidelberg (Jul 2018).
  \doi{10.1007/978-3-319-93638-3\_11}

\bibitem{DBLP:journals/iet-ifs/YeT20}
Ye, C., Tian, T.: Algebraic method to recover superpolies in cube attacks.
  {IET} Inf. Secur.  \textbf{14}(4),  430--441 (2020).
  \doi{10.1049/iet-ifs.2019.0323}

\bibitem{AC:YeTia21}
Ye, C.D., Tian, T.: A practical key-recovery attack on 805-round trivium. In:
  Tibouchi, M., Wang, H. (eds.) ASIACRYPT~2021, Part~I. {LNCS}, vol. 13090, pp.
  187--213. Springer, Heidelberg (Dec 2021). \doi{10.1007/978-3-030-92062-3\_7}

\end{thebibliography}

\appendix

\section{MILP Models for Three Basic Operations}
\label{sec.model}
\setcounter{Model}{0}
\begin{Model}[\texttt{Copy}\cite{AC:HSWW20,EC:HLMTW20}]
	Let $\mathrm{a}\stackrel{\texttt{Copy}}{\longrightarrow} (\mathrm{b_0}, \mathrm{b_1}, \cdots, \mathrm{b_{n-1}})$ be a propagation trail of \texttt{Copy}. The following inequalities are sufficient to describe all  trails for \texttt{Copy}.
	\begin{equation*}
		\begin{cases}
			\mathcal{M}.var \leftarrow \mathrm{a}, \mathrm{b_0}, \mathrm{b_1}, \cdots, \mathrm{b_{n-1}}\text{ as binary};\\
			\mathcal{M}.con \leftarrow \mathrm{b_0} + \mathrm{b_1} + \cdots + \mathrm{b_{n-1}} \ge \mathrm{a}; \\
			\mathcal{M}.con \leftarrow \mathrm{a} \ge \mathrm{b_i} \text{ for all } i\in\{0, 1, \cdots, n-1\}. 
		\end{cases}
	\end{equation*}
\end{Model}

\begin{Model}[\texttt{And}\cite{AC:HSWW20,EC:HLMTW20}]
	Let $(\mathrm{b_0}, \mathrm{b_1}, \cdots, \mathrm{b_{n-1}})\stackrel{And}{\longrightarrow} \mathrm{a} $ be a propagation trail of \texttt{And}. The following equations  are sufficient to describe all  trails for \texttt{And}.
	\begin{equation*}
		\begin{cases}
			\mathcal{M}.var \leftarrow \mathrm{a}, \mathrm{b_0}, \mathrm{b_1}, \cdots, \mathrm{b_{n-1}} \text{ as binary};\\
			\mathcal{M}.con \leftarrow \mathrm{a} = \mathrm{b_i}  \text{ for all } i\in\{0, 1, \cdots, n - 1\}.	
		\end{cases}
	\end{equation*}
\end{Model}

\begin{Model}[\texttt{XOR}\cite{AC:HSWW20,EC:HLMTW20}]
	Let $(\mathrm{b_0}, \mathrm{b_1}, \cdots, \mathrm{b_{n-1}})\stackrel{Xor}{\longrightarrow} \mathrm{a} $ be a propagation trail of \texttt{XOR}. The following equations are sufficient to describe all  trails for \texttt{XOR}.
	\begin{equation*}
		\begin{cases}
			\mathcal{M}.var \leftarrow \mathrm{a}, \mathrm{b_0}, \mathrm{b_1}, \cdots, \mathrm{b_{n-1}} \text{ as binary};\\
			\mathcal{M}.con \leftarrow \sum_{i=0}^{n-1}\mathrm{b_i} = \mathrm{a} 	,
		\end{cases}
	\end{equation*}
	where $\sum$ represents the summation over $\mathbb{Z}$.
\end{Model}


\section{Algorithms for Correlation Cube Attacks}
\begin{algorithm}[H]
  \label{alg.preprocessing phase of correlation cube attacks}
\DontPrintSemicolon
  \SetAlgoLined
  Generate a set of \textit{ISoC} $\mathcal{I}$;\\
  \For {each \textit{ISoC} $I$ in $\mathcal{I}$}{
    $Q_I \gets  \mathtt{Decompostion}(I)$, and goto next $I$ if $Q_I$ is empty; \\
    Estimate the conditional probability $\Pr(h_i = b \mid  f_I)$ for each function $h_i$ in the basis $Q_I$ of the superpoly $f_I$, and select $(I, h_i, b)$ that satisfies $\Pr(h_i = b\mid f_I) > p$.
  }
  \caption{Preprocessing Phase of Correlation Cube Attacks~\cite{EC:LYWL18}}
\end{algorithm}
\begin{algorithm}[H]
  \label{alg.online phase of correlation cube attacks}
\DontPrintSemicolon
\SetKwData{Req}{\textbf{Require:}}
  \SetAlgoLined
  \Req a set of \textit{ISoC} $\mathcal{I}$ and $\Omega = \{(I, h, b)|\Pr(h = b|f_I) > p\}$\\
  $G_0 = \emptyset$ and $G_1 = \emptyset$;\\
  \For {each \textit{ISoC} $I$ in $\mathcal{I}$}{
    Randomly generate $\alpha$ values from free non-cube public bits;\\
    Compute the $\alpha$ values of the superpoly $f_I$ over the cube $C_I$;\\
    \eIf {all the values of $f_I$ equal to 0} {
      Set $G_0 = G_0 \cup \{h = 0|(I,h,0) \in \Omega\}$;}{Set $G_1 = G_1 \cup \{h = 1|(I,h,1) \in \Omega\}$}
    Deal with the case that $\{h|h = 0 \in G_0 \text{ and } h = 1 \in G_1\}$ is not empty;\\
    Randomly choose $r_0$ equations from $G_0$ and $r_1$ equations from $G_1$, solve these $r_0 + r_1$ equations and check whether the solutions are correct;\\
    Repeat Step 12 if none of the solutions is correct.
  }
  \caption{Online Phase of Correlation Cube Attacks~\cite{EC:LYWL18}}
\end{algorithm}

\section{Algorithms for Obtaining \texttt{VT}}
\begin{algorithm}[H]
	\label{alg.obtain_vt}
    \DontPrintSemicolon
    \SetAlgoLined
    \KwIn{$\boldsymbol{u}, r_m$}
	Declare an empty MILP model $\mathcal{M}$.\\
	Let $\boldsymbol{s}_0$ be $n+m$ MILP variables of $\mathcal{M}$ corresponding to the $n+m$ components of $\boldsymbol{x}||\boldsymbol{k}$.\\
	$\mathcal{M}.con \leftarrow \boldsymbol{s}_0[j] = u_j$ for all $j\in[n]$.\\
	Update MILP model $\mathcal{M}$ according to the round function $\boldsymbol{f}_j$ and denote $\boldsymbol{s}_{j+1}$ as the output state after $\boldsymbol{f}_j$ for $0 \le j \le r-1$.\\
	$\mathcal{M}.con \leftarrow \boldsymbol{s}_r[0] = 1$.\\
	Prepare an empty set $\mathtt{VT}_{r_m}$.\\
	$\mathcal{M}.update()$.\\
	$cb=\mathtt{VTCallbackFun}(\& \mathtt{VT}_{r_m}, r_m, \boldsymbol{s}_0, \boldsymbol{s}_1, \cdots, \boldsymbol{s}_r)$.\\
	$\mathcal{M}.setCallback(\&cb)$\\
    $\mathcal{M}.optimize()$.\\
	\KwRet{$\mathtt{VT}_{r_m}$}. 
    \caption{Obtain Valuable Terms}
\end{algorithm}

\begin{algorithm}[H]
	\label{alg.callbackfun}
    \DontPrintSemicolon
	\SetAlgoLined
	\KwIn{$\& \mathtt{VT}_{r_m}, r_m, \boldsymbol{s}_0, \boldsymbol{s}_1, \cdots, \boldsymbol{s}_r$}
		\If{$where = MIPSOL$}
		{Let $\boldsymbol{u}_{j}$ be the solution corresponding to $\boldsymbol{s}_j$.
        }
        flag = \texttt{ComputeNumberOfTrails}($\boldsymbol{u}_{r_m}, r_m)\bmod 2\equiv 1$.\\
		\If{flag} 
		{Add $\pi_{\boldsymbol{u}_{r_m}}(\boldsymbol{y}_{r_m})$ to $\mathtt{VT}_{r_m}$. \\
		}
		addLazy($\boldsymbol{s}_{r_m} \neq \boldsymbol{u}_{r_m}$)\\
  \caption{VTCallbackFun}
  \end{algorithm}
\begin{algorithm}[H]
	\label{alg.computetrails}
    \DontPrintSemicolon
	\SetAlgoLined
	\KwIn{$\boldsymbol{u}_{r_m}, r_m$}
		Declare an empty MILP model $\mathcal{M}'$.\\
		Let $\boldsymbol{s}_{r_m}$ be MILP variables of $\mathcal{M}'$ corresponding to the output states of $\boldsymbol{f}_{r_m-1}$.\\
		$\mathcal{M}'.con \leftarrow \boldsymbol{s}_{r_m}[j] = \boldsymbol{u}_{r_m}[j]$ for all $j$.\\
		Update MILP model $\mathcal{M}'$ according to the round function $\boldsymbol{f}_j$ and denote $\boldsymbol{s}_{j+1}$ as the output state after $\boldsymbol{f}_j$ for $r_m \le j \le r-1$.\\
		$\mathcal{M}'.con \leftarrow \boldsymbol{s}_r[0] = 1$.\\
		$\mathcal{M}'.optimize()$.\\
  \KwRet{The number of solutions of $\mathcal{M}'$}. 
    \caption{ComputeNumberOfTrails}
\end{algorithm}

\section{ Proof of Theorem \ref{vdegthm}}
\label{proof1}
\begin{proof}
	Assume that $g[i] = \bigoplus_{\boldsymbol{w}} g_{\boldsymbol{w}}^i \boldsymbol{x_I}^{\boldsymbol{w}}$. Then we have
	\begin{equation*}
		\boldsymbol{g}^{\boldsymbol{u}} = \bigoplus_{\boldsymbol{w}}\left(\bigoplus_{\substack{\boldsymbol{w}_{0},\cdots,\boldsymbol{w}_{n-1}, \boldsymbol{w}_i \in u[i]\mathbb{F}_2^d\\\boldsymbol{w} = \bigvee_{i=0}^{n-1} u[i]\boldsymbol{w}_{i}}} \prod_{i=0}^{n-1} (g_{\boldsymbol{w}_{i}}^{i})^{u[i]} \right) \boldsymbol{x}_I^{\boldsymbol{w}},
	\end{equation*}
	where $u[i]\mathbb{F}_2^d = \mathbb{F}_2^d$ if $u[i] = 1$ and otherwise, it is equal to $\{ \boldsymbol{0} \}$.
	Then the $(j+1)$-th ($0 \le j \le 2^d -1$) component of the vector degree of $\boldsymbol{g}^{\boldsymbol{u}}$ is less than or equal to 
	$$\max_{\substack{j_0,\cdots,j_{n-1}, 0 \le j_i \le u[i](2^d-1) \\ j = \bigvee_{i=0}^{n-1} u[i]j_i}}\left(\sum_{i=0}^{n-1} u[i]v_{i}[j_i]\right),$$
	since $\deg(g_{\boldsymbol{w}_{i}}^{i}) \le v_{i}[j_i]$, where $\boldsymbol{w}_{i}$ is the $d$-bit binary representation of $j_i$.
	Assume $f = \bigoplus_{\boldsymbol{u}}a_{\boldsymbol{u}}\boldsymbol{x}^{\boldsymbol{u}}$. Then we can conclude that each component of the vector degree of $f \circ \boldsymbol{g}=\bigoplus_{\boldsymbol{u}}a_{\boldsymbol{u}}\boldsymbol{g}^{\boldsymbol{u}}$ is less than or equal to the corresponding 
	component of $\mathtt{VDEG}(f, V)_d$ by Definition \ref{vectormd}.
\end{proof}

\section{ Proof of Theorem \ref{vdegsubset}}
\label{proof2}
\begin{proof}
	Assume that $f = \bigoplus_{\boldsymbol{u}}a_{\boldsymbol{u}}\boldsymbol{x}^{\boldsymbol{u}}$. By  Definition \ref{vectormd}, we have
	\begin{small}
		\begin{equation}\label{eqmax1}
			\begin{aligned}
				&\quad\max_{0 \le j'<2^{d-k}}\left(\mathtt{VDEG}_{d}(f,V_2)[j'\cdot 2^k + j] + wt(j')\right)\\
				&= \max_{0 \le j'<2^{d-k}}\left\{\max_{a_{\boldsymbol{u}} \neq 0}\max_{\substack{j_0' \cdot 2^k+j_0,\ldots,j_{n-1}' \cdot 2^k + j_{n-1} \\ 0 \le j_i \le u[i](2^k-1) \\ 0 \le j_i' \le u[i](2^{d-k}-1) \\ j' \cdot 2^k + j = \bigvee_{i=0}^{n-1} u[i](j_i' \cdot 2^k  + j_i )}}\left\{\sum_{i=0}^{n-1} u[i]V_2[i][j_i' \cdot 2^k + j_i]\right\}+ wt(j')\right\}\\
				&= \max_{0 \le j'<2^{d-k}}\max_{a_{\boldsymbol{u}} \neq 0}\max_{\substack{j_0,\cdots,j_{n-1} \\ 0 \le j_i \le u[i](2^k-1) \\ j = \bigvee_{i=0}^{n-1} u[i]j_i}}\max_{\substack{j'_0,\cdots,j'_{n-1} \\ 0 \le j'_i \le u[i](2^{d-k}-1) \\ j' = \bigvee_{i=0}^{n-1} u[i]j'_i}}\left\{\sum_{i=0}^{n-1} u[i]V_2[i][j_i' \cdot 2^k + j_i]+ wt(j')\right\}\\
				&= \max_{a_{\boldsymbol{u}} \neq 0}\max_{\substack{j_0,\cdots,j_{n-1} \\ 0 \le j_i \le u[i](2^k-1) \\ j = \bigvee_{i=0}^{n-1} u[i]j_i}}\max_{0 \le j'<2^{d-k}}\max_{\substack{j'_0,\cdots,j'_{n-1} \\ 0 \le j'_i \le u[i](2^{d-k}-1) \\ j' = \bigvee_{i=0}^{n-1} u[i]j'_i}}\left\{\sum_{i=0}^{n-1} u[i]V_2[i][j_i' \cdot 2^k + j_i]+ wt(j')\right\}.
			\end{aligned}
		\end{equation}
	\end{small}
	Since $j' = \bigvee_{i=0}^{n-1} u[i]j'_i$,  we know that $\wt(j') \le \sum_{i=0}^{n-1} u[i]wt(j'_i)$. 
	Then by  the inequality  (\ref{vdegI1I2con1}), we have
	\begin{equation}\label{eqmax2}
		\begin{aligned}
			&\quad\max_{0 \le j'<2^{d-k}}\max_{\substack{j'_0,\cdots,j'_{n-1} \\ 0 \le j'_i \le u[i](2^{d-k}-1) \\ j' = \bigvee_{i=0}^{n-1} u[i]j'_i}}\left\{\sum_{i=0}^{n-1} u[i]V_2[i][j_i' \cdot 2^k + j_i]+ wt(j')\right\} \\
			&\le\max_{0 \le j'<2^{d-k}}\max_{\substack{j'_0,\cdots,j'_{n-1} \\ 0 \le j'_i \le u[i](2^{d-k}-1) \\ j' = \bigvee_{i=0}^{n-1} u[i]j'_i}}\left\{\sum_{i=0}^{n-1} u[i]\left(V_2[i][j_i' \cdot 2^k + j_i]+ wt(j_i')\right)\right\}\\
			&\le\max_{0 \le j'<2^{d-k}}\max_{\substack{j'_0,\cdots,j'_{n-1} \\ 0 \le j'_i \le u[i](2^{d-k}-1) \\ j' = \bigvee_{i=0}^{n-1} u[i]j'_i}}\left\{\sum_{i=0}^{n-1} u[i]V_1[i][j_i]\right\}\\
			&=\sum_{i=1}^n u[i]V_1[i][j_i].
		\end{aligned}
	\end{equation}
	By   (\ref{eqmax1}), (\ref{eqmax2}), and Definition \ref{vectormd}, we assert that the inequalities given in Equation (\ref{vdegI1I2conclusion1}) hold for all $0 \le j < 2^k$.
\end{proof}

\section{Example for Estimating Algebraic Degree}
\label{sec.example}
\begin{example}
	\label{exam}
	Let $f = y_0y_1$ and $\boldsymbol{g} = \{x_0x_2 + x_1, x_0x_1 + x_3\}$.
	Then the composite function $$f\circ \boldsymbol{g} = x_0x_1x_2 + x_0x_1 + x_0x_2x_3 + x_1x_3.$$
	\renewcommand\labelenumi{(\arabic{enumi})}
	\begin{enumerate}
		\item Denote $\boldsymbol{d}$ as $\deg(\boldsymbol{g}) = (2,2)$. Then the numeric degree of $f\circ \boldsymbol{g}$ is $\mathtt{DEG}(f,\boldsymbol{d}) = 4 > \deg(f\circ \boldsymbol{g}) = 3$.
		\item For the vector numeric mapping, we consider three cases according as  $I = \{0\}, \{1\}$ or $\{0, 1\}$.
		\begin{enumerate}
			\item Let $I_1 = \{0\}$, and assume  the vector degree of $\boldsymbol{g}$ is $\mathbf{vdeg}_{I_1}(\boldsymbol{g}) = \left( (1,1), (1,1) \right)$, denoted by $V$. 
			The estimated vector degree of $f\circ \boldsymbol{g}$ is $\mathtt{VDEG}_1(f,V)= (2,2) = \textbf{vdeg}_{I_1}(f\circ \boldsymbol{g})$.
			Then the estimated degree of $f\circ \boldsymbol{g}$ can be computed by $\mathtt{VDEG}_1(f,V)$, which is equal to $\max\{0+2, 1+2\} = 3 = \deg(f\circ \boldsymbol{g})$.
			\item Let $I_2 = \{1\}$, and assume the vector degree of $\boldsymbol{g}$ is $\mathbf{vdeg}_{I_2}(\boldsymbol{g}) = \left( (2,0), (1,1) \right)$, denoted by $V$. 
			The estimated vector degree of $f\circ \boldsymbol{g}$ is $\mathtt{VDEG}_1(f,V)= (3,3) \succcurlyeq   \textbf{vdeg}_{I_2}(f\circ \boldsymbol{g}) = (3,2)$.
			Then the estimated degree of $f\circ \boldsymbol{g}$ can be computed by $\mathtt{VDEG}_1(f,V)$, which is equal to $\max\{0+3, 1+3\} = 4 > \deg(f\circ \boldsymbol{g})$.
			\item Let $I_3 = \{0, 1\}$, and assume the vector degree of $\boldsymbol{g}$ is $\mathbf{vdeg}_{I_3}(\boldsymbol{g}) = ( (-\infty,1,0,$ $-\infty ),$ $  (1, -\infty, -\infty, 0) )$, denoted by $V$. 
			The estimated vector degree of $f\circ \boldsymbol{g}$ is $\mathtt{VDEG}_2(f,V)= (-\infty, 2, 1, 1) =  \textbf{vdeg}_{I_3}(f\circ \boldsymbol{g}) $.
			Then the estimated degree of $f\circ \boldsymbol{g}$ can be computed by $\mathtt{VDEG}_2(f,V)$, which is equal to $\max\{0-\infty, 1+2, 1+1, 2+1\} = 3 = \deg(f\circ \boldsymbol{g})$.
			Moreover, the estimated vector degree of $f\circ \boldsymbol{g}$ w.r.t. $I_1$ and $I_2$, respectively, derived from $\mathtt{VDEG}_2(f,V)$, is $\left(\max\{0-\infty, 1+1\}, \max\{0+2, 1+1\} \right) = \deg_{I_1}(f\circ \boldsymbol{g})$   and    
   $\left(\left(\max\{0-\infty, 1+2\}, \max\{0+1, 1+1\} \right) = \deg_{I_2}(f\circ \boldsymbol{g})\right)$, respectively.
		\end{enumerate}
		
	\end{enumerate}
\end{example}

From Example \ref{exam}, we see that the vector numeric mapping for estimating the degree of the composite function is more accurate than the numeric mapping if we choose a suitable $I$. This is because \emph{the vector numeric mapping can  eliminate the repeated degree estimation of variables whose indices are in the index set of vector degree.} In Example \ref{exam}, the degree of $x_0$ would be computed twice if the index set of vector degree does not contain $0$.

\section{Algorithm for Searching Good \textit{ISoC}s}
\label{AlgSearchGoodCube}
\begin{algorithm}[H]
	\label{alg.searching good cubes}
  \DontPrintSemicolon
	\SetAlgoLined
	\KwIn{$J,k,d,a$}
	\KwOut{$\mathcal{I}$}
	$\mathcal{I} \gets \emptyset$\\
	Generate an empty MILP model $\mathcal{M}$\\
	Let $b_0,b_1,\cdots,b_{n-1}$ be the $n$ binary variables of model $\mathcal{M}$\\
	$\mathcal{M}.con \gets \sum_{i=0}^{n-1}b_i=k$\\
	$\mathcal{M}.con \gets b_j=1$, for $\forall j\in J$\\
	$\mathcal{M}.update()$\\
	$cb=\text{CallbackFun}(\& \mathcal{I},d,a,b_0,b_1,\cdots,b_{n-1})$\\
	$\mathcal{M}.setCallback(\&cb)$\\
	$\mathcal{M}.optimize()$\\
	\caption{Search Good \textit{ISoC}s}
\end{algorithm}

  \begin{algorithm}[H]
	  \label{alg.CallbackFun}
	  \DontPrintSemicolon
	  \KwIn{$\& \mathcal{I},d,a,b_0,b_1,\cdots,b_{n-1}$}
	  \If {where=\textit{MIPSOL}}{
		  $I \gets \emptyset$\\
		  \For {each $i$ from $0$ to $n-1$}{
			  \If {$b_i=1$}{
				  $I \gets I \cup \{i\}$
			  }
		  }
		  Get $d_I$, the estimate of algebraic degree w.r.t. the variable $\boldsymbol{v}_I$ and the index set $J$ \\
		  \If {$d_I\leq d$}{
			  $\mathcal{I}\gets \mathcal{I}\cup \{I\}$\\
			  \Else{
				  $d_{I'}=d_{I}$\\
				  \While {$d_{I'}>d$}{
					  \For {each $j$ from $0$ to $a-1$}{
						  $I'=I$ \\
						  Randomly choose $i\in I\setminus J$, $I'=I'\setminus \{i\}$\\
						  Get $d_{I'}$, the estimate of algebraic degree w.r.t. the variable $\boldsymbol{v}_{I'}$ and the index set $J$\\
						  \If {$d_{I'}>d$}{
							  $I=I'$, break
						  }
						  }
			  }
			  }
		  }
		  addLazy($\sum_{i\in I}b_i<|I|$)
	  }
	  \caption{CallbackFun}
  \end{algorithm}

\section{The \textit{ISoC}s for Practical Verification}\label{sec.ISoc}
See Table \ref{tab.ISoC}.
\begin{table}[H]
	\caption{The \textit{ISoC}s}
	\centering  \label{tab.ISoC}
	\begin{tabular}{|c|l|c|} 
		\hline
		No.     &Cube Indices   & Size\\\hline\hline
		\multirow{3}{*}{$I_1$}     & 0, 2, 4, 6, 7, 9, 11, 13, 15, 17, 19, 21, 24, 26, 28, 30, & \multirow{3}{*}{39} \\
		& 32, 34, 36, 39, 41, 43, 45, 47, 49, 51, 54, 56, 58, 60, & \\
		& 62, 64, 66, 69, 71, 73, 75, 77, 79  & \\ \hline
		\multirow{3}{*}{$I_2$}     & 0, 2, 4, 6, 8, 9, 11, 13, 15, 17, 19, 21, 24, 26, 28, 30, & \multirow{3}{*}{39} \\
		& 32, 34, 36, 39, 41, 43, 45, 47, 49, 51, 54, 56, 58, 60, & \\
		& 62, 64, 66, 69, 71, 73, 75, 77, 79  & \\\hline 
		\multirow{3}{*}{$I_3$}     & 0, 1, 2, 4, 6, 8, 9, 11, 13, 15, 17, 19, 21, 24, 26, 28,  & \multirow{3}{*}{40} \\
		& 30, 32, 34, 36, 39, 41, 43, 45, 47, 49, 51, 54, 56, 58,  & \\
		& 60, 62, 64, 66, 69, 71, 73, 75, 77, 79   & \\ \hline
	\end{tabular}
\end{table}

		
		

\section{Found Secret Keys}
\label{sec.foundsecretkeys}
\begin{table}[H]
	\caption{Found secret keys}
	\centering  
	\begin{tabular}{|c|c|c|c|c|} 
		\hline
		No.     & Rounds   & $\boldsymbol{key}$ & Rounds   & $\boldsymbol{key}$\\\hline\hline
		\multirow{3}{*}{$I_1$}     & 836 &  \textsf{0xebd75e597e62736ce784} & 837 &  \textsf{0x43f576b9b75b28e4030c}\\ \cline{2-5}
		& 838 &  \textsf{0xf327d6a7b3bdb3d62c36} & 839 &  \textsf{0xe55eeaa86dc1cd764c83}\\\cline{2-5}
		& 840 &  \textsf{0x75cd618a4e6f7ef37c68}
		& 842 &  \textsf{0x1822b5ad2b15206020d3}\\\hline\hline
		\multirow{4}{*}{$I_2$}     & 836 &  \textsf{0x8c128672e9143c6bdc96}
		& 837 &  \textsf{0xe23551dfcf9d08c4aff4}\\\cline{2-5}
		& 838 &  \textsf{0x4caedab34723fd69c667}
		& 839 &  \textsf{0xb3fb4e2e8f8ec6162f97}\\\cline{2-5}
		& 840 &  \textsf{0x75cd618a4e6f7ef37c68}
		& 841 &  \textsf{0x6bee17ab37a8bf9b8e26}\\\cline{2-5}
		& 842 &  \textsf{0x77a9785a05263c44e8f3} &&\\\hline \hline
		\multirow{2}{*}{$I_3$}     & 837 &  \textsf{0x894e029da347a27baa6}
		& 838 &  \textsf{0x3948a5e54a48eee74d75}\\\cline{2-5}
		& 839 &  \textsf{0xcf48f654983cdd34c923}
		& 842 &  \textsf{0x76766e29cee533d4233e}\\\hline
	\end{tabular}\label{tab.foundkeys}
\end{table}
See Table \ref{tab.foundkeys}, where $\boldsymbol{key} = k_7 \parallel k_6 \parallel \cdots \parallel k_0 \parallel k_{15} \parallel k_{14} \parallel \cdots \parallel k_8 \parallel \cdots \parallel k_{79} \parallel k_{78} \parallel \cdots \parallel k_{72}$.

\section{Running time of the algorithm of degree evaluation}
See Table \ref{tab:CompexityComparisonofDegEva}.
\begin{table}[h]
    \caption{Running time of the algorithm of degree evaluation by vector numeric mapping of $788$-round Trivium with different size of $J$.}
    \label{tab:CompexityComparisonofDegEva}
    \centering
    \footnotesize
    \setlength{\tabcolsep}{4.3pt}
    \renewcommand{\arraystretch}{1.1}
    \begin{tabular}{|c|c|c|c|c|c|c|c|c|c|c|c|} 
		\hline
		$|J|$ &0 &1	&2	&3	&4	&5	&6	&7	&8 & 9 & 10\\\hline
        Time(Sec) &1.63	&1.68	&1.78	&2.04	&2.33	&3.06	&4.58	&7.87	&15.66 & 35.72 & 89.97\\\hline
    \end{tabular}
\end{table}


\section{Algorithm for Estimation of Vector Degree of Trivium}
\label{AlgEstVecDegTri}
The algorithm for estimation of the vector degree of Trivium is detailed in Algorithm \ref{alg.EstimateDegree} and Algorithm \ref{alg.DegreeMul}.

\begin{algorithm}[H]
	\label{alg.EstimateDegree}
    \DontPrintSemicolon
	\SetAlgoLined
	\KwIn{$\boldsymbol{s}^0,I,J,R,mode$}
	$V \gets \mathbf{vdeg}_{[J,\boldsymbol{x}_I]}(\boldsymbol{s}^0)$.\\
	Prepare $V_l$,$V_m$,$V_s$ with size of $288\times 2^{|J|}$.\\
    \For {$t$ from 1 to $R$}{
        \For {$i$ from 1 to 3}{
            $V_l[n_i] \gets \mathtt{VDEG}(l_i,V)$.\\
            $V_m[n_i] \gets \text{DegreeMul}(V,V_l,V_m,V_s,i,t)$.\\
            $V[n_i][j] \gets \max(V_l[n_i][j],V_m[n_i][j])$ for all $j$.\\
            $V_s[n_i]=V[n_i-1]$.\\
        }
        $V \gets (V[287],V[0],\cdots,V[286])$.\\
        $V_l \gets (V_l[287],V_l[0],\cdots,V_l[286])$.\\
        $V_m \gets (V_m[287],V_m[0],\cdots,V_m[286])$.\\
        $V_s \gets (V_s[287],V_s[0],\cdots,V_s[286])$.
    }
    \If {$mode=1$}{
        \KwRet{$\max_j\{\max_{i\in\{65,92,161,176,242,287\}}\{\min\{V[i][j],|I|-|J|\}\}+\wt(j)\}$}\\
    \ElseIf {$mode=3$}{
        \KwRet{$\max_j\{\max_{i\in\{65,92,161,176,242,287\}}\{V[i][j]\}+\wt(j)\}$}
    }
    }
	\caption{Estimation of Vector Degree of Trivium}
\end{algorithm}

\begin{algorithm}[H]
	\label{alg.DegreeMul}
  \DontPrintSemicolon
	\SetAlgoLined
	\KwIn{$V,V_l,V_m,V_s,i,t$}
    $t_1 \gets t-r_i$.\\
    \If {$t_1 < 0$}{
        \KwRet{$\mathtt{VDEGM}(V[n_i-1],V[n_i-2])$}.\\
    }
    $\boldsymbol{v}_1 \gets \mathtt{VDEGM}(V_m[n_i-1],V_s[n_i-3])$.\\
    $\boldsymbol{v}_2 \gets \mathtt{VDEGM}(V_s[n_i-1],V_m[n_i-2])$.\\
    $\boldsymbol{v}_3 \gets \mathtt{VDEGM}(V_s[n_i-1],V_s[n_i-2],V_s[n_i-3])$.\\
    $\boldsymbol{v}_4[j] \gets \min\{\boldsymbol{v}_1[j],\boldsymbol{v}_2[j],\boldsymbol{v}_3[j]\}$ for all $j$.\\
    $\boldsymbol{v}_5 \gets \mathtt{VDEGM}(V_m[n_i-1],V_l[n_i-2])$.\\
    $\boldsymbol{v}_6 \gets \mathtt{VDEGM}(V_l[n_i-1],V[n_i-2])$.\\
    $\boldsymbol{v}[j] \gets \max\{\boldsymbol{v}_4[j],\boldsymbol{v}_5[j],\boldsymbol{v}_6[j]\}$ for all j.\\
    \KwRet{$\boldsymbol{v}$}.
	\caption{DegreeMul}
\end{algorithm}

We denote  $\mathtt{VDEG}(\prod_{i=1}^k x[i],(\boldsymbol{v}_1,\cdots, \boldsymbol{v}_k))$ as $\mathtt{VDEGM}(\boldsymbol{v}_1, \cdots, \boldsymbol{v}_k)$. 
The input of Algorithm \ref{alg.EstimateDegree} are initial internal state $\boldsymbol{s}^{0}$ of Trivium, \textit{ISoC} $I$, index set of vector degree $J \subset I$ and end-round $R$.
And the output of Algorithm \ref{alg.EstimateDegree} is the information about the degree of $R$-round output bit.
The notations $V, V_l, V_m$ and $V_s$ represent array with size of $288 \times 2^{\left\lvert J\right\rvert }.$  
$V[j]$ represents the estimated vector degree of $s_j$, where $\boldsymbol{s}$ is the internal state of Trivium.
$V_l[i], V_m[i]$ and $V_s[i]$ represent the estimated vector degree of the linear component, quadratic term, and one factor of quadratic term in $s_i$ respectively by Equation  (\ref{eq.simupdate}); see Algorithm \ref{alg.EstimateDegree} for details.

In Algorithm \ref{alg.DegreeMul}, if $t_1 \ge 0$, by Equation (\ref{eq.simupdate}) $s_{n_i - 1}s_{n_i - 2}$ can be expanded as
$$(q_1q_2 + l_1 )(q_2q_3 + l_2) = q_1q_2q_3 + q_1q_2l_2 + l_1(q_2q_3+l_2),$$
where $q_1, q_2$ ($q_2, q_3$) are the factors of nonlinear term in $s_{n_i - 1}$ ($s_{n_i - 2}$) and $1_2$ is the common factor, and $l_1$ ($l_2$) is the linear term in $s_{n_i - 1}$ ($s_{n_i - 2}$).
The estimated vector degree of $q_1, q_2, q_3$ are $V_s[n_i-1], V_s[n_i - 2], V_s[n_i - 3]$ respectively.
And $V_m[n_i - 1], V_m[n_i - 2], V_l[n_i - 1], V_l[n_i - 2]$ correspond to the estimation of the vector degree of $q_1q_2, q_2q_3, l_1, l_2$.
To estimate the vector degree of $q_1q_2q_3$, we calculate the minimum value of three values $\boldsymbol{v}_1, \boldsymbol{v}_2$ and $\boldsymbol{v}_3$, 
where $\boldsymbol{v}_1$ is the estimation when view $q_1q_2q_3$ as the multiple of $q_1q_2$ and $q_3$, $\boldsymbol{v}_2$ is the estimation when view $q_1q_2q_3$ as the multiple of $q_1$ and $q_2q_3$,
and $\boldsymbol{v}_3$ is the estimation when view $q_1q_2q_3$ as the multiple of $q_1$, $q_2$ and $q_3$.
Then compute the estimation of the vector degree of $q_1q_2l_1$ and $l_1s_{n_i-2}$ denoted by $\boldsymbol{v}_5, \boldsymbol{v}_6$, and we can obtain the estimation of the vector degree of the nonlinear term.
Every round we compute the estimated vector degree of the linear term, nonlinear term, and the sum of them, and save the values into $V_l,V_m,V,V_s$ in Algorithm \ref{alg.EstimateDegree}.
Finally, return different information about the output bit according to different values of $mode$. 
If $mode = 3$, return the unprocessed estimated degree of the output bit, which is used to search for good cubes.
If $mode = 1$, return the more accurate degree evaluation of the output bit by Corollary \ref{degcor}, which is used to evaluate degree.

\section{The Equations and Probabilities in the $820$-Round Attack}
Refer to Table \ref{tab.T} and Table \ref{tab.T1}, where $h$ is the factor, \# of \textit{ISoC}s denotes the number of \textit{ISoC}s in the set $T_h$, $T_h$ is the set containing all the \textit{ISoC}s the superpoly of which factored by $h$,  $\Pr(0|0)$ represents the probability $\Pr(h=0|f_I = 0\text{ for }\forall I\in T_h)$. 
\label{sec.tablefor820Roud}
\begin{table}[H]

	\caption{Set $\mathcal{T}$ for $820$-round attack}
	\centering  
	\begin{tabular}{|c|c|c|c|c|c|} 
		\hline
		No.  &$h$  & \# of \textit{ISoC}s &  $\Pr(0|0)$ & $\Pr(f_I \neq 0|\exists I\in T_h)$ & \# of Rounds \\\hline\hline
1 & $k_{103} = k_{32}+k_{57}k_{58}+k_{59}$ & 1210 & 0.9996 &0.5005& 820 \\\hline
2 & $k_{104} = k_{31}+k_{56}k_{57}+k_{58}$ & 1432 & 0.9986 &0.5015& 820 \\\hline
3 & $k_{54}$ & 826 & 0.9851 &0.4883& 820 \\\hline
4 & $k_{102} = k_{33}+k_{58}k_{59}+k_{60}$ & 413 & 0.9621 &0.4699& 820 \\\hline
5 & $k_{89} = k_{46}+k_{71}k_{72}+k_{73}$ & 702 & 0.9610 &0.4712& 820 \\\hline
6 & $k_{116} = k_{19}+k_{44}k_{45}+k_{46}$ & 727 & 0.9570 &0.4702& 820 \\\hline
7 & $k_{117} = k_{18}+k_{43}k_{44}+k_{45}$ & 442 & 0.9567 &0.4776& 820 \\\hline
8 & $k_{124} = k_{11}+k_{36}k_{37}+k_{38}$ & 638 & 0.9455 &0.4625& 820 \\\hline
9 & $k_{111} = k_{24}+k_{49}k_{50}+k_{51}$ & 303 & 0.9424 &0.4789& 820 \\\hline
10 & $k_{91} = k_{44}+k_{69}k_{70}+k_{71}$ & 678 & 0.9395 &0.4725& 820 \\\hline
11 & $k_{115} = k_{20}+k_{45}k_{46}+k_{47}$ & 365 & 0.9291 &0.4596& 820 \\\hline
12 & $k_{87} = k_{48}+k_{73}k_{74}+k_{75}$ & 451 & 0.9277 &0.4566& 820 \\\hline
13 & $k_{56}$ & 338 & 0.9209 &0.4511& 820 \\\hline
14 & $k_{110} = k_{25}+k_{50}k_{51}+k_{52}$ & 512 & 0.9151 &0.4537& 820 \\\hline
15 & $k_{109} = k_{26}+k_{51}k_{52}+k_{53}$ & 405 & 0.9102 &0.4552& 820 \\\hline
16 & $k_{59}$ & 171 & 0.9028 &0.4549& 820 \\\hline
17 & $k_{123} = k_{12}+k_{37}k_{38}+k_{39}$ & 374 & 0.8825 &0.4308& 820 \\\hline
18 & $k_{106} = k_{29}+k_{54}k_{55}+k_{56}$ & 666 & 0.8692 &0.4198& 820 \\\hline
19 & $k_{90} = k_{45}+k_{70}k_{71}+k_{72}$ & 462 & 0.8659 &0.4214& 820 \\\hline
20 & $k_{88} = k_{47}+k_{72}k_{73}+k_{74}$ & 374 & 0.8609 &0.4235& 820 \\\hline
21 & $k_{58}$ & 231 & 0.8591 &0.4231& 820 \\\hline
22 & $k_{57}$ & 140 & 0.8584 &0.4181& 820 \\\hline
23 & $k_{63}$ & 976 & 0.8535 &0.4196& 820 \\\hline
24 & $k_{92} = k_{43}+k_{68}k_{69}+k_{70}$ & 452 & 0.8470 &0.4046& 820 \\\hline
25 & $k_{105} = k_{30}+k_{55}k_{56}+k_{57}$ & 416 & 0.8462 &0.4186& 820 \\\hline
26 & $k_{114} = k_{21}+k_{46}k_{47}+k_{48}$ & 102 & 0.8415 &0.4055& 820 \\\hline
27 & $k_{121} = k_{14}+k_{39}k_{40}+k_{41}$ & 170 & 0.8369 &0.4144& 820 \\\hline
28 & $k_{82} = k_{53}+k_{78}k_{79}$ & 290 & 0.7963 &0.3806& 820 \\\hline
29 & $k_{122} = k_{13}+k_{38}k_{39}+k_{40}$ & 189 & 0.7913 &0.3674& 820 \\\hline
30 & $k_{108} = k_{27}+k_{52}k_{53}+k_{54}$ & 124 & 0.7870 &0.3658& 820 \\\hline	
	\end{tabular}
 \label{tab.T}
\end{table}

\begin{table}[H]

	\caption{Set $\mathcal{T}_1$ for $820$-round attack}
	\centering  
	\begin{tabular}{|c|c|c|c|c|c|} 
		\hline
  No.  &$h$  & \#of \textit{ISoC}s &  $\Pr(0|0)$ & $\Pr(f_I \neq 0|\exists I\in T_h)$ & \#of Rounds \\\hline\hline
1 & $k_{55}$ & 99 & 0.76 &0.34& 820 \\\hline
2 & $k_{134} = k_{1}+k_{26}k_{27}+k_{28}$ & 115 & 0.7567 &0.337& 820 \\\hline
3 & $k_{93} = k_{42}+k_{67}k_{68}+k_{69}$ & 218 & 0.7326 &0.3167& 820 \\\hline
4 & $k_{125} = k_{10}+k_{35}k_{36}+k_{37}$ & 130 & 0.72563 &0.3137& 820 \\\hline
5 & $k_{84} = k_{51}+k_{76}k_{77}+k_{78}$ & 62 & 0.7183 &0.2922& 820 \\\hline
6 & $k_{83} = k_{52}+k_{77}k_{78}+k_{79}$ & 50 & 0.7111 &0.2865& 820 \\\hline
7 & $k_{61}$ & 54 & 0.6922 &0.2826& 820 \\\hline
8 & $k_{96} = k_{39}+k_{64}k_{65}+k_{66}$ & 648 & 0.6835 &0.2729& 820 \\\hline
9 & $k_{85} = k_{50}+k_{75}k_{76}+k_{77}$ & 27 & 0.6808 &0.2554& 820 \\\hline
10 & $k_{132} = k_{3}+k_{28}k_{29}+k_{30}$ & 33 & 0.6794 &0.2632& 820 \\\hline
11 & $k_{94} = k_{41}+k_{66}k_{67}+k_{68}$ & 68 & 0.6761 &0.2677& 820 \\\hline
12 & $k_{98} = k_{37}+k_{62}k_{63}+k_{64}$ & 334 & 0.6742 &0.2518& 820 \\\hline
13 & $k_{120} = k_{15}+k_{40}k_{41}+k_{42}$ & 32 & 0.6712 &0.256& 820 \\\hline
14 & $k_{60}$ & 62 & 0.6459 &0.2311& 820 \\\hline
15 & $k_{62}$ & 69 & 0.6457 &0.224& 820 \\\hline
16 & $k_{107} = k_{28}+k_{53}k_{54}+k_{55}$ & 29 & 0.6232 &0.1972& 820 \\\hline
17 & $k_{119} = k_{16}+k_{41}k_{42}+k_{43}$ & 92 & 0.6179 &0.1996& 820 \\\hline
18 & $k_{135} = k_{0}+k_{25}k_{26}+k_{27}$ & 13 & 0.6150 &0.1846& 820 \\\hline
19 & $k_{95} = k_{40}+k_{65}k_{66}+k_{67}$ & 14 & 0.6034 &0.1649& 820 \\\hline
20 & $k_{97} = k_{38}+k_{63}k_{64}+k_{65}$ & 23 & 0.6001 &0.1621& 820 \\\hline
21 & $k_{99} = k_{36}+k_{61}k_{62}+k_{63}$ & 54 & 0.5935 &0.1565& 820 \\\hline
22 & $k_{133} = k_{2}+k_{27}k_{28}+k_{29}$ & 18 & 0.5899 &0.1505& 820 \\\hline
23 & $k_{118} = k_{17}+k_{42}k_{43}+k_{44}$ & 8 & 0.5680 &0.1078& 820 \\\hline
24 & $k_{64}$ & 7 & 0.5533 &0.0946& 820 \\\hline
25 & $k_{131} = k_{4}+k_{29}k_{30}+k_{31}$ & 7 & 0.5506 &0.0898& 820 \\\hline
26 & $k_{65}$ & 5 & 0.5444 &0.0979& 820 \\\hline
27 & $k_{100} = k_{35}+k_{60}k_{61}+k_{62}$ & 14 & 0.5421 &0.0805& 820 \\\hline
28 & $k_{113} = k_{22}+k_{47}k_{48}+k_{49}$ & 7 & 0.5376 &0.0744& 820 \\\hline
29 & $k_{112} = k_{23}+k_{48}k_{49}+k_{50}$ & 4 & 0.5254 &0.0428& 820 \\\hline
30 & $k_{130} = k_{5}+k_{30}k_{31}+k_{32}$ & 2 & 0.5200 &0.0445& 820 \\\hline
31 & $k_{86} = k_{49}+k_{74}k_{75}+k_{76}$ & 2 & 0.5192 &0.0514& 820 \\\hline
 \end{tabular}
 \label{tab.T1}
 \end{table} 

\section{The Equations and Probabilities in the $825$-Round Attack}
\label{sec.tablefor825Roud}
Refer to Table \ref{tab.T_825} and Table \ref{tab.T1_825}, where $h$ is the factor, \# of \textit{ISoC}s denotes the number of \textit{ISoC}s in the set $T_h$, $T_h$ is the set containing all the \textit{ISoC}s the superpoly of which factored by $h$,  $\Pr(0|0)$ represents the probability $\Pr(h=0|f_I = 0\text{ for }\forall I\in T_h)$. 
\begin{table}[H]

	\caption{Set $\mathcal{T}$ for $825$-round attack}
	\centering  
	\begin{tabular}{|c|c|c|c|c|c|} 
		\hline
  No.  &$h$  & \#of \textit{ISoC}s &  $\Pr(0|0)$ & $\Pr(f_I \neq 0|\exists I\in T_h)$ & \#of Rounds \\\hline\hline
1 & $k_{88} = k_{47}+k_{72}k_{73}+k_{74}$ & 16 & 1 &0.5037& 825 \\\hline
2 & $k_{86} = k_{49}+k_{74}k_{75}+k_{76}$ & 726 & 0.9998 &0.5074& 825 \\\hline
3 & $k_{63}$ & 151 & 0.9966 &0.5029& 825 \\\hline
4 & $k_{109} = k_{26}+k_{51}k_{52}+k_{53}$ & 630 & 0.989425 &0.4988& 825 \\\hline
5 & $k_{55}$ & 102 & 0.9742 &0.4851& 825 \\\hline
6 & $k_{84} = k_{51}+k_{76}k_{77}+k_{78}$ & 250 & 0.9726 &0.4773& 825 \\\hline
7 & $k_{110} = k_{25}+k_{50}k_{51}+k_{52}$ & 570 & 0.9544 &0.4762& 825 \\\hline
8 & $k_{85} = k_{50}+k_{75}k_{76}+k_{77}$ & 287 & 0.9539 &0.4686& 825 \\\hline
9 & $k_{95} = k_{40}+k_{65}k_{66}+k_{67}$ & 738 & 0.9536 &0.4716& 825 \\\hline
10 & $k_{108} = k_{27}+k_{52}k_{53}+k_{54}$ & 450 & 0.9530 &0.4763& 825 \\\hline
11 & $k_{127} = k_{8}+k_{33}k_{34}+k_{35}$ & 419 & 0.9363 &0.4691& 825 \\\hline
12 & $k_{93} = k_{42}+k_{67}k_{68}+k_{69}$ & 193 & 0.9308 &0.4622& 825 \\\hline
13 & $k_{128} = k_{7}+k_{32}k_{33}+k_{34}$ & 181 & 0.9257 &0.4642& 825 \\\hline
14 & $k_{65}$ & 288 & 0.9176 &0.4648& 825 \\\hline
15 & $k_{64}$ & 142 & 0.9017 &0.4444& 825 \\\hline
16 & $k_{94} = k_{41}+k_{66}k_{67}+k_{68}$ & 165 & 0.8764 &0.4351& 825 \\\hline
17 & $k_{97} = k_{38}+k_{63}k_{64}+k_{65}$ & 265 & 0.8419 &0.4028& 825 \\\hline
18 & $k_{58}$ & 135 & 0.8365 &0.4075& 825 \\\hline
19 & $k_{92} = k_{43}+k_{68}k_{69}+k_{70}$ & 123 & 0.8362 &0.3969& 825 \\\hline
20 & $k_{129} = k_{6}+k_{31}k_{32}+k_{33}$ & 222 & 0.8207 &0.3939& 825 \\\hline
21 & $k_{56}$ & 91 & 0.8196 &0.3832& 825 \\\hline
22 & $k_{99} = k_{36}+k_{61}k_{62}+k_{63}$ & 183 & 0.8161 &0.3866& 825 \\\hline
23 & $k_{113} = k_{22}+k_{47}k_{48}+k_{49}$ & 85 & 0.8160 &0.3902& 825 \\\hline
24 & $k_{107} = k_{28}+k_{53}k_{54}+k_{55}$ & 112 & 0.8132 &0.3848& 825 \\\hline
25 & $k_{112} = k_{23}+k_{48}k_{49}+k_{50}$ & 108 & 0.8066 &0.3765& 825 \\\hline
26 & $k_{116} = k_{19}+k_{44}k_{45}+k_{46}$ & 72 & 0.7894 &0.3577& 825 \\\hline
27 & $k_{126} = k_{9}+k_{34}k_{35}+k_{36}$ & 45 & 0.7816 &0.3521& 825 \\\hline
28 & $k_{62}$ & 57 & 0.7797 &0.3573& 825 \\\hline
29 & $k_{111} = k_{24}+k_{49}k_{50}+k_{51}$ & 131 & 0.7745 &0.3659& 825 \\\hline
30 & $k_{114} = k_{21}+k_{46}k_{47}+k_{48}$ & 125 & 0.7729 &0.3527& 825 \\\hline
31 & $k_{115} = k_{20}+k_{45}k_{46}+k_{47}$ & 101 & 0.7717 &0.3494& 825 \\\hline

\end{tabular}
\label{tab.T_825}
 \end{table} 

\begin{table}[H]

	\caption{Set $\mathcal{T}_1$ for $825$-round attack}
	\centering  
	\begin{tabular}{|c|c|c|c|c|c|} 
		\hline
  No.  &$h$  & \#of \textit{ISoC}s &  $\Pr(0|0)$ & $\Pr(f_I \neq 0|\exists I\in T_h)$ & \#of Rounds \\\hline\hline
1 & $k_{96} = k_{39}+k_{64}k_{65}+k_{66}$ & 159 & 0.767686 &0.3526& 825 \\\hline
2 & $k_{82} = k_{53}+k_{78}k_{79}$ & 35 & 0.73722 &0.331& 825 \\\hline
3 & $k_{122} = k_{13}+k_{38}k_{39}+k_{40}$ & 93 & 0.736827 &0.3206& 825 \\\hline
4 & $k_{83} = k_{52}+k_{77}k_{78}+k_{79}$ & 61 & 0.71214 &0.2875& 825 \\\hline
5 & $k_{57}$ & 38 & 0.699972 &0.2864& 825 \\\hline
6 & $k_{89} = k_{46}+k_{71}k_{72}+k_{73}$ & 24 & 0.697885 &0.2718& 825 \\\hline
7 & $k_{98} = k_{37}+k_{62}k_{63}+k_{64}$ & 72 & 0.661855 &0.2379& 825 \\\hline
8 & $k_{121} = k_{14}+k_{39}k_{40}+k_{41}$ & 28 & 0.658825 &0.2561& 825 \\\hline
9 & $k_{120} = k_{15}+k_{40}k_{41}+k_{42}$ & 16 & 0.642728 &0.223& 825 \\\hline
10 & $k_{87} = k_{48}+k_{73}k_{74}+k_{75}$ & 29 & 0.641839 &0.2146& 825 \\\hline
11 & $k_{59}$ & 38 & 0.627119 &0.2153& 825 \\\hline
12 & $k_{54}$ & 16 & 0.622961 &0.1908& 825 \\\hline
13 & $k_{117} = k_{18}+k_{43}k_{44}+k_{45}$ & 11 & 0.607733 &0.1776& 825 \\\hline
14 & $k_{130} = k_{5}+k_{30}k_{31}+k_{32}$ & 20 & 0.597595 &0.1685& 825 \\\hline
15 & $k_{125} = k_{10}+k_{35}k_{36}+k_{37}$ & 4 & 0.593281 &0.1606& 825 \\\hline
16 & $k_{103} = k_{32}+k_{57}k_{58}+k_{59}$ & 32 & 0.580514 &0.1399& 825 \\\hline
17 & $k_{60}$ & 15 & 0.574038 &0.1349& 825 \\\hline
18 & $k_{90} = k_{45}+k_{70}k_{71}+k_{72}$ & 8 & 0.573292 &0.1261& 825 \\\hline
19 & $k_{101} = k_{34}+k_{59}k_{60}+k_{61}$ & 7 & 0.567661 &0.1317& 825 \\\hline
20 & $k_{118} = k_{17}+k_{42}k_{43}+k_{44}$ & 5 & 0.558273 &0.0922& 825 \\\hline
21 & $k_{102} = k_{33}+k_{58}k_{59}+k_{60}$ & 10 & 0.554529 &0.0803& 825 \\\hline
22 & $k_{61}$ & 4 & 0.545475 &0.0896& 825 \\\hline
23 & $k_{91} = k_{44}+k_{69}k_{70}+k_{71}$ & 5 & 0.540871 &0.0837& 825 \\\hline
24 & $k_{119} = k_{16}+k_{41}k_{42}+k_{43}$ & 2 & 0.539662 &0.0835& 825 \\\hline
25 & $k_{123} = k_{12}+k_{37}k_{38}+k_{39}$ & 4 & 0.537047 &0.0647& 825 \\\hline
26 & $k_{100} = k_{35}+k_{60}k_{61}+k_{62}$ & 4 & 0.526177 &0.0526& 825 \\\hline
27 & $k_{106} = k_{29}+k_{54}k_{55}+k_{56}$ & 1 & 0.520648 &0.0314& 825 \\\hline
28 & $k_{124} = k_{11}+k_{36}k_{37}+k_{38}$ & 1 & 0.520483 &0.0236& 825 \\\hline
29 & $k_{135} = k_{0}+k_{25}k_{26}+k_{27}$ & 2 & 0.515787 &0.0277& 825 \\\hline
30 & $k_{104} = k_{31}+k_{56}k_{57}+k_{58}$ & 1 & 0.505586 &0.0154& 825 \\\hline

\end{tabular}
\label{tab.T1_825}
 \end{table} 

 \section{The Equations and Probabilities in the $830$-Round Attack}
\label{sec.tablefor830Roud}
Refer to Table \ref{tab.T_830} and Table \ref{tab.T1_830}, where $h$ is the factor, \# of \textit{ISoC}s denotes the number of \textit{ISoC}s in the set $T_h$, $T_h$ is the set containing all the \textit{ISoC}s the superpoly of which factored by $h$,  $\Pr(0|0)$ represents the probability $\Pr(h=0|f_I = 0\text{ for }\forall I\in T_h)$. 
\begin{table}[H]
	\caption{Set $\mathcal{T}$ for $830$-round attack}
	\centering  
	\begin{tabular}{|c|c|c|c|c|c|} 
		\hline
  No.  &$h$  & \#of \textit{ISoC}s &  $\Pr(0|0)$ & $\Pr(f_I \neq 0|\exists I\in T_h)$ & \#of Rounds \\\hline\hline
  1 & $k_{114} = k_{21}+k_{46}k_{47}+k_{48}$ & 2058 & 0.9996 &0.4995& 830 \\\hline
2 & $k_{61}$ & 826 & 0.984731 &0.4957& 830 \\\hline
3 & $k_{101} = k_{34}+k_{59}k_{60}+k_{61}$ & 1280 & 0.958204 &0.4856& 830 \\\hline
4 & $k_{59}$ & 493 & 0.947987 &0.4809& 830 \\\hline
5 & $k_{115} = k_{20}+k_{45}k_{46}+k_{47}$ & 1678 & 0.93623 &0.4637& 830 \\\hline
6 & $k_{57}$ & 477 & 0.932773 &0.4645& 830 \\\hline
7 & $k_{100} = k_{35}+k_{60}k_{61}+k_{62}$ & 1193 & 0.910669 &0.4526& 830 \\\hline
8 & $k_{58}$ & 244 & 0.880441 &0.4371& 830 \\\hline
9 & $k_{120} = k_{15}+k_{40}k_{41}+k_{42}$ & 1102 & 0.88031 &0.4327& 830 \\\hline
10 & $k_{99} = k_{36}+k_{61}k_{62}+k_{63}$ & 386 & 0.868343 &0.4235& 830 \\\hline
11 & $k_{89} = k_{46}+k_{71}k_{72}+k_{73}$ & 885 & 0.86694 &0.4138& 830 \\\hline
12 & $k_{88} = k_{47}+k_{72}k_{73}+k_{74}$ & 912 & 0.859543 &0.4226& 830 \\\hline
13 & $k_{132} = k_{3}+k_{28}k_{29}+k_{30}$ & 1117 & 0.838526 &0.403& 830 \\\hline
14 & $k_{98} = k_{37}+k_{62}k_{63}+k_{64}$ & 296 & 0.825532 &0.389& 830 \\\hline
15 & $k_{118} = k-{17}+k_{42}k_{43}+k_{44}$ & 318 & 0.822193 &0.3836& 830 \\\hline
16 & $k_{90} = k_{45}+k_{70}k_{71}+k_{72}$ & 327 & 0.818627 &0.388& 830 \\\hline
17 & $k_{84} = k_{51}+k_{76}k_{77}+k_{78}$ & 109 & 0.816969 &0.3777& 830 \\\hline
18 & $k_{133} = k_{2}+k_{27}k_{28}+k_{29}$ & 496 & 0.81071 &0.3819& 830 \\\hline
19 & $k_{85} = k_{50}+k_{75}k_{76}+k_{77}$ & 179 & 0.80079 &0.367& 830 \\\hline
20 & $k_{60}$ & 468 & 0.798007 &0.3777& 830 \\\hline
21 & $k_{102} = k_{33}+k_{58}k_{59}+k_{60}$ & 757 & 0.786552 &0.3516& 830 \\\hline
22 & $k_{62}$ & 258 & 0.785423 &0.362& 830 \\\hline
23 & $k_{131} = k_{4}+k_{29}k_{30}+k_{31}$ & 254 & 0.783615 &0.3604& 830 \\\hline
24 & $k_{119} = k_{16}+k_{41}k_{42}+k_{43}$ & 239 & 0.781359 &0.367& 830 \\\hline
25 & $k_{116} = k_{19}+k_{44}k_{45}+k_{46}$ & 250 & 0.776892 &0.3474& 830 \\\hline
 \end{tabular}
\label{tab.T_830}
 \end{table} 
 \begin{table}[H]

	\caption{Set $\mathcal{T}_1$ for $830$-round attack}
	\centering  
	\begin{tabular}{|c|c|c|c|c|c|} 
		\hline
  No.  &$h$  & \#of \textit{ISoC}s &  $\Pr(0|0)$ & $\Pr(f_I \neq 0|\exists I\in T_h)$ & \#of Rounds \\\hline\hline
  1 & $ k_{113} = k_{22}+k_{47}k_{48}+k_{49} $ & 176 & 0.768494 &0.3525& 830 \\\hline
2 & $ k_{91}= k_{44}+k_{69}k_{70}+k_{71} $ & 132 & 0.764579 &0.3518& 830 \\\hline
3 & $ k_{87} = k_{48}+k_{73}k_{74}+k_{75} $ & 375 & 0.760676 &0.3373& 830 \\\hline
4 & $ k_{121} = k_{14}+k_{39}k_{40}+k_{41} $ & 273 & 0.760081 &0.3552& 830 \\\hline
5 & $ k_{56} $ & 87 & 0.743492 &0.3201& 830 \\\hline
6 & $ k_{103} = k_{32}+k_{57}k_{58}+k_{59} $  & 517 & 0.737083 &0.3226& 830 \\\hline
7 & $ k_{86} = k_{49}+k_{74}k_{75}+k_{76} $ & 70 & 0.686507 &0.2826& 830 \\\hline
8 & $ k_{127} = k_{8}+k_{33}k_{34}+k_{35} $ & 131 & 0.671031 &0.2592& 830 \\\hline
9 & $ k_{97} = k_{38}+k_{63}k_{64}+k_{65} $ & 149 & 0.668884 &0.2483& 830 \\\hline
10 & $ k_{63} $ & 265 & 0.666039 &0.2562& 830 \\\hline
11 & $ k_{82} = k_{53}+k_{78}k_{79} $  & 181 & 0.665407 &0.2588& 830 \\\hline
12 & $ k_{117} = k_{18}+k_{43}k_{44}+k_{45} $ & 170 & 0.631779 &0.2089& 830 \\\hline
13 & $ k_{93} = k_{42}+k_{67}k_{68}+k_{69} $ & 20 & 0.630082 &0.2055& 830 \\\hline
14 & $ k_{126} = k_{9}+k_{34}k_{35}+k_{36} $ & 119 & 0.626578 &0.1918& 830 \\\hline
15 & $ k_{83} = k_{52}+k_{77}k_{78}+k_{79} $ & 50 & 0.61258 &0.1717& 830 \\\hline
16 & $ k_{108} = k_{27}+k_{52}k_{53}+k_{54} $ & 66 & 0.596439 &0.1632& 830 \\\hline
17 & $ k_{54} $ & 146 & 0.586709 &0.1408& 830 \\\hline
18 & $ k_{129} = k+{6}+k_{31}k_{32}+k_{33} $ & 62 & 0.582027 &0.1454& 830 \\\hline
19 & $ k_{122} = k_{13}+k_{38}k_{39}+k_{40} $ & 25 & 0.578461 &0.1346& 830 \\\hline
20 & $ k_{107} = k_{28}+k_{53}k_{54}+k_{55} $ & 87 & 0.574397 &0.129& 830 \\\hline
21 & $ k_{134} = k_{1}+k_{26}k_{27}+k_{28} $ & 44 & 0.574093 &0.1261& 830 \\\hline
22 & $ k_{92} = k_{43}+k_{68}k_{69}+k_{70} $ & 18 & 0.573785 &0.1211& 830 \\\hline
23 & $ k_{55} $ & 22 & 0.573388 &0.1252& 830 \\\hline
24 & $ k_{130} = k+{5}+k_{30}k_{31}+k_{32} $ & 111 & 0.564082 &0.1191& 830 \\\hline
25 & $ k_{106} = k_{29}+k_{54}k_{55}+k_{56} $ & 42 & 0.559277 &0.0983& 830 \\\hline
26 & $ k_{68} $ & 21 & 0.552496 &0.0885& 830 \\\hline
27 & $ k_{104} = k_{31}+k_{56}k_{57}+k_{58} $ & 17 & 0.541853 &0.0813& 830 \\\hline
28 & $ k_{64} $ & 33 & 0.541563 &0.0749& 830 \\\hline
29 & $ k_{109} = k_{26}+k_{51}k_{52}+k_{53} $ & 38  & 0.528171 &0.0611& 830 \\\hline
30 & $ k_{125} = k_{10}+k_{35}k_{36}+k_{37} $ & 9 & 0.522999 &0.0478& 830 \\\hline
31 & $ k_{110} = k_{25}+k_{50}k_{51}+k_{52} $ & 9 & 0.522362 &0.043& 830 \\\hline
32 & $ k_{112} = k_{23}+k_{48}k_{49}+k_{50} $ & 3 & 0.518507 &0.0301& 830 \\\hline
33 & $ k_{135} = k+{0}+k_{25}k_{26}+k_{27} $ & 3 & 0.512572 &0.0216& 830 \\\hline
34 & $ k_{124} = k_{11}+k_{36}k_{37}+k_{38} $ & 2 & 0.512195 &0.0078& 830 \\\hline
35 & $ k_{95} = k_{40}+k_{65}k_{66}+k_{67} $ & 2 & 0.510278 &0.0125& 830 \\\hline
36 & $ k_{128} = k_{7}+k_{32}k_{33}+k_{34} $ & 5 & 0.509816 &0.0271& 830 \\\hline
37 & $ k_{123} = k_{12}+k_{37}k_{38}+k_{39} $ & 2 & 0.506351 &0.008& 830 \\\hline
38 & $ k_{94} = k_{41}+k_{66}k_{67}+k_{68} $ & 2 & 0.501012 &0.0118& 830 \\\hline
39 & $ k_{96} = k_{39}+k_{64}k_{65}+k_{66} $ & 1 & 0.49995 &0.0059& 830 \\\hline
40 & $ k_{105} = k_{30}+k_{55}k_{56}+k_{57} $ & 5 & 0.499898 &0.0158& 830 \\\hline
41 & $ k_{65} $ & 1 & 0.49772 &0.0133& 830 \\\hline
 \end{tabular}
\label{tab.T1_830}
 \end{table}

\section{Hypothesis Testing for Success Probability}
\label{sec.hypothesis}
Hypothesis testing is a statistical technique used to make decisions or draw conclusions about a population based on sample data. We have treated the proportion denoted as $P$ in the Table \ref{tab:res820}, \ref{tab:res825}, \ref{tab:res830} as the expected success probability. Consequently, we've formulated H0, the null hypothesis, positing that the success rate equals $P$. Our chosen significance level, alpha, is set at 0.01. The Table \ref{tab:hypo820}, \ref{tab:hypo825}, \ref{tab:hypo830} presents the outcomes derived from randomly sampling 10,000 keys and recording the instances of success. Utilizing Python, we conducted a binomial distribution hypothesis test to scrutinize whether the frequency of successful key recoveries significantly deviates from the anticipated success probability, P. The results indicate that we do not have grounds to reject the null hypothesis H0, i.e., $P$ is highly related with the success probability.
\begin{table}[h]
    \caption{The number of successful key recoveries for 820 rounds.}
    \label{tab:hypo820}
    \centering
    \footnotesize
    \setlength{\tabcolsep}{6.8pt}
    \renewcommand{\arraystretch}{1.1}
    \begin{tabular}{|c|c|c|c|c|c|} 
		\hline
		$\mathcal{C}$ &  $2^{52}$	&$2^{54}$ &$2^{56}$ &$2^{58}$ & $2^{60}$\\\hline
         \# &  5732 & 6877 & 7638 & 8255	& 8698	\\\hline
    \end{tabular}
\end{table}

\begin{table}[h]
    \caption{The number of successful key recoveries for 825 rounds.}
    \label{tab:hypo825}
    \centering
    \footnotesize
    \setlength{\tabcolsep}{6.8pt}
    \renewcommand{\arraystretch}{1.1}
    \begin{tabular}{|c|c|c|c|c|} 
		\hline
		$\mathcal{C}$ &  $2^{54}$	&$2^{56}$  &$2^{58}$ & $2^{60}$\\\hline
         \# &  6022  & 7066 & 7763	&8313	\\\hline
    \end{tabular}
\end{table}

\begin{table}[h]
    \caption{The number of successful key recoveries for 830 rounds.}
    \label{tab:hypo830}
    \centering
    \footnotesize
    \setlength{\tabcolsep}{6.8pt}
    \renewcommand{\arraystretch}{1.1}
    \begin{tabular}{|c|c|c|c|c|c|c|} 
		\hline
		$\mathcal{C}$ &  $2^{55}$	&$2^{56}$ &$2^{57}$ &$2^{58}$ &$2^{59}$	& $2^{60}$\\\hline
         \# &  4601 & 5021  &5447	&5824 	&6189	&6575\\\hline
    \end{tabular}
\end{table}

\end{document}